%% file: paper_sicon.tex
\let\SIAMOriginalLabel\label
\let\SIAMOriginalRefstepcounter\refstepcounter
\AddToHook{package/hyperref/before}{%
  \let\label\SIAMOriginalLabel
  \let\refstepcounter\SIAMOriginalRefstepcounter
}

\documentclass[final]{siamart190516}

\usepackage{amssymb}
\usepackage{mathtools}
\usepackage{booktabs}
\usepackage{graphicx}
\usepackage{float}
\usepackage{pgfplots}
\pgfplotsset{compat=1.16, table/search path={figures_py/}}
\usepackage{xspace}
\usepackage[normalem]{ulem}

\newsiamremark{remark}{Remark}
\newsiamremark{assumption}{Assumption}
\newsiamremark{example}{Example}

\newcommand{\X}{\mathsf{X}}
\newcommand{\Ux}{\mathsf{U}}
\newcommand{\E}{\mathbb{E}}
\newcommand{\Prob}{\mathbb{P}}
\newcommand{\R}{\mathbb{R}}

\newcommand{\Tcal}{\mathcal{T}}
\newcommand{\PSi}{\Psi}
\newcommand{\Phix}{\Phi}
\newcommand{\bbeta}{\beta}
\newcommand{\norminfty}[1]{\lVert #1\rVert_\infty}
\newcommand{\calH}{\mathcal{H}}

\newcommand{\calC}{\mathcal{C}}
\newcommand{\calJ}{\mathcal{J}}
\newcommand{\calM}{\mathcal{M}}

\newcommand{\norm}[1]{\lVert #1\rVert}
\newcommand{\veps}{\varepsilon}
\newcommand{\one}{\mathbf{1}}

\DeclareMathOperator*{\argmin}{arg\,min}

\DeclareMathOperator{\dist}{dist}

\graphicspath{{figures_py/}{tables_py/}}

\begin{document}

\title{Poisson Tangent Limits and Critical Policy Switching for Sampled
Bellman Operators}

\author{Ming-Zhe Dai\thanks{School of Automation, Central South University,
Changsha 410083, China (\email{mingzhe\_dai@csu.edu.cn}).}
\and
Chengxi Zhang\thanks{School of IoT Engineering, Jiangnan University, Wuxi,
Jiangsu 214000, China (\email{dongfangxy@163.com}). Corresponding author.}}

\maketitle

\begin{abstract}
Consider a discounted Markov decision process with continuous action
space in which, at each state visit, the controller draws a random pool
of $N$ candidate actions and selects among them.  When the optimal
action set has zero mass under the sampling distribution, the value of
this random-candidate model converges to the optimal value as $N$ grows,
but the rate of convergence and the asymptotic selection rule are
governed jointly by the geometry of the optimal set and by the
transition kernels of the near-optimal candidates.  This paper develops
an exact first-order theory of both.  The rescaled near-optimal
candidates converge to a marked Poisson point process, and the leading
term of the value gap is the fixed point of a nonlinear tangential
Bellman operator, a stochastic generalization of the classical resolvent
that emerges when several optimal actions compete.  The asymptotic
selection at exact ties is genuinely dynamic, driven by the transition
kernels through the fixed point, and admits an explicit Mecke integral formula at the limiting
Poisson level; perturbing the tie at the critical rate yields a
switching layer in which the value gap and the selection interpolate
continuously between the branch regimes, and in the unique-optimum
heterogeneous case the critical class is propagated at the slowest
global scale through the discounted reachability of the limit policy.  Numerical
experiments confirm the predicted rates, constants, and selection
probabilities.
\end{abstract}



\section{Introduction}
\label{sec:intro}

Consider a discounted Markov decision process (MDP) with a continuous
action space, and suppose that at each visit to a state the controller does
not optimize over the full action space but rather over a random set of
$N$ candidate actions drawn i.i.d.\ from a state-dependent sampling
distribution.  This ``random candidate action'' interface is a stylized
abstraction of the rank-and-select step of several sampling-based
schemes: a cheap upstream mechanism (a generative model, a black-box
simulator, a motion planner, or a screening score) produces a finite
menu of actions that a downstream evaluator then ranks and selects.
Motivating analogues include sampling-based stochastic optimal control
\cite{Theodorou10, Williams18}, model-based random search (the
cross-entropy method \cite{deBoer05}, model-reference adaptive search
\cite{HuFuMarcus07}), and sampling-based motion planning
\cite{KaramanFrazzoli11}.  Although these methods differ materially in
pool generation, adaptation, and evaluation, the present model isolates
their common per-decision rank-and-select interface.

The object of this paper is the \emph{random-candidate Bellman model}: the
controller observes the entire candidate set before choosing, so the
one-step value of a state is the expectation of the maximum over the $N$
offered actions of the $Q$-function.  Two questions are natural and
coupled: how fast does the discounted value $V_N$ converge to
$V^\star$ as $N\to\infty$, and which near-optimal candidate does the
finite-pool controller asymptotically select?  Both are answered
exactly in terms of the \emph{local geometry
of the optimal action set}: the value gap $V^\star - V_N$ decays at a
polynomial rate set by the flatness of the Bellman deficiency and the
effective codimension of the optimal set under the sampling measure,
and tie-breaking among several optimal actions is a genuinely dynamic
phenomenon governed by the transition kernels through a nonlinear
stochastic Bellman equation.

The classical contraction theory of discounted dynamic programming
\cite{Blackwell65,Denardo67,Puterman94}, its convergence refinements
\cite{SantosRust04,Munos07}, and the stochastic approximation machinery
underlying $Q$-learning-type methods
\cite{Jaakkola94,WatkinsDayan92,NeufeldSester24} provide the framework
for the analysis here.  For continuous
action spaces, adaptive discretization has been analyzed in the online
reinforcement learning literature \cite{Sinclair23,CaoKrishnamurthy20},
whose near-optimality dimension (introduced for X-armed bandits
\cite{BubeckMunosStoltzSzepesvari11}) plays the role of the critical
exponent here, along with kernel-based approximation \cite{OrmoneitSen02}
and sensitivity analysis of the optimal value \cite{Kern20}.  The
random-candidate interface
differs: the ``grid'' is random and regenerated at every state visit;
the approximation error is governed by extreme value theory rather than
by mesh width; and, crucially for policy selection, the limiting
error and the limiting policy are coupled through the transition
kernels of the near-optimal candidates.  The closest antecedents are
the stochastic-action-set MDP of Boutilier et al.\ \cite{Boutilier18},
of which Section~2 is the continuous-action specialization, and the
EMaQ operator \cite{Ghasemipour21}, with random action sets also
studied game-theoretically \cite{Flesch24}; stochastic Q-learning over
random subsets of a large \emph{discrete} action space is analyzed in
\cite{Fourati24}, whose convergence results concern the finite action
set, whereas the present paper treats continuous action spaces,
optimal sets of zero sampling mass, exact Poisson asymptotics, and the
resulting selection law.  The volume and
sample-maximum laws are classical
\cite{ZhigljavskyZilinskas,DavidNagaraja03,deHaanFerreira06,Resnick87};
the novelty of the marked Poisson limit of
Proposition~\ref{prop:poisson} lies in the marks (the transition
kernels), not in the point-process convergence itself.  The closest
algorithmic precedents, Sampled MuZero \cite{Hubert21} and
doubly-asynchronous value iteration \cite{TianYoungSutton22}, establish
almost-sure convergence or non-asymptotic bounds for the
\emph{sampled solution itself}, whereas the present analysis treats
the value of the random-candidate interface for a \emph{fixed} model
at the exact rate and with the asymptotic selection rule; limit
theorems for the estimated optimum of a noise-corrupted function under
non-adaptive random search \cite{ChiaGlynn13} quantify an
exploration-estimation trade-off between sampled points and
simulations per point, absent here ($Q$ is evaluated exactly), and
concern a single function's sampled optimum rather than the
random-candidate Bellman value gap and selection rule.

What none of these analyses addresses is the \emph{asymptotic regime} of
this paper: for a finite (hence positive-mass) candidate space the value
gap is governed by the probability of missing a single optimal action and
vanishes exponentially.  The complementary regime considered here is that
of continuous
action spaces, in which the optimal action set has zero mass under the
sampling measure:
the sampled maximum still converges to $V^\star$ almost surely (the pool
eventually contains actions arbitrarily close to the optimal set), but
the value gap vanishes only at the rate governed by the extreme values of
the sampled deficiency \emph{toward} the optimal set.  In that regime the
rescaled near-optimal candidates converge to a marked Poisson point
process carrying the transition kernels of the candidate actions, and the
leading constant of the value gap is the fixed point of a nonlinear
tangential Bellman operator rather than a product of a tail constant and
a resolvent factor; the resolvent appears only in the single-optimum
case, where the nonlinear Poisson envelope collapses to an affine
resolvent.  The leading term is \emph{not} the Hadamard semidifferential of the
Bellman map $\Tcal$ at $V^\star$ \cite{AkianGN16}: that map
$h\mapsto\bigl(\bbeta\max_{u\in\mathsf{A}_x^\star}P_{x,u}h\bigr)_{x\in\X}$,
nonlinear exactly when several optimal actions with different kernels
coexist, has the unique fixed point $h\equiv 0$ and carries no
information about the leading constant of the value gap.  The
tangential operator is instead the limit of the rescaled deficiency
operators of the family $\{\Tcal_N\}_{N\ge1}$: an affine resolvent with
an extreme-value constant in the single-optimum case, a genuinely
nonlinear stochastic envelope in the multiple-optimum case.

\paragraph{Contributions}
The main results of the paper are the following.
\begin{enumerate}
\item[(i)] \emph{Exact asymptotics of the sampled Bellman fixed point.}
Under a geometric regularity condition on the optimal action set
(finite unions of separated smooth strata with homogeneous normal
deficiency and power-law sampling density), the rescaled near-optimal
candidates converge to a marked Poisson point process
whose marks are the transition kernels, and the rescaled value gap
converges in sup norm to the unique fixed point of a nonlinear
tangential operator: the classical resolvent with an explicit
extreme-value constant when the optimal action is unique, a genuinely
nonlinear Poisson envelope when several optimal actions coexist.
\item[(ii)] \emph{Dynamic geometric tie-breaking and the critical
switching layer.}  For isolated competing optimal actions: when
several optimal actions of a state tie exactly,
the empirical optimal candidate converges to a random optimal action
whose law is governed jointly by the local volume constants and the
future-cost differences propagated through the transition kernels (a
symmetric two-state example selects the dynamic $0.75$ against the
static $1/2$); for a general compact optimal set, the limiting Poisson
selector admits the Mecke representation \eqref{eq:mecke}.  Exact
ties, nongeneric under generic perturbations,
arise structurally under symmetries of the model data; perturbing the
tie at the critical rate produces a switching layer in which value gap
and selection rule interpolate continuously between the two branch
regimes.
\item[(iii)] \emph{Propagation of the slowest critical scale under
heterogeneous exponents.}  When every state has a unique optimal
action and the local exponents differ across
states, the value gap is governed at the slowest global scale by the
critical class of states with the largest exponents, propagated
through the discounted reachability of the limit policy; finer layers
require additional transition-order assumptions.
\end{enumerate}

\paragraph{Organization and notation}
Section~2 introduces the random-candidate model and its contraction
properties; Section~3 develops the geometry of near-optimal actions;
Sections~4--5 prove the marked Poisson limit theorems and state the
first main theorem; Section~6 proves dynamic geometric tie-breaking
for exactly tied optimal actions under Assumption~\ref{ass:isolated}
(Theorem~\ref{thm:tiebreak}), derives the selection law as an explicit
Mecke integral representation on general compact optimal sets, and
perturbs the tie at the critical rate into a critical switching layer
(Theorem~\ref{thm:switching}); Section~7 treats
heterogeneous exponents; Section~8 the numerical experiments; Section~9
concludes.  Technical details are collected in the Supplement.

\emph{Notation.}  Throughout, $\X = \{1,\dots,n\}$ is a finite state
space.  For each $x\in\X$, the action space $\Ux_x\subset\R^{d_x}$ is
nonempty and compact.  The discount factor is $\bbeta\in(0,1)$.  Norms are
sup norms: $\norminfty{V} = \max_{x\in\X}|V(x)|$ for $V:\X\to\R$.  The
candidate sampling measures are $\mu_x$, $x\in\X$.  The value functions of
the continuous-action and random-candidate models are $V^\star$ and $V_N$,
respectively.  The Bellman deficiency is
$\Delta_x(u) = V^\star(x) - Q_{V^\star}(x,u)$, and the optimal action set
is $\mathsf{A}_x^\star = \{u: \Delta_x(u) = 0\}$, where $Q_V$ is the
Bellman operator of the continuous-action model (Section~2).  Write
$\calH^k$ for the $k$-dimensional Hausdorff measure and $B(x,r)$ for the
open Euclidean ball of radius $r$ around $x$.

\section{Random-Candidate Bellman Model}
\label{sec:model}

\subsection{The Continuous-Action MDP}
\label{sec:model:mdp}

The underlying model is a discounted Markov decision process with
finite state space $\X$ (Section~1).  For each state $x\in\X$, the
action space $\Ux_x\subset\R^{d_x}$ is nonempty and compact, equipped
with its Borel $\sigma$-algebra.  The immediate reward $r_x:\Ux_x\to\R$ is
continuous and bounded, and the transition probabilities
$p_{xy}:\Ux_x\to[0,1]$, $y\in\X$, are continuous with
$\sum_{y\in\X}p_{xy}(u)=1$ for all $u\in\Ux_x$.  The discount factor is
$\bbeta\in(0,1)$.

For a value function $V:\X\to\R$, define the $Q$-function
\begin{equation}
Q_V(x,u) \;=\; r_x(u) + \bbeta \sum_{y\in\X} p_{xy}(u)\, V(y),
\qquad x\in\X,\; u\in\Ux_x,
\label{eq:Q}
\end{equation}
and the Bellman operator
\begin{equation}
(\Tcal V)(x) \;=\; \max_{u\in\Ux_x} Q_V(x,u), \qquad x\in\X.
\label{eq:Bellman}
\end{equation}
By compactness of $\Ux_x$ and continuity of $r_x$ and $p_{xy}$, the
maximum in \eqref{eq:Bellman} is attained for every $x$, and
$\Tcal$ maps $\R^n$ into itself.  Since
\begin{equation}
|Q_V(x,u) - Q_W(x,u)| \;\le\; \bbeta\norminfty{V-W}
\qquad \forall\,u\in\Ux_x,
\end{equation}
the operator $\Tcal$ is a $\bbeta$-contraction in the sup norm
\cite{Denardo67,Blackwell65,Puterman94}, and by Banach's fixed point
theorem it has a unique fixed point
\begin{equation}
V^\star \;=\; \Tcal V^\star,
\label{eq:fixedstar}
\end{equation}
which is the infinite-horizon discounted value of the
continuous-action model.

\subsection{The Random-Candidate Interface}
\label{sec:model:candidates}

At each visit to state $x$, an upstream interface generates
\begin{equation}
U_{x,1},\dots,U_{x,N} \;\overset{\mathrm{i.i.d.}}{\sim}\; \mu_x,
\label{eq:candidates}
\end{equation}
where $\mu_x$ is a Borel probability measure on $\Ux_x$ with full support
in the sense that $\mu_x(B(u,\varepsilon))>0$ for every $u\in\Ux_x$ and
every $\varepsilon>0$.  The controller observes the entire candidate set
before selecting one action, and the candidates are regenerated
independently at each visit.  The decision problem is then a Markov
decision process whose value function satisfies the \emph{expected sampled
Bellman equation}
\begin{equation}
(\Tcal_N V)(x) \;=\; \E\left[\, \max_{1\le i\le N} Q_V(x, U_{x,i}) \right],
\qquad x\in\X,
\label{eq:BellmanN}
\end{equation}
where the expectation is with respect to the i.i.d.\ candidates
\eqref{eq:candidates}.  Let
\begin{equation}
V_N \;=\; \Tcal_N V_N
\label{eq:fixedN}
\end{equation}
be the unique fixed point of $\Tcal_N$, which exists by the same
contraction argument (Proposition~\ref{prop:contraction} below).  The
random-candidate model \eqref{eq:BellmanN} corresponds to the protocol
regenerating the pool at every stage; a pool fixed over the system
lifetime leads to a random Bellman operator whose distributional limit
can be developed within the point-process framework of Section~4, but
is not pursued here.

\begin{proposition}[Contraction, comparison, and uniform convergence]
\label{prop:contraction}
For every $N\ge 1$, the following hold.
\begin{enumerate}
\item[(a)] $\Tcal$ and $\Tcal_N$ are $\bbeta$-contractions on
$(\R^n,\norminfty{\cdot})$; in particular, $V^\star$ and $V_N$ are the
unique fixed points of $\Tcal$ and $\Tcal_N$, respectively.
\item[(b)] $\Tcal_N V \le \Tcal_{N+1} V \le \Tcal V$ componentwise for
every $V\in\R^n$, and $V_N \le V_{N+1} \le V^\star$ componentwise.
\item[(c)] If each $\mu_x$ has full support, then
$\lim_{N\to\infty}\norminfty{V_N - V^\star} = 0$.
\end{enumerate}
\end{proposition}

\begin{proof}
(a) For $V,W\in\R^n$ and $u\in\Ux_x$, the bound
$|Q_V(x,u)-Q_W(x,u)|\le\bbeta\norminfty{V-W}$ gives, for the pointwise
maxima,
\begin{equation}
\Bigl| \max_i Q_V(x,U_{x,i}) - \max_i Q_W(x,U_{x,i}) \Bigr|
\;\le\; \bbeta\norminfty{V-W}.
\end{equation}
Taking expectations yields
$\norminfty{\Tcal_N V - \Tcal_N W}\le\bbeta\norminfty{V-W}$,
and the same argument applies to $\Tcal$.  Banach's fixed point theorem
gives uniqueness of the fixed points.

(b) The maximum over $N+1$ candidates exceeds the maximum over the
first $N$, so $\Tcal_N V \le \Tcal_{N+1} V$ pointwise, and
$\{U_{x,i}\}_{i=1}^N \subseteq \Ux_x$ gives
$\Tcal_{N+1} V \le \Tcal V$.  Applying the first inequality with
$V = V_N$ and monotonicity of $\Tcal_{N+1}$ (coordinatewise
nondecreasing, since $Q_V$ increases in $V$),
$V_N = \Tcal_N V_N \le \Tcal_{N+1} V_N$; iterating the contraction
$\Tcal_{N+1}$ from $V_N$ gives $\Tcal_{N+1}^k V_N \to V_{N+1}$, so
$V_N \le V_{N+1}$, and iterating $\Tcal$ from $V_N$ likewise gives
$V_N \le V^\star$.

(c) Fix $x\in\X$ and $V\in\R^n$.  Let
$g_x(u) = Q_V(x,u)$ and $g_x^\star = \max_{\Ux_x} g_x$, attained on the
compact set $\mathsf{A} = \{u: g_x(u) = g_x^\star\}$.  For
$\varepsilon>0$, the tube $A^\varepsilon = \{u:
\dist(u,\mathsf{A}) < \varepsilon\}$ has $\mu_x(A^\varepsilon) > 0$ by
full support, and $g_x^\star - g_x$ is positive on the compact
complement, with positive minimum there.  Hence, with probability at
least $1 - e^{-N\mu_x(A^\varepsilon)}$, at least one candidate lies in
$A^\varepsilon$, on which event
$\max_i g_x(U_{x,i}) \ge g_x^\star - \omega_x(\varepsilon)$, where
$\omega_x(\varepsilon) = \max\{g_x^\star - g_x(u): \dist(u,\mathsf{A})
\le \varepsilon\}\downarrow 0$ by uniform continuity of $g_x$.
Therefore
\begin{equation}
\E\left[\max_i g_x(U_{x,i})\right] \to g_x^\star
\qquad\text{as } N\to\infty.
\end{equation}
Since $\X$ is finite, the componentwise convergence just established
gives $\norminfty{\Tcal_N V^\star - \Tcal V^\star}\to 0$.  Finally,
the triangle inequality and contraction give
\begin{equation}
\norminfty{V_N - V^\star}
\;\le\; \frac{1}{1-\bbeta}\,
\norminfty{\Tcal_N V^\star - \Tcal V^\star} \;\to\; 0.
\end{equation}
\end{proof}

\begin{remark}
Proposition~\ref{prop:contraction} guarantees that the asymptotics of
$V^\star - V_N$ are meaningful.  It is, however, deliberately coarse: the
bound carries no information about the scaling rate, the constants, or
the selection rule; all of these are determined by the geometric and
probabilistic structure developed in the next three sections, and the
exact rate is a fixed-point question, not a horizon question.
\end{remark}

\section{Geometry of Near-Optimal Actions}
\label{sec:geometry}

\subsection{Bellman Deficiency and the Optimal Action Set}
\label{sec:geometry:deficiency}

The distance of a candidate action from optimality is measured by the
\emph{Bellman deficiency}
\begin{equation}
\Delta_x(u) \;=\; V^\star(x) - Q_{V^\star}(x,u) \;\ge\; 0,
\qquad u\in\Ux_x,
\label{eq:deficiency}
\end{equation}
and the set of exactly optimal actions is
\begin{equation}
\mathsf{A}_x^\star \;=\; \bigl\{u\in\Ux_x : \Delta_x(u) = 0\bigr\}.
\label{eq:optimalset}
\end{equation}
By continuity of $\Delta_x$ and compactness of $\Ux_x$, the set
$\mathsf{A}_x^\star$ is nonempty and compact.  The value gap
$V^\star(x) - V_N(x)$ is ultimately determined by the distribution of the
random variable $\Delta_x(U)$ for $U\sim\mu_x$ in a shrinking neighborhood
of $\mathsf{A}_x^\star$; in the terminology of extreme value theory
\cite{deHaanFerreira06,DavidNagaraja03}, what matters is the Weibull-type
tail of $\Delta_x(U)$ near zero, and the purpose of this section is to
derive that tail law from three objects hidden inside a scalar tail
exponent: the \emph{geometry} of $\mathsf{A}_x^\star$, the
\emph{flatness order} of $\Delta_x$ normal to the optimal set, and the
\emph{sampling density} of $\mu_x$.

\subsection{Geometric Regularity: Stratified Optimal Sets with Homogeneous Normal
Deficiency}
\label{sec:geometry:assumption}

The standing structural assumption is the following.

\begin{assumption}[Stratified optimal set and homogeneous normal
deficiency]
\label{ass:G}
For each state $x\in\X$, the optimal action set is a finite union of
pairwise separated compact embedded $C^2$ submanifolds without
boundary:
\begin{equation}
\mathsf{A}_x^\star \;=\; \bigcup_{j=1}^{J_x} \calM_{xj},
\qquad \dim\calM_{xj} = k_{xj},
\label{eq:stratification}
\end{equation}
where ``separated'' means $\dist(\calM_{xj},\calM_{x\ell})>0$ for
$j\ne\ell$.  Each $\calM_{xj}$ lies in the interior of $\Ux_x$, with a
tubular neighborhood contained in $\Ux_x$ (no boundary sampling), and
$k_{xj} < d_x$ (positive codimension, so $\mathsf{A}_x^\star$ has zero
$\mu_x$-measure); boundary, angular, and self-intersecting
configurations are excluded, see Remark~\ref{rem:boundary}.  Write
$m_{xj} = d_x - k_{xj}$ for the normal codimension.
In a tubular neighborhood of $\calM_{xj}$, parameterize actions as
$u = \Theta_{xj}(z,v)$ with $z\in\calM_{xj}$ and normal vector
$v\in N_z\calM_{xj}$.  There exist $r_{xj}>0$ and functions
$q_{xj,z}: N_z\calM_{xj}\to\R$ such that the map
$(z,v)\mapsto q_{xj,z}(v)$ is continuous on the normal bundle
$N\calM_{xj}$, and:
\begin{enumerate}
\item[(i)] $q_{xj,z}(v) > 0$ for $v\ne 0$, and
$q_{xj,z}(av) = a^{r_{xj}} q_{xj,z}(v)$ for all $a>0$
($r_{xj}$-homogeneity);
\item[(ii)] uniformly in $z$,
\begin{equation}
\Delta_x\bigl(\Theta_{xj}(z,v)\bigr)
\;=\; q_{xj,z}(v) + o\bigl(\norm{v}^{r_{xj}}\bigr),
\qquad v\to 0.
\label{eq:homog}
\end{equation}
\end{enumerate}
Moreover, $\mu_x$ is absolutely continuous in a tubular neighborhood of
$\calM_{xj}$ with density $f_x$ satisfying, for
$v = \rho\theta$, $\rho>0$, $\theta$ a unit normal vector,
\begin{equation}
f_x\bigl(\Theta_{xj}(z,\rho\theta)\bigr)
\;=\; \rho^{\eta_{xj}}\, g_{xj}(z,\theta)
+ o\bigl(\rho^{\eta_{xj}}\bigr),
\label{eq:samplingdensity}
\end{equation}
uniformly in $z$ and $\theta$, where $\eta_{xj} > -m_{xj}$ and
$g_{xj}$ is continuous, nonnegative, and not identically zero.  All
remainders in \eqref{eq:homog}-\eqref{eq:samplingdensity} are uniform in
$z$ (and in $\theta$, for the density).
\end{assumption}

Assumption~\ref{ass:G} is a restrictive but verifiable local
regularity condition on the geometry of the optimal set.  It
is satisfied by the standard model in which each state has a unique
interior optimal action with smooth quadratically nondegenerate
deficiency and a sampling density continuous and positive in a
neighborhood of it ($k_{xj}=0$, $r_{xj}=2$, $\eta_{xj}=0$,
$m_{xj}/r_{xj}=d_x/2$, the classical rate of the Euclidean volume of a
$t$-level set of a smooth maximum).  It also covers the case in which
the optimal actions form a decision boundary of codimension one with a
nondegenerate linear deficiency normal to it ($k_{xj}=d_x-1$,
$r_{xj}=1$), where $m_{xj}/r_{xj} = 1$ irrespective of $d_x$.
Boundary-angular and self-intersecting optimal sets are excluded; see
Remark~\ref{rem:boundary}.

\subsection{The Near-Optimal Volume Law}
\label{sec:geometry:volume}

\begin{proposition}[Geometric near-optimal volume law]
\label{prop:volume}
Under Assumption~\ref{ass:G}, for each stratum $\calM_{xj}$ define
\begin{equation}
\alpha_{xj} \;=\; \frac{m_{xj} + \eta_{xj}}{r_{xj}}.
\label{eq:alpha}
\end{equation}
Then there exists a constant $C_{xj}\in(0,\infty)$ such that
\begin{equation}
\mu_x\bigl\{u : \Delta_x(u) \le t,\; u \text{ near } \calM_{xj}\bigr\}
\;\sim\; C_{xj}\, t^{\alpha_{xj}},
\qquad t\downarrow 0,
\label{eq:volume-law}
\end{equation}
where ``near $\calM_{xj}$'' means within a fixed tubular neighborhood
whose radius does not depend on $t$, and the constant is
\begin{equation}
C_{xj}
\;=\;
\int_{\calM_{xj}}
\int_{N_z\calM_{xj}}
\one\bigl\{q_{xj,z}(w)\le 1\bigr\}
\, g_{xj}\!\left(z,\tfrac{w}{\norm{w}}\right)
\, \norm{w}^{\eta_{xj}}\, dw\, d\calH^{k_{xj}}(z),
\label{eq:Cxj}
\end{equation}
with the direction term at $w=0$ defined arbitrarily (it is a
$\mu_x$-null set).  Set
\begin{equation}
\alpha_x \;=\; \min_j \alpha_{xj},
\qquad
\calJ_x^{\mathrm{crit}} \;=\; \{j : \alpha_{xj} = \alpha_x\},
\qquad
C_x \;=\; \sum_{j\in\calJ_x^{\mathrm{crit}}} C_{xj}.
\label{eq:critical}
\end{equation}
Then the overall near-optimal volume satisfies
\begin{equation}
\mu_x\bigl\{u : \Delta_x(u) \le t\bigr\}
\;\sim\; C_x\, t^{\alpha_x},
\qquad t\downarrow 0.
\label{eq:overall-volume}
\end{equation}
\end{proposition}

\begin{proof}
The contribution of one stratum $\calM_{xj}$ is computed; the summation
over $j$ at the end uses the separation assumption.  In tubular
coordinates $u = \Theta_{xj}(z,v)$ with $z\in\calM_{xj}$ and
$v\in N_z\calM_{xj}$ restricted to $\norm{v}\le\rho_0$ (a uniform tubular
radius exists because $\calM_{xj}$ is compact and $C^2$), the Euclidean
volume element satisfies
\begin{equation}
du \;=\; \bigl(1 + O(\norm{v})\bigr)\, dv\, d\calH^{k_{xj}}(z).
\label{eq:volume-element}
\end{equation}
For $t\le t_0$ small, the region
$\{\Delta_x(u)\le t,\; u \text{ near } \calM_{xj}\}$ is contained in the
tubular neighborhood, and the uniformity in \eqref{eq:homog} gives the
sandwich
\begin{equation*}
\{q_{xj,z}(v)\le (1-\delta)t\}
\;\subseteq\;
\{\Delta_x(u)\le t,\; u\text{ near }\calM_{xj}\}
\;\subseteq\;
\{q_{xj,z}(v)\le (1+\delta)t\},
\end{equation*}
in tubular coordinates, for every $\delta>0$ and all sufficiently
small $t$.  By $r_{xj}$-homogeneity, $q_{xj,z}(t^{1/r_{xj}}w) =
t\,q_{xj,z}(w)$, so in the $w$ coordinates $v = t^{1/r_{xj}}w$ the two
outer regions are $\{q_{xj,z}(w)\le 1\pm\delta\}$, with
$t^{-\alpha_{xj}}\mu_x$-masses converging, as $t\downarrow 0$, to
$(1\pm\delta)^{\alpha_{xj}}$ times the right-hand side of
\eqref{eq:volume-limit}; sandwiching and letting $\delta\downarrow 0$
identifies the middle limit.  The normal volume element contributes a
factor
$t^{m_{xj}/r_{xj}}$, and the density asymptotics \eqref{eq:samplingdensity}
contributes $t^{\eta_{xj}/r_{xj}}$, for a total exponent
$\alpha_{xj} = (m_{xj}+\eta_{xj})/r_{xj}$.  Concretely,
\begin{align}
& t^{-\alpha_{xj}}\,
\mu_x\bigl\{u : \Delta_x(u)\le t,\; u\text{ near }\calM_{xj}\bigr\} \notag\\
&\qquad \longrightarrow
\int_{\calM_{xj}} \int_{N_z\calM_{xj}}
\one\bigl\{q_{xj,z}(w)\le 1\bigr\}
\, g_{xj}\!\left(z,\tfrac{w}{\norm{w}}\right)
\, \norm{w}^{\eta_{xj}}\, dw\, d\calH^{k_{xj}}(z),
\label{eq:volume-limit}
\end{align}
where \eqref{eq:volume-element}, the homogeneity factors, and the
uniformity of the remainders in
\eqref{eq:homog}-\eqref{eq:samplingdensity} were used; dominated
convergence justifies the limit uniformly in the sandwich parameter
$\delta$, and the sandwich identifies the limit as $\delta\downarrow 0$.
To check finiteness: by continuity and strict positivity of $q_{xj,z}$
on the unit normal sphere there are uniform constants $0<c<C<\infty$
with
\begin{equation}
c\norm{w}^{r_{xj}} \;\le\; q_{xj,z}(w) \;\le\; C\norm{w}^{r_{xj}},
\qquad \forall\, z,\ w,
\label{eq:qcC}
\end{equation}
so $\{q_{xj,z}(w)\le 1\}$ lies in the ball $\norm{w}\le c^{-1/r_{xj}}$
and $g_{xj}$ is bounded; $\eta_{xj} > -m_{xj}$ makes
$\int_{\norm{w}\le R}\norm{w}^{\eta_{xj}}dw$ finite, and $\calM_{xj}$
has finite $k_{xj}$-dimensional Hausdorff measure.  This proves
\eqref{eq:volume-law}-\eqref{eq:Cxj}.

It remains to handle the region outside the union of the tubular
neighborhoods and the summation over strata.  The complement
$\Ux_x\setminus\bigcup_j N_j$ is compact, and $\Delta_x$ attains a
strictly positive minimum $\delta_0>0$ there (the $\calM_{xj}$ are
exactly the zero set of $\Delta_x$, finitely many of them), so for
$t<\delta_0$ no mass from the complement contributes to
\eqref{eq:overall-volume}.  For the summation, the $j$-stratum
contribution scales as $t^{\alpha_{xj}}$; as $t\downarrow 0$ the
smallest exponent dominates and only the strata with
$\alpha_{xj}=\alpha_x$ survive, giving $C_x$ in \eqref{eq:critical}.
This proves \eqref{eq:overall-volume}.
\end{proof}

\subsection{Interpretation}
\label{sec:geometry:interpretation}

Under the common-exponent condition \eqref{eq:common-alpha} of
Section~5, the candidate-pool value gap scales as
$V^\star(x) - V_N(x) \asymp N^{-1/\alpha_x}$, the exponent set by the
\emph{critical} strata (noncritical strata give slower local rates) as
the ratio of the \emph{effective codimension} $m_{xj}+\eta_{xj}$ (the
normal codimension $m_{xj} = d_x-k_{xj}$ plus the sampling-density
exponent $\eta_{xj}$ of \eqref{eq:samplingdensity}) to the
\emph{flatness order} $r_{xj}$ of the homogeneity \eqref{eq:homog},
i.e. $N^{-1/\alpha_x} = N^{-r_{xj}/(m_{xj}+\eta_{xj})}$ for
$j\in\calJ_x^{\mathrm{crit}}$.  Canonical instances: the isolated
quadratic maximum gives
$\alpha_x = d_x/2$ (rate $N^{-2/d_x}$); a sharp decision boundary
gives $\alpha_x = 1$ (rate $N^{-1}$) regardless of $d_x$; a sampling
density vanishing linearly in the normal direction ($\eta_{xj} = 1$)
adds one to the effective codimension.

\begin{remark}[Excluded boundary and angular cases]
\label{rem:boundary}
If an optimal stratum lies on the boundary of $\Ux_x$, or two strata
intersect at an angle, the volume law \eqref{eq:volume-law} acquires
boundary or conic correction factors; a complete treatment requires
stratified integration on manifolds with corners
\cite{deHaanFerreira06} and is deferred.
\end{remark}

\section{Marked Poisson Limits}
\label{sec:poisson}

The scalar deficiency $\Delta_x(U)$ does not propagate the gap through
the dynamics: near-optimal candidates may have very different
transition kernels, which matters at the leading order.  Each candidate
therefore carries as a mark its \emph{projection onto the optimal
action set}, which determines its transition kernel.

\subsection{Projection and Intensity Measures}
\label{sec:poisson:measure}

For each state $x$, let $\pi_x$ denote the nearest-point projection onto
$\mathsf{A}_x^\star$: in a tubular neighborhood, $\pi_x(u)$ is the unique
point of $\mathsf{A}_x^\star$ realizing $\dist(u,\mathsf{A}_x^\star)$;
outside these neighborhoods $\pi_x$ is extended arbitrarily to a
measurable map $\Ux_x\to\mathsf{A}_x^\star$, the choice being irrelevant
for all limits below (only candidates with $\Delta_x(u)\le t$,
$t\downarrow 0$, contribute, lying in the tubes eventually).

By Proposition~\ref{prop:volume}, the same tubular-coordinate
calculation localizes the near-optimal volume law to the optimal set:
for every Borel set $B\subseteq\mathsf{A}_x^\star$ with
$\sigma_x(\partial B)=0$,
\begin{equation}
t^{-\alpha_x}\;\mu_x\bigl\{u : \Delta_x(u)\le t,\; \pi_x(u)\in B\bigr\}
\;\longrightarrow\; \sigma_x(B),
\qquad t\downarrow 0,
\label{eq:local-volume}
\end{equation}
where the finite measure $\sigma_x$ on $\mathsf{A}_x^\star$ is defined on
each stratum $\calM_{xj}$ as the $z$-integration in
\eqref{eq:Cxj}, restricted to the critical strata:
\begin{equation}
\sigma_x(\calM_{xj}) \;=\; C_{xj},
\qquad
\sigma_x(\mathsf{A}_x^\star) \;=\; C_x
\;=\; \sum_{j\in\calJ_x^{\mathrm{crit}}} C_{xj}.
\label{eq:sigma}
\end{equation}
The identity \eqref{eq:local-volume} is proved in the Supplement: it
follows from the same tubular-coordinate calculation as
Proposition~\ref{prop:volume}, with the $z$-integration restricted to
$B$.

\subsection{Convergence to a Marked Poisson Point Process}
\label{sec:poisson:convergence}

\begin{proposition}[Marked extreme point process]
\label{prop:poisson}
Fix a state $x$ and let $\alpha_x$ and $\sigma_x$ be as in
\eqref{eq:critical}-\eqref{eq:sigma}.  Define the point process on
$[0,\infty)\times\mathsf{A}_x^\star$ by
\begin{equation}
\Pi_{x,N}
\;=\;
\sum_{i=1}^{N}
\delta_{\bigl(N^{1/\alpha_x}\Delta_x(U_{x,i}),\,
\pi_x(U_{x,i})\bigr)}.
\label{eq:PiN}
\end{equation}
Then $\Pi_{x,N}$ converges in distribution, in the vague topology on
locally finite counting measures on $[0,\infty)\times\mathsf{A}_x^\star$,
to a Poisson point process $\Pi_x$ with intensity measure
\begin{equation}
\Lambda_x(dz,du) \;=\; \alpha_x z^{\alpha_x-1}\, dz\, \sigma_x(du),
\qquad z\ge 0,\; u\in\mathsf{A}_x^\star,
\label{eq:Lambda}
\end{equation}
where the $z$-coordinate is integrated against Lebesgue measure and
$\Lambda_x(\{0\}\times\mathsf{A}_x^\star) = 0$ (no atom at $z=0$).
Equivalently,
\begin{equation}
\Lambda_x\bigl([0,z]\times B\bigr)
\;=\; z^{\alpha_x}\, \sigma_x(B),
\qquad z\ge 0,\; B\subseteq\mathsf{A}_x^\star \text{ Borel}.
\label{eq:Lambda-z}
\end{equation}
\end{proposition}

\begin{proof}
By Kallenberg's criterion \cite{Kallenberg21} it suffices to prove
Laplace-functional convergence for nonnegative continuous functions
with compact support.
Let $f:[0,\infty)\times\mathsf{A}_x^\star\to[0,\infty)$ be continuous
with support contained in $[0,L]\times\mathsf{A}_x^\star$.  By
\eqref{eq:local-volume} applied with $B$ the projection of the support,
and the change of variables $z = N^{1/\alpha_x} t$,
\begin{align}
& N\, \E\left[1 - e^{-f\bigl(N^{1/\alpha_x}\Delta_x(U),\,\pi_x(U)\bigr)}\right]
\notag\\
&\qquad = N \int_0^\infty \int_{\mathsf{A}_x^\star}
\left(1 - e^{-f(z,u)}\right)
\, d\mu_x\Bigl\{u : N^{1/\alpha_x}\Delta_x(u)\in dz,\; \pi_x(u)\in du\Bigr\} \notag\\
&\qquad \longrightarrow
\int_0^\infty \int_{\mathsf{A}_x^\star}
\left(1 - e^{-f(z,u)}\right)
\, \alpha_x z^{\alpha_x-1}\, dz\, \sigma_x(du),
\label{eq:laplace}
\end{align}
where the limit uses the pointwise convergence of the rescaled volume law
\eqref{eq:local-volume}, boundedness of $f$, and dominated convergence
(the integrand is bounded by $1$ and the intensity measure is finite on
the support of $f$).  The independence of the candidates then gives the
Laplace functional of $\Pi_{x,N}$:
\begin{align}
\E\exp\Bigl\{-\bigl\langle f,\Pi_{x,N}\bigr\rangle\Bigr\}
&=\;
\left[1 - \E\left(1 - e^{-f\bigl(N^{1/\alpha_x}\Delta_x(U),\,\pi_x(U)\bigr)}\right)
\right]^N \notag\\
&\;\longrightarrow\;
\exp\Bigl\{-\int_{\R_+\times\mathsf{A}_x^\star}
\left(1 - e^{-f(z,u)}\right)\,\Lambda_x(dz,du)\Bigr\},
\label{eq:laplace-limit}
\end{align}
which is the Laplace functional of the Poisson point process with
intensity $\Lambda_x$ (\cite{DaleyVereJones03}).  The vague-convergence statement follows (topological details and the
extension to all bounded continuous functions are in the Supplement),
and \eqref{eq:Lambda-z} is immediate from \eqref{eq:Lambda} by
Fubini.
\end{proof}

\subsection{Minimal Deficiencies and Static Geometric Selection}
\label{sec:poisson:orderstats}

For a fixed state $x$, let
\begin{equation}
D_{x,1:N}\;\le\; D_{x,2:N}\;\le\;\cdots\;\le\; D_{x,N:N}
\end{equation}
be the ascending order statistics of $\Delta_x(U_{x,1}),\dots,
\Delta_x(U_{x,N})$.

\begin{proposition}[Minimal deficiencies, order statistics, and static selection]
\label{prop:orderstats}
Under Assumption~\ref{ass:G}, for each fixed $j\ge 1$ and $N\ge j$,
\begin{equation}
N^{1/\alpha_x} D_{x,j:N} \;\Rightarrow\;
C_x^{-1/\alpha_x}\, \Gamma_j^{1/\alpha_x},
\label{eq:orderstat-limit}
\end{equation}
where $\Gamma_j = E_1 + \cdots + E_j$ with $E_i$ i.i.d.\ unit
exponentials.  Moreover,
\begin{equation}
N^{1/\alpha_x}\, \E\, D_{x,j:N}
\;\longrightarrow\;
C_x^{-1/\alpha_x}\,
\frac{\Gamma(j + 1/\alpha_x)}{\Gamma(j)},
\label{eq:orderstat-moment}
\end{equation}
If the critical optimal set consists of several separated strata, the
stratum hosting the least deficient candidate satisfies
\begin{equation}
\Prob\{J_{x,N} = j\}
\;\longrightarrow\;
\frac{C_{xj}}{\sum_{\ell\in\calJ_x^{\mathrm{crit}}} C_{x\ell}},
\qquad j\in\calJ_x^{\mathrm{crit}},
\label{eq:branch-prob}
\end{equation}
while strata with $\alpha_{xj}>\alpha_x$ are selected with probability
tending to zero.
\end{proposition}

\begin{proof}
From \eqref{eq:overall-volume},
\begin{equation}
\Prob\bigl(N^{1/\alpha_x} D_{x,1:N} > z\bigr)
\;=\;
\left[1 - \Prob\bigl(\Delta_x(U) \le z N^{-1/\alpha_x}\bigr)\right]^N
\;\longrightarrow\;
e^{-C_x z^{\alpha_x}},
\label{eq:weibull}
\end{equation}
so $N^{1/\alpha_x}D_{x,1:N}$ converges in distribution to a
Weibull-type minimum law with tail $e^{-C_x z^{\alpha_x}}$.  For the joint law of the leading order
statistics, apply Proposition~\ref{prop:poisson}: the ordered points of
$\Pi_x$ have the representation $Z_j = C_x^{-1/\alpha_x}
\Gamma_j^{1/\alpha_x}$, $j\ge 1$, where $\Gamma_j$ is the $j$-th arrival
time of a unit-rate Poisson process, since the map
$z\mapsto C_x z^{\alpha_x}$ converts the intensity
\eqref{eq:Lambda} to unit rate.  This gives \eqref{eq:orderstat-limit}.

For the moment convergence \eqref{eq:orderstat-moment}, the family
$\{N^{1/\alpha_x}D_{x,j:N}\}$ must be uniformly integrable: the
minimum ($j=1$) is so by Lemma~\ref{lem:ui-min}; for $j\ge 2$ the
minimum tail does not control $D_{x,j:N}$ (order statistics move only
upward), and Lemma~S.1 of the Supplement supplies the uniform
integrability for every fixed $j$ through the exact binomial tail of
the $j$-th order statistic.  Hence the family is uniformly
integrable, and the distributional limit gives the moment limit, the
gamma moments of \eqref{eq:orderstat-limit}:
\begin{equation*}
\E\,\Gamma_j^{1/\alpha_x} \;=\; \frac{\Gamma(j+1/\alpha_x)}{\Gamma(j)}.
\end{equation*}
For \eqref{eq:branch-prob}: the limit Poisson processes on the
separated strata are independent, on stratum $j$ the minimal
$z$-coordinate $Z_j$ has tail $\Prob(Z_j > z) = e^{-C_{xj} z^{\alpha_x}}$
by \eqref{eq:Lambda-z}; since the strata are separated, the global
minimum lies on the stratum with the smallest $Z_j$, and
\begin{align}
\Prob\bigl(Z_j < \min_{\ell\ne j} Z_\ell\bigr)
&=\;
\int_0^\infty \alpha_x C_{xj} z^{\alpha_x-1}\,
e^{-\bigl(\sum_\ell C_{x\ell}\bigr) z^{\alpha_x}}\, dz
\;=\;
\frac{C_{xj}}{\sum_\ell C_{x\ell}},
\label{eq:branch-integral}
\end{align}
which gives the limit \eqref{eq:branch-prob}.  A stratum with
$\alpha_{xj}>\alpha_x$ cannot host the minimum with positive
probability in the limit: its expected number of candidates within
rescaled deficiency $zN^{-1/\alpha_x}$ is
\begin{equation*}
N\,\Prob\bigl(\Delta_x(U)\le zN^{-1/\alpha_x}\bigr)
\;\sim\; C_{xj} z^{\alpha_{xj}} N^{1-\alpha_{xj}/\alpha_x}
\;\longrightarrow\; 0,
\end{equation*}
since $\alpha_{xj}>\alpha_x$.
\end{proof}

\begin{remark}
Proposition~\ref{prop:orderstats} exhibits a \emph{static} geometric
selection: among exactly tied optimal actions, the finite random pool
favors the strata that are easier to sample (larger $C_{xj}$).  In the
dynamic model the strata carry different transition kernels, and
selection enters the Bellman fixed point instead (Section~6).
\end{remark}

\begin{lemma}\label{lem:ui-min}
Under Assumption~\ref{ass:G}, for every $x\in\X$ and every
$\eta>0$,
\begin{equation}
\sup_{N\ge 1}\;\E\Bigl[\bigl(N^{1/\alpha_x}D_{x,1:N}\bigr)^{1+\eta}\Bigr]
\;<\;\infty,
\label{eq:ui-min}
\end{equation}
where $D_{x,1:N} := \min_{1\le i\le N}\Delta_x(U_{x,i})$ is the
minimal deficiency among the sampled candidates; and the same bound
holds, uniformly over $\norminfty{h}\le R$ for every $R<\infty$, for
the shifted minima
$m_N(h) := \min_{1\le i\le N}\{N^{1/\alpha_x}\Delta_x(U_{x,i})
+ \bbeta P_{x,U_{x,i}}h\}$, where
$P_{x,u}h = \sum_{y\in\X}p_{xy}(u)h(y)$.
\end{lemma}
\begin{proof}
By \eqref{eq:volume-law} there are constants $0<c<C<\infty$ such that
$c\,t^{\alpha_x}\le\Prob(\Delta_x(U)\le t)\le C\,t^{\alpha_x}$ for
all sufficiently small $t$.  Writing $Y_N := N^{1/\alpha_x}D_{x,1:N}$
and $\Delta_x^{\max} := \max_{u\in\Ux_x}\Delta_x(u)<\infty$, the
Weibull tail gives $\Prob(Y_N>z)\le e^{-c z^{\alpha_x}}$ for $z\le
z_N := N^{1/\alpha_x}x_0$ ($x_0$ small), while for
$z_N<z\le N^{1/\alpha_x}\Delta_x^{\max}$,
$\Prob(Y_N>z)\le\Prob(D_{x,1:N}>x_0)\le e^{-c' N}$.  The layer-cake representation
$\E[Y_N^{1+\eta}]\le 1+(1+\eta)\int_1^\infty z^{\eta}\Prob(Y_N>z)\,dz$
splits the integral at $z_N$: the Weibull regime contributes the
gamma integral $\int_1^\infty z^{\eta}e^{-c z^{\alpha_x}}dz<\infty$,
and the intermediate regime at most
$(N^{1/\alpha_x}\Delta_x^{\max})^{1+\eta}e^{-c'N}\to 0$.  Finally,
$|m_N(h)-m_N|\le\bbeta R$ almost surely, so
$\E[|m_N(h)|^{1+\eta}]\le 2^{1+\eta}(\E[Y_N^{1+\eta}]+(\bbeta
R)^{1+\eta})$.
\end{proof}

\section{The Tangent Extreme-Value Bellman Equation}
\label{sec:tangent}

\subsection{Common Critical Exponent}
\label{sec:tangent:common}

The main theorem requires that all states share a common critical
exponent:
\begin{equation}
\alpha_x \;\equiv\; \alpha
\qquad \forall\, x\in\X.
\label{eq:common-alpha}
\end{equation}
This does not require the action dimensions, optimal-set dimensions, or
flatness orders to coincide across states, only
$(d_x - k_{xj} + \eta_{xj})/r_{xj} = \alpha$ for the critical strata of
every state (the heterogeneous case is treated in Section~7).  Under
\eqref{eq:common-alpha}, set
\begin{equation}
\veps_N \;=\; N^{-1/\alpha},
\label{eq:epsN}
\end{equation}
and for $h\in\R^n$ define the \emph{rescaled deficiency operator}
\begin{equation}
(\Phix_N h)(x)
\;=\;
\veps_N^{-1}\left[ V^\star(x) - \Tcal_N\bigl(V^\star - \veps_N h\bigr)(x) \right],
\qquad x\in\X.
\label{eq:PhiN}
\end{equation}
Expanding the definition of $\Tcal_N$,
\begin{equation}
(\Phix_N h)(x)
\;=\;
\E \min_{1\le i\le N}
\left\{ \veps_N^{-1}\Delta_x(U_{x,i}) + \bbeta P_{x,U_{x,i}} h \right\},
\label{eq:PhiN-expanded}
\end{equation}
where $P_{x,u} h = \sum_{y\in\X} p_{xy}(u) h(y)$ is the transition
operator of action $u$.  This coupling matters: a near-optimal
candidate carries both its local extreme deficiency and the mark (its
transition kernel) acting on the future value gap.

\subsection{The Limit Operator and Its Properties}
\label{sec:tangent:Psi}

Let $\Pi_x$ be the Poisson point process of
Proposition~\ref{prop:poisson}.  Define the \emph{tangential
extreme-value Bellman operator}
\begin{equation}
(\PSi h)(x)
\;=\;
\E \inf_{(z,u)\in\Pi_x}
\left\{ z + \bbeta P_{x,u} h \right\},
\qquad h\in\R^n,\; x\in\X,
\label{eq:Psi}
\end{equation}
where the marks $u\in\mathsf{A}_x^\star$ are the projections of the
candidate actions.  Because $\Lambda_x$ is finite on bounded $z$-intervals
but has infinite total mass, the infimum in \eqref{eq:Psi} is almost
surely finite and attained.

\begin{lemma}[Contraction of the tangential operators]
\label{lem:contraction}
For every $N\ge 1$, $\Phix_N$ and $\PSi$ are $\bbeta$-contractions on
$(\R^n,\norminfty{\cdot})$:
\begin{equation}
\norminfty{\Phix_N h - \Phix_N k}
\;\le\; \bbeta\norminfty{h-k},
\qquad
\norminfty{\PSi h - \PSi k}
\;\le\; \bbeta\norminfty{h-k},
\qquad \forall\, h,k\in\R^n.
\label{eq:contraction-bound}
\end{equation}
Consequently $\PSi$ has a unique fixed point $h^\star$:
\begin{equation}
h^\star \;=\; \PSi h^\star.
\label{eq:hstar}
\end{equation}
\end{lemma}

\begin{proof}
For any action $u$,
$|P_{x,u}h - P_{x,u}k| \le \norminfty{h-k}$.  Comparing the two infima
over the same candidate points, their difference is bounded by the
maximum over the candidates of the perturbation, hence
$\norminfty{\Phix_N h - \Phix_N k}\le\bbeta\norminfty{h-k}$ after taking
the expectation.  The same argument applies to $\PSi$ pointwise on each
realization of $\Pi_x$.  Banach's fixed point theorem gives the unique
fixed point of $\PSi$.
\end{proof}

\begin{lemma}[Local uniform Poisson-envelope convergence]
\label{lem:uniform}
Assume Assumption~\ref{ass:G}, the common-exponent condition
\eqref{eq:common-alpha}, and continuity of the transition probabilities.
For every $R<\infty$,
\begin{equation}
\sup_{\norminfty{h}\le R} \norminfty{\Phix_N h - \PSi h}
\;\longrightarrow\; 0
\qquad\text{as } N\to\infty.
\label{eq:uniform-convergence}
\end{equation}
\end{lemma}

\begin{proof}
\emph{Step 1: pointwise distributional convergence.}
Fix $x$ and $h$ with $\norminfty{h}\le R$.  Write
$m_N(h) = \min_{1\le i\le N}\{N^{1/\alpha_x}\Delta_x(U_{x,i})
+ \bbeta P_{x,U_{x,i}}h\}$ and $m(h) =
\inf_{(z,u)\in\Pi_x}\{z + \bbeta P_{x,u}h\}$.  The first step replaces the
actions of the candidates by their projections to $\mathsf{A}_x^\star$.
Candidates with rescaled deficiency $z_i := N^{1/\alpha_x}\Delta_x(U_{x,i})$
bounded lie in $\{\Delta_x \le cN^{-1/\alpha_x}\}$, which shrinks to
$\mathsf{A}_x^\star$ ($\Delta_x\ge 0$ continuous, vanishing exactly
there); as $u\mapsto P_{x,u}h$ is continuous and
$\sup_{\Delta_x(u)\le s}\operatorname{dist}(u,\mathsf{A}_x^\star)\to 0$
as $s\downarrow 0$ by compactness of $\Ux_x$, the projection error of
the bounded-deficiency candidates is controlled by a deterministic
sequence: for $L<\infty$ define
\begin{equation}
\eta_N(L) \;:=\;
\sup\Bigl\{ \bigl|P_{x,u}h - P_{x,\pi_x(u)}h\bigr| :
\Delta_x(u) \le L N^{-1/\alpha_x} \Bigr\},
\label{eq:etaN}
\end{equation}
then $\eta_N(L)\to 0$ for every $L$.  Candidates
with $z_i > t+\bbeta R$ cannot be minimizers at level $t$
($P_{x,u}h\ge -\norminfty{h}\ge -R$).  The \emph{projected} minimum
$\hat m_N(h) = \min_i\{z_i + \bbeta P_{x,\pi_x(U_{x,i})}h\}$ therefore
satisfies, for the deterministic sequence $\eta_N(t+\bbeta R)\to 0$,
\[
\Prob\bigl(m_N(h)>t\bigr) \;\ge\;
\Prob\bigl(\hat m_N(h) > t+\bbeta\eta_N\bigr),\qquad
\Prob\bigl(m_N(h)>t\bigr) \;\le\;
\Prob\bigl(\hat m_N(h) > t-\bbeta\eta_N\bigr).
\]
Now $\hat m_N(h) > t$ iff the empirical marked point process
$\Pi_{x,N}$ of Proposition~\ref{prop:poisson} is empty on
\begin{equation}
B_t \;=\; \bigl\{(z,u)\in\R_+\times\mathsf{A}_x^\star :
 z + \bbeta P_{x,u}h \le t \bigr\},
\label{eq:tail-set}
\end{equation}
a relatively compact set (its points satisfy $z\le t+\bbeta R$),
whose boundary $\partial B_t = \{z = t - \bbeta P_{x,u}h\}$ is a graph
over $u$.  The $z$-marginal of the intensity measure \eqref{eq:Lambda}
has the continuous density $\alpha_x z^{\alpha_x-1}$, so
$\Lambda_x(\partial B_t) = 0$ for every $t$.  The empty-set functional
$\xi\mapsto 1\{\xi(B_t)=0\}$ is continuous in the vague topology at
every configuration without mass on $\partial B_t$, and
$\Prob(\Pi_x(\partial B_t)=0)=1$; by Proposition~\ref{prop:poisson}
and the continuous mapping theorem,
\begin{equation}
\Prob\bigl(m_N(h) > t\bigr)
\;\longrightarrow\;
\Prob\bigl(\Pi_x(B_t) = 0\bigr)
\;=\; \exp\{-\Lambda_x(B_t)\}
\;=\; \Prob\bigl(m(h) > t\bigr),
\label{eq:tail-limit}
\end{equation}
for every $t$; the limiting survival function is continuous in $t$
(similarly $\Lambda_x(\partial B_t)=0$), so $m_N(h)\Rightarrow m(h)$ in
distribution.  \emph{Shifted marks.}  The argument also applies when the values are
shifted by a deterministic perturbation $a_N\colon\Ux_x\to\R$ with
$\sup_u|a_N(u)-a(\pi_x(u))|\to 0$ for some bounded $a$ continuous on
$\mathsf{A}_x^\star$.  The shifts may be negative, so a
shifted point process is not formed; instead the minimum functional is
perturbed directly.  Write $m^a_N(h) = \min_i\{z_i + a_N(U_{x,i}) + \bbeta
P_{x,U_{x,i}}h\}$ and, for the projected version,
$\tilde m^a_N(h) = \min_i\{z_i + a(\pi_x(U_{x,i})) + \bbeta
P_{x,\pi_x(U_{x,i})}h\}$.  Candidates with
$z_i > t + \|a\|_\infty + \bbeta R + 1$ exceed the level $t$ in both
minima (their shifted marks are at least
$z_i - \|a_N-a\circ\pi_x\|_\infty - \|a\|_\infty - \bbeta R > t$ for
$N$ large), so both minima are decided on the complementary sublevel
set, where the projection error of Step~1 is at most
$\bbeta\,\eta_N(t+\|a\|_\infty+\bbeta R+1)$; hence, with
$\delta_N = \sup_u|a_N(u)-a(\pi_x(u))| + \bbeta\,\eta_N(t+\|a\|_\infty
+\bbeta R+1)\to 0$,
\[
\Prob\bigl(\tilde m^a_N(h) > t+\delta_N\bigr) \;\le\;
\Prob\bigl(m^a_N(h) > t\bigr) \;\le\;
\Prob\bigl(\tilde m^a_N(h) > t-\delta_N\bigr).
\]
The projected minimum has the empty-set characterization of Step~1 in
terms of the \emph{unshifted} process $\Pi_{x,N}$, with the set
$B^a_t = \{(z,u)\in\R_+\times\mathsf{A}_x^\star : z + a(u) + \bbeta
P_{x,u}h \le t\}$ in place of \eqref{eq:tail-set}; its boundary
$\partial B^a_t = \{z = t - a(u) - \bbeta P_{x,u}h\}$ is a graph over
$u$ with $\Lambda_x(\partial B^a_t) = 0$, the $z$-marginal of the
intensity \eqref{eq:Lambda} having a continuous density, so
\begin{equation}
\Prob\bigl(m^a_N(h) > t\bigr)
\;\longrightarrow\;
\Prob\bigl(m_a(h) > t\bigr),
\label{eq:tail-shifted}
\end{equation}
where $m_a(h) := \inf_{(z,u)\in\Pi_x}\{z + a(u) + \bbeta P_{x,u}h\}$ is
the limiting minimum of the shifted marks.  This variant is used in
Section~6 (Theorem~\ref{thm:tiebreak}, Step~1).

\emph{Step 2: convergence of expectations (uniform integrability).}
Write
\begin{equation*}
m_N \;=\; \min_i \veps_N^{-1}\Delta_x(U_{x,i}),
\qquad
m_N(h) \;=\; \min_i\bigl\{\veps_N^{-1}\Delta_x(U_{x,i})
+ \bbeta P_{x,U_{x,i}}h\bigr\}.
\end{equation*}
Since $|P_{x,u}h|\le R$,
\begin{equation}
|m_N(h) - m_N| \;\le\; \bbeta R
\end{equation}
almost surely.  The family $\{m_N\}_N$ is uniformly integrable
(Lemma~\ref{lem:ui-min}), and the second claim of that lemma extends
this to the shifted minima $\{m_N(h)\}_N$ (the perturbation
$|m_N(h)-m_N|\le \bbeta R$).  Combined with Step~1, this gives
\begin{equation}
\E m_N(h) \;\longrightarrow\; \E \inf_{(z,u)\in\Pi_x}\{z + \bbeta P_{x,u} h\},
\qquad x\in\X,
\label{eq:expectation-limit}
\end{equation}
i.e.\ $\Phix_N h \to \PSi h$ pointwise in $h$.

\emph{Step 3: uniformity on bounded balls.}
Both $\Phix_N$ and $\PSi$ are $\bbeta$-Lipschitz (Lemma~\ref{lem:contraction}).
Fix $\delta>0$ and a finite $\delta$-net $\{h^\ell\}$ of the ball
$\{h: \norminfty{h}\le R\}$ in $\R^n$.  For each $h$ in the ball choose
$h^\ell$ with $\norminfty{h - h^\ell}\le\delta$.  Then
\begin{align}
\norminfty{\Phix_N h - \PSi h}
&\le\; \norminfty{\Phix_N h - \Phix_N h^\ell}
+ \norminfty{\Phix_N h^\ell - \PSi h^\ell}
+ \norminfty{\PSi h^\ell - \PSi h} \notag\\
&\le\; 2\bbeta\delta + \max_\ell \norminfty{\Phix_N h^\ell - \PSi h^\ell}.
\end{align}
Taking $N\to\infty$ and then $\delta\to 0$ gives \eqref{eq:uniform-convergence}.
\end{proof}

\subsection{The Main Theorem}
\label{sec:tangent:main}

\begin{theorem}[Poisson tangential limit of random-candidate fixed points]
\label{thm:main}
Assume Proposition~\ref{prop:contraction} hypotheses,
Assumption~\ref{ass:G}, the common-exponent condition
\eqref{eq:common-alpha}, and continuity of the transition probabilities.
Let $h^\star$ be the unique fixed point \eqref{eq:hstar} of $\PSi$.  Then
\begin{equation}
N^{1/\alpha}\bigl(V^\star - V_N\bigr)
\;\longrightarrow\; h^\star
\qquad\text{in } \norminfty{\cdot},
\label{eq:main-limit}
\end{equation}
equivalently,
\begin{equation}
V_N \;=\; V^\star - N^{-1/\alpha} h^\star + o\bigl(N^{-1/\alpha}\bigr).
\label{eq:main-expansion}
\end{equation}
\end{theorem}

\begin{proof}
Set
\begin{equation}
h_N \;=\; \veps_N^{-1}\bigl(V^\star - V_N\bigr).
\label{eq:hN}
\end{equation}
Using $V_N = \Tcal_N V_N$ and the definition \eqref{eq:PhiN} of $\Phix_N$,
\begin{align}
(\Phix_N h_N)(x)
&=\;
\veps_N^{-1}\left[ V^\star(x) - \Tcal_N\bigl(V^\star - \veps_N h_N\bigr)(x) \right] \notag\\
&=\;
\veps_N^{-1}\left[ V^\star(x) - \Tcal_N V_N(x) \right]
\;=\;
\veps_N^{-1}\left[ V^\star(x) - V_N(x) \right]
\;=\; h_N(x),
\label{eq:hN-fixed}
\end{align}
so $h_N$ is a fixed point of $\Phix_N$.  By Lemma~\ref{lem:contraction},
\begin{equation}
\norminfty{h_N}
\;=\; \norminfty{\Phix_N h_N - \Phix_N 0 + \Phix_N 0}
\;\le\; \bbeta\norminfty{h_N} + \norminfty{\Phix_N 0},
\end{equation}
hence
\begin{equation}
\norminfty{h_N}
\;\le\; \frac{\norminfty{\Phix_N 0}}{1 - \bbeta}.
\label{eq:hN-bound}
\end{equation}
By Lemma~\ref{lem:uniform}, $\Phix_N 0 \to \PSi 0$, so the sequence
$\{h_N\}$ is bounded.  Finally, using the fixed-point equations
$h_N = \Phix_N h_N$ and $h^\star = \PSi h^\star$,
\begin{align}
\norminfty{h_N - h^\star}
&=\;
\norminfty{\Phix_N h_N - \PSi h^\star} \notag\\
&\le\;
\norminfty{\Phix_N h_N - \Phix_N h^\star}
+ \norminfty{\Phix_N h^\star - \PSi h^\star} \notag\\
&\le\;
\bbeta\norminfty{h_N - h^\star}
+ \norminfty{\Phix_N h^\star - \PSi h^\star},
\end{align}
so that
\begin{equation}
\norminfty{h_N - h^\star}
\;\le\;
\frac{1}{1-\bbeta}\,
\norminfty{\Phix_N h^\star - \PSi h^\star}
\;\longrightarrow\; 0
\label{eq:main-bound}
\end{equation}
by Lemma~\ref{lem:uniform} ($h^\star$ is bounded by the same
argument).  Substituting \eqref{eq:hN} gives \eqref{eq:main-limit} and
\eqref{eq:main-expansion}.
\end{proof}

\begin{remark}[Nonlinearity at the leading order]
The leading coefficient of the value gap is \emph{not} a tail constant
multiplied by a resolvent factor, except in the single-optimum case, in which
the optimal actions are unique (Corollary~\ref{cor:resolvent}) or share a
common transition kernel.  When several optimal actions have different
kernels, the leading term is the fixed point of the \emph{nonlinear}
operator $\PSi$: candidate reachability, the geometry of the optimal
set, and the transition kernels compete at the same order; $\PSi$ is
the limit of the rescaled deficiency operators \eqref{eq:PhiN} of the
family $\{\Tcal_N\}_{N\ge1}$, a differentiability statement about the
family rather than the single map $\Tcal$ (cf.\ the Introduction).
\end{remark}

\subsection{The Unique-Optimum Case: Linear Resolvent}
\label{sec:tangent:resolvent}

\begin{corollary}[Unique optimal action: linear resolvent]
\label{cor:resolvent}
Suppose that for every state $x$ the optimal action is unique:
$\mathsf{A}_x^\star = \{u_x^\star\}$.  Define the transition matrix
$P^\star_{xy} = p_{xy}(u_x^\star)$ and the vector
\begin{equation}
c_x \;=\; C_x^{-1/\alpha}\, \Gamma\bigl(1 + 1/\alpha\bigr),
\qquad x\in\X,
\label{eq:c}
\end{equation}
with $C_x$ as in \eqref{eq:critical}.  Then
\begin{equation}
(\PSi h)(x) \;=\; c_x + \bbeta \sum_{y\in\X} P^\star_{xy} h(y),
\label{eq:Psi-affine}
\end{equation}
and the fixed point of Theorem~\ref{thm:main} is
\begin{equation}
h^\star
\;=\; (I - \bbeta P^\star)^{-1} c
\;=\;
\sum_{t=0}^{\infty} \bbeta^t (P^\star)^t c.
\label{eq:resolvent}
\end{equation}
\end{corollary}

\begin{proof}
When the optimal action is unique, every point of the Poisson process
$\Pi_x$ has the same mark $u_x^\star$, and therefore the same transition
operator $P_{x,u} = P^\star$.  The infimum in \eqref{eq:Psi} then
factorizes:
\begin{equation}
(\PSi h)(x)
\;=\;
\E\Bigl[\min_{(z,u)\in\Pi_x} z\Bigr] + \bbeta \sum_y P^\star_{xy} h(y).
\label{eq:Psi-factor}
\end{equation}
The minimum $Z_x$ of the $z$-coordinates has tail
$\Prob(Z_x > z) = \exp\{-C_x z^{\alpha}\}$ by \eqref{eq:Lambda-z}, so
\begin{equation}
\E Z_x
\;=\;
\int_0^\infty \Prob(Z_x > z)\, dz
\;=\;
C_x^{-1/\alpha} \Gamma\bigl(1 + 1/\alpha\bigr)
\;=\; c_x.
\label{eq:EZ}
\end{equation}
This proves \eqref{eq:Psi-affine}.  Since
$\norm{\bbeta P^\star}_\infty \le \bbeta < 1$, the operator
$I - \bbeta P^\star$ is invertible and the Neumann series gives
\eqref{eq:resolvent}.
\end{proof}

\begin{remark}
Corollary~\ref{cor:resolvent} shows that fixed-policy resolvent
asymptotics are the single-optimum specialization of
Theorem~\ref{thm:main} (unique optimal actions, or all near-optimal
actions sharing a common transition kernel): the gap is $N^{-1/\alpha}$
times the discounted occupancy of the critical states under the limit
policy; Theorem~\ref{thm:hetero} is its heterogeneous version.
\end{remark}

\section{Geometric Tie-Breaking in the Candidate Pool}
\label{sec:tiebreak}

Section~5 establishes the leading-order value gap but leaves open a
practical question: when the continuous-action model has \emph{several}
optimal actions, which one does the finite-pool controller effectively
employ?  This section answers it for isolated optimal actions, locally
unique but globally multiple, by a \emph{selection theorem}: the
empirical optimal candidate converges to a random element of the
optimal set with law explicit in the limit Poisson process and the
fixed point $h^\star$; perturbing the tie at the critical rate then
organizes value and selection in a critical switching layer
(Section~\ref{sec:switching}).

\subsection{Isolated Optimal Actions and the Nonlinear Selection Equation}
\label{sec:tiebreak:setup}

\begin{assumption}[Isolated optimal actions]
\label{ass:isolated}
For every state $x$, the optimal action set is a finite set of isolated
points:
\begin{equation}
\mathsf{A}_x^\star \;=\; \{u_{x,1}^\star,\dots,u_{x,J_x}^\star\},
\qquad J_x < \infty,
\label{eq:isolated}
\end{equation}
and Assumption~\ref{ass:G} holds with $\calM_{xj} = \{u_{x,j}^\star\}$,
$k_{xj}=0$, so that the critical exponent of stratum $j$ is
\begin{equation}
\alpha_{xj} \;=\; \frac{d_x + \eta_{xj}}{r_{xj}},
\label{eq:alpha-isolated}
\end{equation}
with $C_{xj}$ the local volume constant
\eqref{eq:Cxj} of the stratum.
\end{assumption}

Under Assumption~\ref{ass:isolated}, each stratum is a singleton and
the mark set is finite.

\begin{corollary}[Tangential equation for isolated optimal actions]
\label{cor:isolated}
Assume Assumptions~\ref{ass:G} and \ref{ass:isolated}, and let
$\calJ_x^{\mathrm{crit}} = \{j : \alpha_{xj} = \alpha_x\}$ be the
critical strata of state $x$.  Let $(Z_{x,j})_{j\in\calJ_x^{\mathrm{crit}}}$
be independent random variables with Weibull tails
\begin{equation}
\Prob\bigl(Z_{x,j} > z\bigr)
\;=\;
e^{-C_{xj} z^{\alpha_x}},
\qquad z\ge 0,
\label{eq:Ztail}
\end{equation}
and write $P_{x,j} = P_{x,u_{x,j}^\star}$ for the transition operator of
action $u_{x,j}^\star$.  Then for every state $x$ and every
$h\in\R^n$,
\begin{equation}
(\PSi h)(x)
\;=\;
\E\, \min_{j\in\calJ_x^{\mathrm{crit}}}
\left\{ Z_{x,j} + \bbeta\, P_{x,j}\, h \right\},
\label{eq:Psi-isolated}
\end{equation}
where the minimum over an empty index set is interpreted as $+\infty$;
strata with $\alpha_{xj}>\alpha_x$ do not contribute to the expectation.
\end{corollary}

\begin{proof}
For an isolated optimal action $u_{x,j}^\star$, the Poisson process
$\Pi_x$ has a single mark $u_{x,j}^\star$ on the stratum $\calM_{xj}$,
and the marginal $z$-coordinate process on that stratum is Poisson with
intensity measure $C_{xj}\alpha_x z^{\alpha_x-1}dz$
(by \eqref{eq:Lambda} with $\sigma_x(\{u_{x,j}^\star\}) = C_{xj}$).  The
smallest point of that process has distribution
\eqref{eq:Ztail} by the same calculation as \eqref{eq:EZ}.  The Poisson
processes on distinct strata are independent by Proposition
\ref{prop:poisson} (the strata are separated), so the family
$(Z_{x,j})$ is independent.  The mark of every point of the process on
stratum $j$ is $u_{x,j}^\star$, so the infimum over the stratum is
$Z_{x,j} + \bbeta P_{x,j}h$, and the infimum over the whole set is the
minimum over $j$.  For a stratum with $\alpha_{xj}>\alpha_x$ the expected number of
candidates with rescaled deficiency below $z$ is
$C_{xj}z^{\alpha_{xj}}N^{1-\alpha_{xj}/\alpha_x}\to 0$, so the
probability that stratum $j$ hosts the global minimum tends to zero
(cf.\ Proposition~\ref{prop:orderstats}).
\end{proof}

\begin{remark}
The variables $Z_{x,j}$ are the \emph{rescaled minimal deficiencies}
of the sampled candidates around $u_{x,j}^\star$; their scales
$C_{xj}^{-1/\alpha_x}$, from \eqref{eq:Cxj}, encode the flatness order
$r_{xj}$ and the sampling exponent $\eta_{xj}$, so larger $C_{xj}$
(easier sampling) means a stochastically smaller $Z_{x,j}$ and favored
selection, everything else equal.  The novelty of
Theorem~\ref{thm:tiebreak} is that the ``everything else'' is
\emph{not} equal: the kernels $P_{x,j}$ enter through the fixed point
$h^\star$, coupling the selection rule across states.
\end{remark}

\subsection{The Exact Tie: Dynamic Geometric Tie-Breaking}
\label{sec:tiebreak:main}

At the exact tie the finite pool must break the tie among several
optimal actions of a state.  Under Assumption~\ref{ass:isolated} the
answer is a genuine selection theorem: the empirical optimal candidate
converges in distribution to a random optimal action whose law is
explicit and governed jointly by the local volume constants $C_{xj}$
and the future-cost differences $\bbeta P_{x,j}h^\star$ propagated
through the transition kernels; the convergence arguments below rest
on a general continuity property of the argmin.

\begin{lemma}[Continuous argmin of converging random costs]
\label{lem:selector}
Let $S$ be a compact metric space, and let $f_N, f: S \to \R$ be
random continuous functions such that $\sup_{y\in S}|f_N(y) - f(y)|
\to 0$ a.s.\ on a common probability space.  If the set of minimizers
of $f$ is a.s.\ the singleton $\{I\}$, then every measurable selection
$\hat I_N \in \argmin_{y\in S} f_N(y)$ satisfies $\hat I_N \Rightarrow
I$.  If only weak convergence of the random elements $f_N, f$ in
$(C(S),\|\cdot\|_\infty)$ is known, the conclusion is unchanged: $S$
compact makes $C(S)$ Polish, so the Skorokhod representation applies.
\end{lemma}

\begin{proof}
On the coupling, $\min_S f_N \to \min_S f$ a.s.  Because $I$ is a.s.\
unique and $f$ is continuous on the compact space $S$, for every
$\delta > 0$ the separation $\eta_\delta := \min_{S\setminus
B(I,\delta)} f - f(I)$ is a.s.\ positive; eventually
$\sup_S|f_N - f| \le \eta_\delta/2$, so every minimizer of $f_N$ lies
in $B(I,\delta)$ and $\hat I_N \to I$ a.s., hence $\hat I_N
\Rightarrow I$.
\end{proof}

Define the \emph{empirical optimal candidate} $I_{x,N}$ of state $x$ in
a pool of size $N$ as any fixed measurable tie-breaking choice among
the candidates of the realized pool maximizing $Q_{V_N}(x,\cdot)$,
recorded as its index in the list of critical strata; the choice is
immaterial for the limits below, the limiting selector of
Theorem~\ref{thm:tiebreak} being almost surely a singleton.  If the
maximum is attained outside the critical strata, set $I_{x,N}$
arbitrarily (the event has probability tending to zero by
Proposition~\ref{prop:orderstats}).

\begin{theorem}[Dynamic geometric tie-breaking]
\label{thm:tiebreak}
Assume the hypotheses of Theorem~\ref{thm:main} and
Assumption~\ref{ass:isolated}.  Let $h^\star$ be the fixed point of
$\PSi$, and let $(Z_{x,j})_{j\in\calJ_x^{\mathrm{crit}}}$ be the
independent Weibull variables of Corollary~\ref{cor:isolated}, jointly
independent across states.  Then for every state $x$ and every critical
stratum $j$,
\begin{equation}
\Prob\{I_{x,N} = j\}
\;\longrightarrow\;
\Prob\Bigl\{ j \in \argmin_{\ell\in\calJ_x^{\mathrm{crit}}}
\bigl\{ Z_{x,\ell} + \bbeta\, P_{x,\ell}\, h^\star \bigr\} \Bigr\}.
\label{eq:tiebreak-limit}
\end{equation}
The limiting selector
\begin{equation}
I_x^\star
\;=\;
\argmin_{j\in\calJ_x^{\mathrm{crit}}}
\left\{ Z_{x,j} + \bbeta\, P_{x,j}\, h^\star \right\}
\label{eq:selector}
\end{equation}
is a.s.\ a singleton, and its law is explicit: writing
$\Phi_j(t) = 0$ for $t<0$ and $\Phi_j(t) = 1 - e^{-C_{xj} t^{\alpha_x}}$
for $t\ge 0$ for the cumulative
distribution function of $Z_{x,j}$,
\begin{equation}
\Prob\{I_x^\star = j\}
\;=\;
\int_0^\infty
\prod_{\ell\ne j}
\left[ 1 - \Phi_\ell\!\left( z + \bbeta\, P_{x,j}h^\star
- \bbeta\, P_{x,\ell}h^\star \right) \right]
\, d\Phi_j(z).
\label{eq:selection-prob}
\end{equation}
If the critical optimal actions of state $x$ satisfy the offset
agreement $P_{x,\ell}h^\star = P_{x,j}h^\star$ for every pair of
critical strata $\ell,j$ (in particular if they share a common
transition operator), then
\begin{equation}
\Prob\{I_x^\star = j\}
\;=\;
\frac{C_{xj}}{\sum_{\ell\in\calJ_x^{\mathrm{crit}}} C_{x\ell}},
\label{eq:static-selection}
\end{equation}
the static volume ratio of Proposition~\ref{prop:orderstats}.
\end{theorem}

\begin{proof}
\emph{Step 1: the empirical selector converges to the limit selector.}
Fix $x$ and a pool of size $N$.  Let $\veps_N = N^{-1/\alpha_x}$,
and write
\begin{equation}
\hat Z_{x,j} \;=\;
\veps_N^{-1}\min_{1\le i\le N}
\left\{ \Delta_x(U_{x,i}) : \pi_x(U_{x,i}) = u_{x,j}^\star \right\}
\end{equation}
for the rescaled minimal deficiency among the candidates projected onto
stratum $j$; if no candidate projects onto $j$, set
$\hat Z_{x,j}=+\infty$.  By Proposition~\ref{prop:poisson} applied on
each stratum, the vector $(\hat Z_{x,j})_j$ converges in distribution
to $(Z_{x,j})_j$, and $\hat Z_{x,j}\to\infty$ for noncritical $j$.

The candidate chosen by the controller maximizes $Q_{V_N}$; equivalently
it minimizes the \emph{rescaled} deficiency-plus-future term
\begin{multline}
\veps_N^{-1}\left[ V_N(x) - Q_{V_N}(x,u) \right]
\;=\;
\veps_N^{-1}\left[ V^\star(x) - Q_{V^\star}(x,u) \right]
+ \bbeta \sum_{y} p_{xy}(u) \frac{V^\star(y) - V_N(y)}{\veps_N} \\
+ \veps_N^{-1}\bigl[ V_N(x) - V^\star(x) \bigr].
\label{eq:decomp}
\end{multline}
The first term is the rescaled deficiency $\veps_N^{-1}\Delta_x(u)$;
the third term is $-\veps_N^{-1}(V^\star(x)-V_N(x))\to -h^\star(x)$,
independent of $j$ and cancelling from the argmin.  For the future term,
define the \emph{branchwise completed costs}
\begin{equation}
C_{N,j} \;=\; \min_{1\le i\le N}\Bigl\{
\veps_N^{-1}\Delta_x(U_{x,i})
+ \bbeta\, P_{x,U_{x,i}}\, h_N
\;:\; \pi_x(U_{x,i}) = u_{x,j}^\star \Bigr\},
\label{eq:completed-cost-N}
\end{equation}
with $C_{N,j} = +\infty$ if no candidate projects onto $j$; these are
the finite-pool analogues of the completed-cost offsets
\eqref{eq:completed-cost} of the limit process.  Since
$C_{N,j}\le \hat Z_{x,j} + \bbeta\norminfty{h_N}$ and
$\hat Z_{x,j} = O_p(1)$ (convergence in distribution above, with
$\norminfty{h_N}=O(1)$), the minimizing
candidate of \eqref{eq:completed-cost-N} satisfies
$\veps_N^{-1}\Delta_x(U) = O_p(1)$; by the homogeneity \eqref{eq:homog} it lies at normal
distance $O(\veps_N^{1/r_{xj}})$ from $\calM_{xj}$, its kernel
converges to $p_{x\cdot}(u_{x,j}^\star)$ (continuity, $h_N\to
h^\star$).  For the reverse inequality, replace $h_N$ by $h^\star$ in
\eqref{eq:completed-cost-N} (cost $\bbeta\norminfty{h_N-h^\star}\to
0$) and bound each candidate's kernel by continuity.
Hence
\begin{equation}
\max_j \Bigl| C_{N,j} - \bigl( \hat Z_{x,j} + \bbeta P_{x,j} h^\star
\bigr) \Bigr| \;\xrightarrow{\mathrm{p}}\; 0,
\label{eq:completed-cost-limit}
\end{equation}
where the maximum runs over the critical strata (for noncritical $j$,
$\hat Z_{x,j}\to\infty$ and $C_{N,j}\to\infty$).  Combining the
convergence in distribution of $(\hat Z_{x,j})_j$ with
\eqref{eq:completed-cost-limit}, the minimizer
\begin{equation}
\hat I_{x,N}
\;=\;
\argmin_j C_{N,j}
\label{eq:empirical-selector}
\end{equation}
satisfies $\hat I_{x,N}\Rightarrow I_x^\star$ in distribution:
on the event $E_N$ that every critical stratum contains at least one
candidate ($\Prob(E_N)\to 1$ by Proposition~\ref{prop:poisson}), the
costs $C_{N,j}$ are finite and, by \eqref{eq:completed-cost-limit},
approximate $y_j + \bbeta P_{x,j} h^\star$ in probability; the
limiting argmin \eqref{eq:selector} is a.s.\ a singleton, so
Lemma~\ref{lem:selector} applies on $E_N$; the absent-stratum event
carries cost $+\infty$ and contributes nothing to the limit.

\emph{Step 2: the empirical optimal candidate is the minimizer of
\eqref{eq:empirical-selector}.}
The critical candidates have rescaled deficiencies of order one and
the noncritical ones diverge, so the chosen candidate is critical with
probability tending to one; on this event, minimizing the rescaled
deficiency-plus-future term is equivalent to maximizing $Q_{V_N}$, by
\eqref{eq:decomp} (the $j$-independent shift cancels).  By
\eqref{eq:completed-cost-limit}, Slutsky's theorem, the almost-sure
uniqueness of the limiting argmin
(Lemma~\ref{lem:selector}), and the continuous mapping theorem,
$\widehat I_{x,N}\Rightarrow I_x^\star$; hence
$\Prob\{I_{x,N}=j\}\to\Prob\{I_x^\star=j\}$.

\emph{Step 3: the explicit formula.}
By independence and continuity of the $Z_{x,\ell}$,
\begin{align}
\Prob\{I_x^\star = j\}
&=\;
\E\left[ \one\Bigl\{ Z_{x,j} + \bbeta P_{x,j}h^\star
< \min_{\ell\ne j}\{Z_{x,\ell} + \bbeta P_{x,\ell}h^\star\} \Bigr\} \right] \notag\\
&=\;
\int_0^\infty
\prod_{\ell\ne j}
\Prob\Bigl( Z_{x,\ell} > z + \bbeta P_{x,j}h^\star
- \bbeta P_{x,\ell}h^\star \Bigr)
\, d\Phi_j(z),
\label{eq:prob-integral}
\end{align}
which is \eqref{eq:selection-prob}.  If the offset
agreement $P_{x,\ell}h^\star = P_{x,j}h^\star$ holds for all critical
$\ell,j$ (in particular when the transition operators coincide), the
offsets cancel, and the probability becomes the
probability that $Z_{x,j}$ is the smallest of the independent Weibull
variables, which by the same exponential-tails computation as
\eqref{eq:branch-integral} equals the volume ratio
\eqref{eq:static-selection}.
\end{proof}

\begin{remark}[Static vs.\ dynamic selection]
The static branch probabilities \eqref{eq:branch-prob} of
Proposition~\ref{prop:orderstats} describe which stratum hosts the
\emph{least deficient} candidate, ignoring the transition kernels; the
dynamic selection probabilities \eqref{eq:selection-prob} describe
which stratum hosts the \emph{selected} candidate, depending on them
through the offsets $\bbeta P_{x,j}h^\star$.  The two agree precisely
when the offset terms agree, i.e.\ when $P_{x,\ell}h^\star =
P_{x,j}h^\star$ for all critical $\ell,j$ of state $x$ (a common
transition kernel is sufficient but not necessary); otherwise the pool
deviates from the static ratios through the fixed point.  A genuine
model symmetry forcing $P_{x,\ell}h^\star = P_{x,j}h^\star$ collapses
the dynamic effect to the static ratios; dynamic selection needs the
Q-values tied but shadow offsets distinct.
In Poisson
terms, the selected candidate is the point of $\Pi_x$ minimizing the
\emph{completed} cost $z + \bbeta P_{x,u}h^\star$: $\bbeta
P_{x,j}h^\star$ is the exact first-order shadow price of choosing
$u_{x,j}^\star$ in the candidate regime, so the tie-breaking is
``geometric'', entirely determined by the local geometry of the value
gap through the constants $C_{xj}$ and the fixed point $h^\star$.
Section~\ref{sec:numerics:ex2} exhibits the symmetric example with
dynamic probability $0.75$ against the static $0.5$.
\end{remark}

\subsection{The Poisson-Mecke Selection Formula}
\label{sec:mecke}

The selection law of Theorem~\ref{thm:tiebreak} is stated for
stratified optimal sets; conditional on the fixed point $h^\star$, the
law of the selected mark has a closed form, directly from the
Slivnyak-Mecke formula, with no stratification (finite-pool
convergence covered by Theorem~\ref{thm:tiebreak} under
Assumption~\ref{ass:isolated}).  Fix a state $x$, let
$\Lambda_x(dz,du) = \alpha_x z^{\alpha_x-1}dz\,\sigma_x(du)$ be the
intensity \eqref{eq:Lambda} of the marked Poisson limit of
Proposition~\ref{prop:poisson}, and let $h^\star$ be the fixed point
\eqref{eq:hstar} of the limit operator.  Write
\begin{equation}
b(u) \;=\; \bbeta\, P_{x,u}\, h^\star
\label{eq:completed-cost}
\end{equation}
for the \emph{completed-cost offset} of action $u$, and let $U_x^\infty$
be the mark of the point of $\Pi_x$ minimizing $z + b(u)$.

\begin{theorem}[Poisson-Mecke selection]
\label{thm:mecke}
The minimizer of $z + b(u)$ over the points of $\Pi_x$ is almost
surely unique, and for every Borel set $B$ of the optimal set
$\mathsf{A}_x^\star$,
\begin{equation}
\begin{aligned}
\Prob\{U_x^\infty \in B\} &= \int_{\R_+\times B}
\exp\Bigl\{ - \Lambda_x\bigl\{(z',u')\in\R_+\times\mathsf{A}_x^\star :
\\
&\qquad\qquad z'+b(u')<z+b(u)\bigr\} \Bigr\}
\, \Lambda_x(dz,du).
\end{aligned}
\label{eq:mecke}
\end{equation}
\end{theorem}

\begin{proof}
\emph{Step 1: uniqueness.}  Two distinct points $(z,u),(z',u')$ of
$\Pi_x$ tie in the completed cost only if $z' = z + b(u) - b(u')$;
for fixed $(z,u,u')$ this is a single value of $z'$, a $\Lambda_x$-null
set since the $z$-marginal $\alpha_x z^{\alpha_x-1}dz$ is diffuse;
integrating against the Palm measure shows that the probability of a
tie between any two points of $\Pi_x$ is zero, and the minimizer is
a.s.\ unique;
$U_x^\infty$ is therefore a.s.\ well defined.

\emph{Step 2: the Slivnyak-Mecke identity.}  For a Poisson process
$\Pi_x$ with intensity $\Lambda_x$ and a bounded measurable
$f(p, \Pi_x\setminus\{p\})$, the Slivnyak-Mecke formula states
(see \cite{LastPenrose17})
\begin{equation}
\E\!\!\sum_{p\in\Pi_x} f\bigl(p, \Pi_x\setminus\{p\}\bigr)
\;=\;
\int \Lambda_x(dp)\; \E\, f(p, \Pi_x),
\label{eq:slivnyak}
\end{equation}
where on the right the point $p$ is fixed and $\Pi_x$ is an
independent copy.  Apply \eqref{eq:slivnyak} to
\[
f\bigl(p, \Pi_x^- \bigr)
\;=\;
\mathbf{1}_B(u_p)\,
\mathbf{1}\bigl\{ z_q + b(u_q) > z_p + b(u_p)
\;\text{for all } q\in\Pi_x^-\bigr\},
\]
the indicator that $p$ is the (unique, by Step~1) minimizer of the
completed cost and carries its mark in $B$.  The summand on the left
of \eqref{eq:slivnyak} is then $\mathbf{1}\{U_x^\infty\in B\}$ for the
single minimizing point, so the left-hand side equals
$\Prob\{U_x^\infty\in B\}$.  On the right, conditioning on the fixed
point $(z,u)$, the event that no point of the independent copy falls
in the sublevel set $\{z' + b(u') < z + b(u)\}$ has probability
$\exp\{-\Lambda_x\{z' + b(u') < z + b(u)\}\}$ by the void probability
of a Poisson process, which gives \eqref{eq:mecke}.  Taking
$B = \mathsf{A}_x^\star$ makes the inner indicator identically one,
and the identity \eqref{eq:slivnyak} yields the normalization
$\Prob\{U_x^\infty\in\mathsf{A}_x^\star\} = 1$.

\emph{Step 3: stratified specialization.}  Under
\ref{ass:G} with $\sigma_x = \sum_j C_{xj}\,\delta_{u_{x,j}^\star}$,
the sublevel set of Step~2 for $u = u_{x,j}^\star$ has measure
$\sum_\ell C_{x\ell}\,(z + b_j - b_\ell)_+^{\alpha_x}$, where
$b_j = \bbeta P_{x,j}h^\star$; hence \eqref{eq:mecke} reads
\begin{equation}
\Prob\{U_x^\infty = u_{x,j}^\star\}
\;=\;
C_{xj}\int_0^\infty
\exp\Bigl\{ - \textstyle\sum_\ell C_{x\ell}\,
\bigl(z + b_j - b_\ell\bigr)_+^{\alpha_x} \Bigr\}
\, \alpha_x z^{\alpha_x-1}\, dz.
\label{eq:mecke-discrete}
\end{equation}
Isolating the $\ell = j$ term $e^{-C_{xj}z^{\alpha_x}}$ reproduces the
Weibull density $d\Phi_j(z)$ of Corollary~\ref{cor:isolated}, and the
remaining product $e^{-C_{x\ell}(z+b_j-b_\ell)_+^{\alpha_x}} =
1-\Phi_\ell(z+b_j-b_\ell)$ gives exactly the dynamic selection
probabilities \eqref{eq:selection-prob} of
\ref{thm:tiebreak}.  When the offsets agree, $b_\ell = b_j$
for all critical $\ell,j$, the integral in \eqref{eq:mecke-discrete}
collapses to $C_{xj}\int_0^\infty \alpha_x z^{\alpha_x-1}
e^{-C_x z^{\alpha_x}} dz = C_{xj}/C_x$, the static volume ratio
\eqref{eq:static-selection}.
\end{proof}

\begin{remark}
The formula \eqref{eq:mecke} is a property of the marked Poisson
process alone: given the intensity $\Lambda_x$ it needs no
stratification of $\mathsf{A}_x^\star$, and the integration in $u$ over
a continuous component of the optimal set is an ordinary integral
against $\sigma_x$ (the mark density after integrating out $z$ is
$\sigma_x(du)\int_0^\infty \alpha_x z^{\alpha_x-1}
e^{-\Lambda_x\{z' + b(u') < z + b(u)\}}\,dz$).  What
\emph{does} require the stratified geometry of Assumption~\ref{ass:G}
is the convergence of the empirical pool to $\Pi_x$ itself
(Proposition~\ref{prop:poisson}); on the boundary of the action space
the volume law and the limit intensity may change.
\end{remark}

\subsection{Perturbing the Tie: The Critical Switching Layer}
\label{sec:switching}

Theorem~\ref{thm:main} identifies the leading constant of the value gap
at a \emph{fixed} model, and Theorem~\ref{thm:tiebreak} resolved the
selection at the \emph{exact tie}.  Between these two extremes lies the response of value and selection
to a perturbation making one branch slightly better.  This response takes place on a \emph{critical
switching layer} of
width $\veps_N = N^{-1/\alpha}$ in the branch-offset parameter: if
the gap is of order $\veps_N$, the limit
is a genuine two-branch competition governed by a one-parameter family
of nonlinear fixed points interpolating between the two
branch-specific endpoint fixed points (linear branch resolvents when
every other state has a unique optimal action); if the
gap is much larger, one branch wins almost surely and its endpoint
fixed point governs.

For a state $x$ with two critical branches, let $\theta$ denote the
offset of the branch-2 minimum relative to the branch-1 minimum, and let
$\theta_c = 0$ be the critical offset at which the branches are tied.
Consider the one-parameter family of models at offset
$\theta_N = \theta_c + \veps_N \xi$.  Formally, the family is a
\emph{parameterized MDP} whose $Q$-function takes the form
\begin{equation}
\begin{split}
Q^{N,\xi}_V(x,u)
&= Q_V(x,u)
- \veps_N\bigl[\,\lambda_{x,\pi_x(u)}(\xi)
+ \rho_N(x,u,\xi)\bigr],\\
\lambda_{x1}(\xi) &= \max(0,-\xi),\quad
\lambda_{x2}(\xi) = \max(0,\xi),
\end{split}
\label{eq:param-Q}
\end{equation}
with \emph{unchanged transition probabilities} $p^{N,\xi}_{xy} = p_{xy}$;
equivalently, the perturbation is a stage-cost penalty
$r^{N,\xi}_x(u) = r_x(u) - \veps_N[\lambda_{x,\pi_x(u)}(\xi)
+ \rho_N(x,u,\xi)]$.  (For $u$ outside the tubes of the critical strata the
value of $\lambda_{x,\pi_x(u)}$ is immaterial: such candidates have
deficiency bounded away from zero and never win the minimum.)  The
remainder is assumed uniformly negligible on compact intervals,
\begin{equation}
\sup_{\xi\in K}\;\sup_{x\in\X}\;\sup_{u\in\Ux_x} |\rho_N(x,u,\xi)|
\;\longrightarrow\; 0,
\qquad N\to\infty,
\label{eq:param-remainder}
\end{equation}
for every compact $K\subset\R$.
The perturbed Bellman deficiency, defined through the $Q$-function at
the unperturbed value $V^\star$,
\begin{equation}
\Delta^{N,\xi}_x(u) \;:=\;
V^\star(x) - Q^{N,\xi}_{V^\star}(x,u)
\;=\;
\Delta_x(u) + \veps_N\bigl[\,\lambda_{x,\pi_x(u)}(\xi)
+ \rho_N(x,u,\xi)\bigr],
\label{eq:switching-perturb}
\end{equation}
thus expands the unperturbed deficiency $\Delta_x$ by the rescaled
branch gap with a uniform remainder.  Writing $\Delta_{xj}(\theta_N)$ for
the minimal deficiency of branch $j$ \emph{relative to the best branch}
at offset $\theta_N$, \eqref{eq:switching-perturb} is equivalently
\begin{equation}
\lim_{N\to\infty}\, \veps_N^{-1} \Delta_{xj}(\theta_N)
\;=\; \lambda_{xj}(\xi),
\label{eq:switching-gap}
\end{equation}
so $\lambda_{xj}(\xi)$ is the rescaled branch gap at the switching
offset; note that $\min_j \lambda_{xj}(\xi) = 0$, i.e.\ the better branch
is never penalized, and $\lambda_{x1}(\xi) - \lambda_{x2}(\xi) = -\xi$.

\begin{theorem}[Critical switching layer]
\label{thm:switching}
Assume the hypotheses of Theorem~\ref{thm:main} together with
Assumption~\ref{ass:isolated}, and let $\calJ_x^{\mathrm{crit}} = \{1,2\}$
for the state $x$ under consideration.  Let $(Z_{x,j})_{j=1,2}$ be the
independent Weibull variables of Corollary~\ref{cor:isolated} and define
the \emph{switching operator}
\begin{equation}
(\PSi_\xi h)(x)
\;=\;
\E \min_{j=1,2}
\left\{ Z_{x,j} + \lambda_{xj}(\xi) + \bbeta\, P_{x,j}\, h \right\},
\label{eq:Psi-xi}
\end{equation}
with $\lambda_{xj}$ as in \eqref{eq:param-Q}; for states without
a tie, $\PSi_\xi$ coincides with $\PSi$ of \eqref{eq:Psi-isolated}.  Then
the following hold for the parameterized family \eqref{eq:param-Q} with
sampled fixed point $V_N^{\xi} = \Tcal_N^{\xi} V_N^{\xi}$, where
$\Tcal_N^{\xi} V(x) = \E\max_{1\le i\le N} Q^{N,\xi}_V(x,U_{x,i})$.
\begin{enumerate}
\item[(a)] $\PSi_\xi$ is a $\bbeta$-contraction on $\R^n$ uniformly in
$\xi$; let $h_\xi$ be its unique fixed point.  The map $\xi\mapsto h_\xi$
is continuous, $h_0 = h^\star$, and $h_\xi$ is \emph{minimized at the
exact tie}: coordinatewise, $\xi\mapsto h_\xi$ is nonincreasing on
$(-\infty,0]$ and nondecreasing on $[0,\infty)$, so that
$h_\xi(x) \ge h^\star(x)$ for every state $x$ and every $\xi$ (comparison
principle, Proposition~\ref{prop:contraction}).
\item[(b)] The rescaled value gap converges in sup norm,
\begin{equation}
\lim_{N\to\infty}\,
\veps_N^{-1}\bigl( V^\star(x) - V_N^{\xi}(x) \bigr)
\;=\; h_\xi(x),
\label{eq:switching-gap-limit}
\end{equation}
uniformly for $\xi$ in compact intervals, where $V^\star$ is the value of
the unperturbed continuous-action model.  (The perturbed
continuous-action value $V^{\star,N,\xi}$ is the fixed point of the
$\bbeta$-contraction $V\mapsto \max_{u\in\Ux_x} Q^{N,\xi}_V(\cdot,u)$;
since $\Delta_x\ge 0$ and $\lambda_{x,\pi_x(u)}(\xi)\ge 0$ with
$\min_j\lambda_{xj}(\xi)=0$, the minimal perturbed deficiency
$\delta_N^\xi := \min_{u\in\Ux_x}\Delta^{N,\xi}_x(u)$ satisfies, by
\eqref{eq:switching-perturb},
\begin{equation*}
\delta_N^\xi \;=\; \inf_{u\in\Ux_x}\Bigl\{ \Delta_x(u)
+ \veps_N\bigl[\lambda_{x,\pi_x(u)}(\xi) + \rho_N(x,u,\xi)\bigr] \Bigr\},
\end{equation*}
hence $-\veps_N\|\rho_N\|_\infty \le \delta_N^\xi \le
\veps_N\|\rho_N\|_\infty$ (the lower bound by $\Delta_x,\lambda\ge 0$,
the upper attained on the better branch, where $\Delta_x = 0$ and
$\lambda = 0$); the fixed-point comparison of a $\bbeta$-contraction
gives $\|V^{\star,N,\xi} - V^\star\|_\infty \le
(1-\bbeta)^{-1}\veps_N\|\rho_N\|_\infty = o(\veps_N)$ by
\eqref{eq:param-remainder}, uniformly in $\xi$ on compact intervals.)
\item[(c)] The empirical selector converges: writing $I_{x,N}^{\xi}$ for
the branch selected by the finite-pool controller of the perturbed model,
\begin{multline}
\Prob\{ I_{x,N}^{\xi} = 1 \} \;\longrightarrow\; p_1(\xi)
\;:=\;
\Prob\Bigl\{ Z_{x,1} + \lambda_{x1}(\xi) + \bbeta P_{x,1} h_\xi \\
< Z_{x,2} + \lambda_{x2}(\xi) + \bbeta P_{x,2} h_\xi \Bigr\},
\label{eq:switching-prob}
\end{multline}
and $p_1$ is continuous in $\xi$, with $p_1(\xi)\to 0$ as
$\xi\to -\infty$ and $p_1(\xi)\to 1$ as $\xi\to +\infty$.  (If the two branches share a common transition operator at $x$, then
$p_1(\xi) = \Prob\{Z_{x,1} - Z_{x,2} < \xi\}$, a continuous
distribution function, hence nondecreasing.)
\item[(d)] Endpoint regimes.  As $\xi\to +\infty$, $\lambda_{x2}(\xi) =
\xi\to\infty$ forces branch~2 out of the minimum and
$\PSi_\xi h \to \PSi^{x,1}h$ pointwise in $h$, where the mixed operator
$\PSi^{x,1}$ agrees with $\PSi$ at every state except $x$, where
$(\PSi^{x,1} h)(x) = \E Z_{x,1} + \bbeta P_{x,1} h$; consequently
$h_\xi \to h^{\mathrm{res},1}$, the unique fixed point of $\PSi^{x,1}$,
the \emph{branch-1 resolvent}.  Symmetrically, as $\xi\to -\infty$,
$h_\xi \to h^{\mathrm{res},2}$, the fixed point of the mixed operator
$\PSi^{x,2}$ defined by $(\PSi^{x,2}h)(x) = \E Z_{x,2} + \bbeta P_{x,2}h$.
When state $x$ is the only state with tied optimal actions, the mixed
operators are affine and
\begin{equation}
h^{\mathrm{res},j} \;=\; \bigl(I - \bbeta P^{(x,j)}\bigr)^{-1} c^{(x,j)},
\label{eq:branch-resolvent}
\end{equation}
where $P^{(x,j)}$ is the transition matrix of the policy that plays
$u_{x,j}^\star$ at $x$ and the unique optimal action at every other
state, and $c^{(x,j)}_x = \E Z_{x,j}$, $c^{(x,j)}_y = c_y$ for $y\ne x$
(Corollary~\ref{cor:resolvent}).
\item[(e)] Exact tie: at $\xi = 0$ the switching operator reduces to the
tangential operator of Theorem~\ref{thm:main}, $h_0 = h^\star$, and
$p_1(0)$ is the exact-tie selection probability computed in
Theorem~\ref{thm:tiebreak} (Section~\ref{sec:tiebreak:main}).
\end{enumerate}
\end{theorem}

\begin{proof}
\emph{Step 1: the gap equation.}  Repeating the derivation of
Theorem~\ref{thm:main} (linearity of $u\mapsto Q_V(x,u)$ in $V$), the
rescaled gap $g_N^\xi = \veps_N^{-1}(V^\star - V_N^{\xi})$ satisfies
exactly the fixed-point equation
\begin{equation}
\begin{split}
g_N^\xi(x)
&= (\Phix_N^\xi g_N^\xi)(x),\\
(\Phix_N^\xi h)(x)
&= \E \min_{1\le i\le N}
\left\{ \veps_N^{-1}\Delta^{N,\xi}_x(U_{x,i})
+ \bbeta \sum_y p_{xy}(U_{x,i}) h(y) \right\},
\end{split}
\label{eq:PhiN-xi}
\end{equation}
where $\Delta^{N,\xi}_x$ is the deficiency \eqref{eq:switching-perturb}
of the parameterized family; the transitions in \eqref{eq:PhiN-xi} are
the (unchanged) unperturbed kernels, since the gap identity is linear
in $V$ and the perturbation enters $Q$ additively.  By
\eqref{eq:switching-perturb} and \eqref{eq:param-remainder},
$\veps_N^{-1}\Delta^{N,\xi}_x(u) =
\veps_N^{-1}\Delta_x(u) + \lambda_{x,\pi_x(u)}(\xi) + \rho_N(x,u,\xi)$,
the remainder uniform in $u$ and in $\xi$ on compact intervals, and
candidates away from the critical strata have rescaled deficiencies
$\to+\infty$ and do not affect the minimum.

\emph{Step 2: uniform convergence of the operators.}  Fix
$\xi\in K$ and apply the shifted-marks variant of Step~1 of the proof
of Lemma~\ref{lem:uniform} to the shifted values
$N^{1/\alpha_x}\Delta^{N,\xi}_x(u) + \bbeta P_{x,u}h$: by
\eqref{eq:switching-perturb} and \eqref{eq:param-remainder} the
rescaled deficiency of a candidate on stratum $j$ is
$N^{1/\alpha_x}\Delta_x(u) + \lambda_{xj}(\xi) + \rho_N(x,u,\xi)$ with
$\rho_N$ uniform, i.e.\ the marks are shifted by
$a_N(u) = \lambda_{x,\pi_x(u)}(\xi) + \rho_N(x,u,\xi)$, converging
uniformly to $a(u) = \lambda_{x,\pi_x(u)}(\xi)$, which is constant on
each stratum and hence continuous on $\mathsf{A}_x^\star$ (the strata
are separated).  The expectation step of Lemma~\ref{lem:uniform} (Step~2) applies
verbatim (bounded shifts on compact $K$), so the uniform integrability
is unchanged.  Hence, for every $R<\infty$ and every compact
$K\subset\R$,
\begin{equation}
\sup_{\norminfty{h}\le R}\;\sup_{\xi\in K}\;
\norminfty{\Phix_N^\xi h - \PSi_\xi h}
\;\longrightarrow\; 0,
\qquad N\to\infty,
\label{eq:uniform-xi}
\end{equation}
the shifts $\lambda_{xj}$ being bounded and $1$-Lipschitz on $K$; the uniformity in $\xi$ then follows from the $\delta$-net argument of Lemma~\ref{lem:uniform}, Step~3, applied jointly in $(h,\xi)$.

\emph{Step 3: the limit fixed point and the value gap.}  Both
$\Phix_N^\xi$ and $\PSi_\xi$ are $\bbeta$-contractions uniformly in $\xi$
(Lemma~\ref{lem:contraction} verbatim: the shifts are constant on the
strata), and $\sup_{\xi\in K}\norminfty{h_\xi}<\infty$ (iterate the
fixed-point equation from $0$).  From $g_N^\xi = \Phix_N^\xi g_N^\xi$ and $h_\xi = \PSi_\xi h_\xi$,
\begin{equation}
\norminfty{g_N^\xi - h_\xi}
\;\le\;
\frac{1}{1-\bbeta}\,
\norminfty{\Phix_N^\xi h_\xi - \PSi_\xi h_\xi}
\;\longrightarrow\; 0,
\label{eq:switching-bound}
\end{equation}
uniformly in $\xi\in K$ by \eqref{eq:uniform-xi}; this is (b), and
$h_0 = h^\star$ follows from $\lambda_{xj}(0) = 0$.

\emph{Step 4: the selector.}  Fix $\xi$.  Repeat Steps~1-2 of the proof
of Theorem~\ref{thm:tiebreak} with $\Delta^{N,\xi}_x$ in place of
$\Delta_x$: on the event that the winning candidate is critical, its
branch minimizes $\hat Z_{x,j} + \lambda_{xj}(\xi) + \bbeta P_{x,j}h_\xi$
(Lemma~\ref{lem:selector}: the offsets are deterministic shifts, the
limiting winner is a.s.\ unique), so $\Prob\{I_{x,N}^{\xi} =
1\}\to p_1(\xi)$.  Continuity of $p_1$ in $\xi$ follows from the
continuity of $\xi\mapsto h_\xi$ (bound
$\norminfty{h_\xi - h_{\xi'}} \le (1-\bbeta)^{-1}\norminfty{\PSi_\xi h_{\xi'} -
\PSi_{\xi'} h_{\xi'}}\to 0$ by dominated convergence), the
differences $Z_{x,1} - Z_{x,2} - \xi + \bbeta(P_{x,1}-P_{x,2})h_\xi$
having continuous distributions; and the endpoint limits
$p_1(\xi)\to 1,0$ as $\xi\to\pm\infty$ follow from Step~5 (the
branch-2 term of the minimum diverges).

\emph{Step 5: endpoint regimes.}  For $\xi\to+\infty$:
$\lambda_{x1}(\xi) = 0$ and $\lambda_{x2}(\xi) = \xi$.  For each state
and each $h$, dominated convergence (for $\xi\ge\xi_0$ the integrand is
bounded by $Z_{x,1}+Z_{x,2}+\xi_0+2\bbeta\norminfty{h}$, $\E
Z_{x,j}<\infty$) gives $\PSi_\xi h\to\PSi^{x,1}h$ pointwise, uniformly
on bounded balls of $h$ (equi-$\bbeta$-Lipschitz).  Let
$h^{\mathrm{res},1}$ be the unique fixed point of the mixed
$\bbeta$-contraction $\PSi^{x,1}$.  Then
\begin{equation}
\norminfty{h_\xi - h^{\mathrm{res},1}}
\;\le\;
\frac{1}{1-\bbeta}\,
\norminfty{\PSi_\xi h^{\mathrm{res},1} - \PSi^{x,1} h^{\mathrm{res},1}}
\;\longrightarrow\; 0,
\label{eq:endpoint-bound}
\end{equation}
and the case $\xi\to-\infty$ is symmetric, giving $h^{\mathrm{res},2}$.
Finally, if no state other than $x$ has tied optimal actions, $\PSi$ is
affine at every state $y\ne x$ by Corollary~\ref{cor:resolvent},
$\PSi h(y) = c_y + \bbeta P^\star_y h$, so $\PSi^{x,j}$ is globally
affine, $\PSi^{x,j}h = c^{(x,j)} + \bbeta P^{(x,j)}h$, and its fixed
point is the resolvent \eqref{eq:branch-resolvent}.

\emph{Step 6: comparison and monotonicity.}  On $[0,\infty)$,
$\lambda_{x1}(\xi) = 0$ and $\lambda_{x2}(\xi) = \xi$ is nondecreasing
in $\xi$, and $b\mapsto\min\{a,b\}$ is nondecreasing, so $\PSi_\xi h$
is nondecreasing in $\xi$ pointwise in $h$; the comparison principle
of Proposition~\ref{prop:contraction} (monotone $\bbeta$-contraction
on the sup-norm lattice) makes the fixed point $h_\xi$ nondecreasing
in $\xi$ on $[0,\infty)$, and the symmetric argument with
$\lambda_{x1}(\xi) = -\xi$ makes it nonincreasing on $(-\infty,0]$.
At the tie, $\lambda_{xj}(\xi)\ge\lambda_{xj}(0) = 0$ for all $j$ and
$\xi$, so $\PSi_\xi\ge\PSi$ pointwise and the monotone iteration
$h_\xi = \lim_m \PSi_\xi^m h^\star \ge h^\star$: $h_\xi$ is minimized
coordinatewise at $\xi = 0$, completing (a).
\end{proof}

\begin{remark}[The width of the layer]
The switching layer has width $\veps_N = N^{-1/\alpha_x}$ in the
branch-offset parameter $\theta$: at offsets $\theta_N = \theta_c +
\veps_N\xi$ with $|\xi|$ of order one,
\eqref{eq:switching-gap-limit}-\eqref{eq:switching-prob} show that
the finite-pool value interpolates between the two branch-specific
endpoint fixed points and the selection
probability between $0$ and $1$, through the
nonlinear family $h_\xi$.  In the sequential limit (send $N\to\infty$
at fixed $\xi$, then $|\xi|\to\infty$) part~(d) applies: gap and
selection converge to the winning branch's endpoint fixed point.  At a fixed
offset $\theta\ne\theta_c$ the perturbed model has a unique best
branch, and Theorem~\ref{thm:main} with Corollary~\ref{cor:resolvent}
apply directly: $\veps_N^{-1}(V^{\star,\theta}(x)-V_N^{\theta}(x))
\to h^{\mathrm{res}}(\theta)$, the endpoint constant of the winning
branch.
\end{remark}

\begin{proposition}[Parametric switching under a uniform local deficiency expansion]
\label{prop:param}
Let the model depend on a parameter $\theta$, with $Q^{\theta}$,
$V^{\star,\theta}$, and $\Delta_x^{\theta}(u) = V^{\star,\theta}(x) -
Q^{\theta}_{V^{\star,\theta}}(x,u)$ the $Q$-function, value, and
Bellman deficiency of the continuous-action model at $\theta$.
Assume that at the critical value $\theta_c$ the critical branches
of state $x$ are $\calJ_x^{\mathrm{crit}} = \{1,2\}$, tied with
minimal-deficiency gap
$\Delta_{x2}(\theta)-\Delta_{x1}(\theta) = \kappa_x(\theta-\theta_c)
+ o(|\theta-\theta_c|)$, $\kappa_x>0$, and that at the critical
offset $\theta_N = \theta_N(\xi) := \theta_c + \veps_N\xi$,
$\veps_N = N^{-1/\alpha_x}$, the model satisfies the
\emph{uniform local deficiency expansion at every state}: for every
compact set $K\subset\R$ and every finite $L$,
\begin{equation}
\sup_{\xi\in K}\;
\sup_{y\in\X}\;
\sup_{\{u\in\Ux_y\,:\,\Delta_y^{\theta_c}(u)\le L\veps_N\}}
\Bigl|\veps_N^{-1}\bigl(\Delta_y^{\theta_N}(u)
-\Delta_y^{\theta_c}(u)\bigr)
-\lambda_{y,\pi_y(u)}(\xi)\Bigr| \;\longrightarrow\; 0,
\label{eq:param-deficiency}
\end{equation}
with $\lambda_{x1}(\xi) = (-\kappa_x\xi)_+$, $\lambda_{x2}(\xi) =
(\kappa_x\xi)_+$ at the switching state $x$ and
$\lambda_{y,\cdot}(\xi)\equiv 0$ at $y\ne x$ (so that
$\min_j\lambda_{xj}(\xi) = 0$ and $\lambda_{x1}(\xi) -
\lambda_{x2}(\xi) = -\kappa_x\xi$), and a global control at every
state: for every compact $K\subset\R$ and every $y\in\X$ there is
$M_{K,y}<\infty$ with
$\sup_{\xi\in K}\sup_{u\in\Ux_y}|\veps_N^{-1}(\Delta_y^{\theta_N}(u)
-\Delta_y^{\theta_c}(u))| \le M_{K,y}$ for all large $N$ ($M_{K,x}$
at the switching state $x$).  The data are asymptotically Lipschitz
in $\xi$: for every compact $K$ there is $C_K<\infty$ with
$\sup_{y\in\X}\sup_{u\in\Ux_y}|\veps_N^{-1}(\Delta_y^{\theta_N(\xi)}(u)
-\Delta_y^{\theta_N(\xi')}(u))| \le C_K|\xi-\xi'|$ and likewise
$\sup_{y,z\in\X}\sup_{u\in\Ux_y}|p^{\theta_N(\xi)}_{yz}(u)
- p^{\theta_N(\xi')}_{yz}(u)| \le C_K|\xi-\xi'|$ for all $N$ and
$\xi,\xi'\in K$ (implied by the $C^1$ hypotheses of
Proposition~\ref{prop:param-smooth}).  Assume further
that the hypotheses of Theorem~\ref{thm:main} hold for the
triangular array $\{\theta_N = \theta_c + \veps_N\xi : N\ge 1,
\ \xi\in K\}$ uniformly in $\xi\in K$, with the common critical
exponent $\alpha_y \equiv \alpha$, in the following \emph{branchwise}
sense: at every state $y$ the local constants of the $\theta_c$ model
converge, $C_{yj}(\theta)\to C_{yj}$ and $r_{yj}(\theta)\to r_{yj}$
uniformly over the strata of $y$, and the local exponents are
constant near $\theta_c$, $\alpha_y(\theta) = \alpha_y \equiv \alpha$,
so $\veps_N$ stays the correct scaling at every state; only the branches of state $x$ switch, the other states'
optimal actions being unique and locally stable at $\theta_N$, and no
new optimal branches appear at $x$;
and the transition kernels converge uniformly on the shrinking
critical regions: for every finite $L$,
\begin{equation}
\sup_{\xi\in K}\;\sup_{y\in\X}\;
\sup_{\{u\in\Ux_y\,:\,\Delta_y^{\theta_c}(u)\le L\veps_N\}}
\bigl\| p^{\theta_N}_{y\cdot}(u) - p^{\theta_c}_{y\cdot}(u)
\bigr\|_\infty \;\longrightarrow\; 0,
\qquad N\to\infty.
\label{eq:param-kernel}
\end{equation}
In particular, evaluating \eqref{eq:param-deficiency} at the
$\theta_c$ optima $u^\star_{x,j}$, both branches of $x$ remain
$O(\veps_N)$-near-optimal at $\theta_N$,
$\veps_N^{-1}\Delta_x^{\theta_N}(u^\star_{x,j}) \to
\lambda_{xj}(\xi)$: the losing branch cannot be excluded at the
critical scale.  Then the rescaled value gap
measured against
the \emph{perturbed} continuous value converges uniformly on compact
intervals in $\xi$,
\begin{equation}
\veps_N^{-1}
\bigl(V^{\star,\theta_N}(x) - V_N^{\theta_N}(x)\bigr)
\;\longrightarrow\;
h_{\kappa_x\xi}(x),
\label{eq:param-gap}
\end{equation}
where $h_\xi$ is the fixed point of the switching operator of
Theorem~\ref{thm:switching} and $V_N^{\theta_N}$ is the value of the
sampled model at $\theta_N$.  Measured against the unperturbed value
instead, and assuming the value map $\theta\mapsto V^{\star,\theta}$
is Hadamard directionally differentiable at $\theta_c$,
$\veps_N^{-1}(V^{\star,\theta_c}(x) - V_N^{\theta_N}(x)) \to
h_{\kappa_x\xi}(x) - [DV^\star(\theta_c;\xi)](x)$, the Hadamard directional derivative of the value map at
$\theta_c$ in the direction $\xi$.
\end{proposition}

\begin{proof}
The baseline of Theorem~\ref{thm:switching}(b) is the unperturbed
value $V^\star$; here it is the perturbed value $V^{\star,\theta_N}$,
against which $\Delta^{\theta_N}$ is defined ($\Delta_x^{\theta_N}\ge 0$
with minimum zero at the winning branch, so the perturbed continuous
value is exactly $V^{\star,\theta_N}$).

\emph{Step 1: the branch minima.}
Write $\hat Z_{x,j}^{(N)}(\theta_N) := \veps_N^{-1}\min_{1\le
i\le N}\{\Delta_x^{\theta_N}(U_{x,i}) : \pi_x(U_{x,i}) =
u^\star_{x,j}\}$ for the rescaled minimal deficiency of branch $j$
among the sampled candidates at $\theta_N$, and $D_N^{(\xi)} :=
\min_{j=1,2}\hat Z_{x,j}^{(N)}(\theta_N)$ for the global rescaled
minimum (empty minima being $+\infty$).  By the branchwise hypotheses
the volume laws of Proposition~\ref{prop:poisson} hold uniformly in
$\theta$ near $\theta_c$; on the shrinking critical regions the
expansion \eqref{eq:param-deficiency} shifts the rescaled
deficiencies by $\lambda_{x,\pi_x(u)}(\xi)$ up to $o(1)$, and the
global control renders the candidates beyond them irrelevant at any
fixed level, so each $\hat Z_{x,j}^{(N)}(\theta_N)$ converges in
distribution to the \emph{shifted} Weibull variable
$Z_{x,j} + \lambda_{xj}(\xi)$.  (No separate estimate of the tail of
$D_N^{(\xi)}$ is needed: the uniform integrability of Step~2 below
absorbs it.)

\emph{Step 2: the untruncated sampled operator converges.}  For the
sampled model at $\theta_N$ introduce the one-step operator
\begin{equation}
\Phix^{\theta_N} h(y) \;:=\;
\E \min_{1\le i\le N}\Bigl\{
\veps_N^{-1}\Delta_y^{\theta_N}(U_{y,i})
+ \bbeta P^{\theta_N}_{y,U_{y,i}} h \Bigr\},
\qquad h\in C(\X),
\label{eq:param-operator}
\end{equation}
and note that the rescaled gap $g_N^{(\xi)}(y) :=
\veps_N^{-1}(V^{\star,\theta_N}(y) - V_N^{\theta_N}(y))$ is its
fixed point: from $V^{\star,\theta_N} - \max_i a_i = \min_i
(V^{\star,\theta_N} - a_i)$ the Bellman equations give
$V^{\star,\theta_N}(y) - V_N^{\theta_N}(y) = \E\min_{1\le i\le N}\{
\Delta_y^{\theta_N}(U_{y,i}) + \bbeta P^{\theta_N}_{y,U_{y,i}}
(V^{\star,\theta_N} - V_N^{\theta_N})\}$, and rescaling by
$\veps_N^{-1}$ yields $g_N^{(\xi)} = \Phix^{\theta_N}
g_N^{(\xi)}$.  We claim that, for every $R<\infty$ and every compact $K\subset\R$,
\begin{equation}
\sup_{\xi\in K}\;\sup_{\norminfty{h}\le R}\;
\norminfty{\Phix^{\theta_N}h - \PSi_{\kappa_x\xi}h} \;\longrightarrow\;
0,
\qquad N\to\infty.
\label{eq:param-opconv}
\end{equation}
Fix $y\in\X$, $h$ with $\norminfty{h}\le R$, and $\xi\in K$, and write
$m_{N,\xi}(h)(y)$ for the minimum inside the expectation in
\eqref{eq:param-operator} (the untruncated sampled minimum), which we
show is uniformly integrable and converges in distribution, uniformly
in $\xi$ and $h$.

\emph{Distributional convergence.}  Fix $L<\infty$.  On the shrinking
critical region $\{\Delta_y^{\theta_c}\le L\veps_N\}$ the
expansion \eqref{eq:param-deficiency} is the uniform one-step
expansion \eqref{eq:switching-perturb} of the switching argument,
with remainder uniform in $y$, in $u$, and in $\xi\in K$ (the
all-state expansion \eqref{eq:param-deficiency}, whose offsets
$\lambda_{y,\cdot}(\xi)\equiv 0$ vanish at $y\ne x$); and
\eqref{eq:param-kernel} freezes the rescaled future terms
$\bbeta P^{\theta_N}_{y,\cdot}h$ at their $\theta_c$ limits.  By the
global control, candidates with
$\Delta_y^{\theta_c} > (L+M_{K,y}+\bbeta R)\veps_N$ have
rescaled cost at least $L$ ($\bbeta P h\ge -\bbeta R$) and cannot
affect the minimum of
$\min\{L, m_{N,\xi}(h)(y)\}$.  The shifted-marks argument of
Steps~1-2 of Theorem~\ref{thm:switching} applies verbatim to the
triangular array $(N,\theta_N)$: for every fixed $L$,
$\min\{L, m_{N,\xi}(h)(y)\} \Rightarrow \min\{L, m_\xi(h)(y)\}$,
where
\begin{equation}
m_\xi(h)(y) \;:=\; \min_{(z,u)\in\Pi_y}\bigl\{ z +
\lambda_{y,\pi_y(u)}(\xi) + \bbeta P^{\theta_c}_{y,u}h \bigr\},
\label{eq:param-poisson-min}
\end{equation}
with $\lambda_{x,\pi_x(u)}(\xi)$ at the switching state $x$,
$\lambda\equiv 0$ elsewhere, and $\Pi_y$ the marked Poisson process
of state $y$ (Lemma~\ref{lem:uniform}); hence
$\PSi_{\kappa_x\xi}h(y) = \E\, m_\xi(h)(y)$.

\emph{Uniform integrability.}  By the global control at state $y$,
$\veps_N^{-1}\Delta_y^{\theta_N}(u)
\le \veps_N^{-1}\Delta_y^{\theta_c}(u) + M_{K,y}$ for every
$u$, while $\bbeta P h \ge -\bbeta R$, so
\begin{equation}
-\bbeta R \;\le\; m_{N,\xi}(h)(y) \;\le\;
\veps_N^{-1}D^{\theta_c}_{y,1:N} + M_{K,y} + \bbeta R,
\label{eq:param-ui-bound}
\end{equation}
where $D^{\theta_c}_{y,1:N} := \min_{1\le i\le N}
\Delta_y^{\theta_c}(U_{y,i})$ (the right inequality evaluates the
minimum at the candidate attaining it).  By Lemma~\ref{lem:ui-min}
(with $h\equiv 0$) the family
$\{\veps_N^{-1}D^{\theta_c}_{y,1:N}\}$ has uniformly bounded
$(1+\eta)$-moments, uniformly in $y$; hence
$\{m_{N,\xi}(h)(y)\}$ is uniformly integrable, uniformly in $(\xi,h)$,
and sending $L\to\infty$ in the distributional statement (the
truncated limits are integrable by the same bound) gives
$\Phix^{\theta_N}h(y)\to\PSi_{\kappa_x\xi}h(y)$ at fixed
$(y,h,\xi)$.

\emph{Uniform upgrade.}  The maps $h\mapsto\Phix^{\theta_N}h$ are
$\bbeta$-Lipschitz uniformly in $N$ and $\xi$ (minima of
$\bbeta$-Lipschitz maps), and $\xi\mapsto\Phix^{\theta_N}h$ is
$C_K$-Lipschitz uniformly in $N$ and in $\norminfty{h}\le R$ by the
asymptotic Lipschitz hypothesis, the limit operators sharing the
same Lipschitz bounds (the shifts $\lambda_{xj}$ are
$\kappa_x$-Lipschitz on $K$).  A finite $\delta$-net in $(h,\xi)$
together with the pointwise convergence just proved and the two
Lipschitz bounds gives
$\sup_{\xi\in K}\sup_{\norminfty{h}\le R}\norminfty{\Phix^{\theta_N}h
- \PSi_{\kappa_x\xi}h} \le O(\delta) + o(1)$, hence
\eqref{eq:param-opconv} as $\delta\to 0$.

\emph{Step 3: contraction closing.}  Both $\Phix^{\theta_N}$ and
$\PSi_{\kappa_x\xi}$ are $\bbeta$-contractions on $C(\X)$ (minima of
$\bbeta$-Lipschitz maps), so the fixed-point identities and the
contraction inequality give
$\norminfty{g_N^{(\xi)} - h_{\kappa_x\xi}} \le
(1-\bbeta)^{-1}\norminfty{\Phix^{\theta_N}h_{\kappa_x\xi}
- \PSi_{\kappa_x\xi}h_{\kappa_x\xi}} \to 0$, uniformly in $\xi\in K$
by \eqref{eq:param-opconv}; this is \eqref{eq:param-gap}.  The
comparison with $V^{\star,\theta_c}$ is the Hadamard directional
sensitivity of the value at the tie, $V^{\star,\theta_N} -
V^{\star,\theta_c} = \veps_N\, DV^\star(\theta_c;\xi) +
o(\veps_N)$ in $C(\X)$, and subtracting the two limits gives
the claim.
\end{proof}

\begin{proposition}[Smooth parametric perturbations under two-sided
transversality]
\label{prop:param-smooth}
Let the data be smooth in $\theta$ and $u$: the partial derivatives
$\partial_\theta r^{\theta}(y,u)$ and $\partial_\theta
p^{\theta}_{yz}(u)$ are jointly continuous in $(\theta,u)$ and
uniformly bounded in $y,z\in\X$, $u\in\Ux_y$, and the maps are
locally $C^2$ in $u$, with Hessians
$\nabla^2_u r^{\theta}(y,u)$, $\nabla^2_u p^{\theta}_{yz}(u)$
continuous in $(\theta,u)$; all actions in $\mathsf{A}^\star_y$ are
interior to $\Ux_y$; and the sampling densities are continuous and
strictly positive near each $u\in\mathsf{A}_y^\star$
(Assumption~\ref{ass:G}
with $\eta_{yj} = 0$).  Suppose that at $\theta_c$ the optimal set of
state $x$ consists of two points $u^\star_{x,1}$, $u^\star_{x,2}$
with positive definite Hessians
$\nabla^2_u\Delta_x^{\theta_c}(u^\star_{x,j})$, $j=1,2$; that every
other state has a unique optimal action $u^\star_y$ with positive
definite Hessian; that the minimal-deficiency gap is nondegenerate,
$\Delta_{x2}(\theta)-\Delta_{x1}(\theta) = \kappa_x(\theta-\theta_c)
+ o(|\theta-\theta_c|)$, $\kappa_x>0$; and that the deficiencies have
a uniform positive gap off the critical region,
$\Delta_y^{\theta}(u)\ge g_\delta$ whenever
$\Delta_y^{\theta_c}(u)\ge\delta$, uniformly in $\theta$ near
$\theta_c$.  Define the \emph{raw parameter perturbation}
\begin{equation}
b_y(u) \;=\; \partial_\theta r^{\theta}(y,u)\big|_{\theta_c}
\;+\; \bbeta \sum_z \partial_\theta p^{\theta}_{yz}(u)\big|_{\theta_c}
\, V^{\star,\theta_c}(z),
\label{eq:raw-perturb}
\end{equation}
and let $d_\xi = DV^\star(\theta_c;\xi)$ be the unique fixed point of
the \emph{directional Bellman operator} on $C(\X)$,
\begin{equation}
d \;\longmapsto\; \max_{u\in\mathsf{A}_y^\star}\bigl\{ \xi\, b_y(u)
+ \bbeta P^{\theta_c}_{y,u}\, d \bigr\}, \qquad y\in\X,
\label{eq:dir-bellman}
\end{equation}
where $\mathsf{A}_y^\star$ is the optimal action set at $\theta_c$
(a $\bbeta$-contraction, the maximum of $\bbeta$-Lipschitz maps).
Then the standing common-exponent assumption $\alpha_y\equiv\alpha$
of Proposition~\ref{prop:param} (critical rate $\veps_N =
N^{-1/\alpha}$) makes the hypotheses of Proposition~\ref{prop:param}
hold: \eqref{eq:param-deficiency} holds with the branch offsets
\begin{equation}
\lambda_{xj}(\xi) \;=\; d_\xi(x) - \bigl[ \xi\, b_x(u^\star_{x,j})
+ \bbeta P_{x,j}\, d_\xi \bigr],
\label{eq:branch-offset}
\end{equation}
which satisfy $\lambda_{xj}(\xi)\ge 0$ and $\min_j\lambda_{xj}(\xi)
= 0$ and reduce, by the nondegenerate gap, to
$\lambda_{x1}(\xi) = (-\kappa_x\xi)_+$, $\lambda_{x2}(\xi) =
(\kappa_x\xi)_+$.  The common slope $\kappa_x$ is an explicit
two-sided transversality assumption, not implied by $C^1$ data alone.
The global control and the asymptotic Lipschitz
condition hold uniformly on compact subsets of $\R$; the branchwise
uniformity and \eqref{eq:param-kernel} hold; only the branches of
state $x$ switch, the others remaining unique and locally stable at
$\theta_N$; and the value map is Hadamard directionally
differentiable at $\theta_c$ with directional derivative $d_\xi$.
\end{proposition}
\begin{proof}
\emph{The directional derivative.}  The directional Bellman operator
\eqref{eq:dir-bellman} is a $\bbeta$-contraction on $C(\X)$, so
$d_\xi$ is well defined, and
$\xi\mapsto d_\xi$ is Lipschitz on compact sets.  It is the Hadamard
directional derivative of
the value map: the positive definite Hessians and the uniform gap
exclude maximizers of $u\mapsto Q^{\theta}_{V^{\star,\theta_c}}(y,u)$
beyond a shrinking neighborhood of the $\theta_c$-optimal set
$\mathsf{A}_y^\star$, and
there the $C^1$ expansion of the smooth data together with the
continuity of $b_y$ give, uniformly in $y$ and in $\xi$ on compact
sets,
\begin{equation}
T^{\theta_c+\varepsilon\xi}\bigl(V^{\star,\theta_c}
+ \varepsilon d_\xi\bigr)(y)
\;=\; V^{\star,\theta_c}(y) + \varepsilon\, d_\xi(y) + o(\varepsilon),
\label{eq:dir-approx}
\end{equation}
with $T^{\theta}V := \max_u Q^{\theta}_V(\cdot,u)$ the perturbed
Bellman operator; thus $V^{\star,\theta_c}+\varepsilon d_\xi$ is an
$o(\varepsilon)$-approximate fixed point of the $\bbeta$-contraction
$T^{\theta_c+\varepsilon\xi}$, and the comparison principle gives
$\norminfty{V^{\star,\theta_c+\varepsilon\xi} - V^{\star,\theta_c}
- \varepsilon d_\xi} = o(\varepsilon)$, uniformly in $\xi$ on compact
sets.

\emph{The branch offsets.}  On the shrinking critical regions
$\{\Delta_y^{\theta_c}\le L\veps_N\}$ of all states, expanding
the smooth data about $\theta_c$ and substituting the directional
expansion of the value map gives, uniformly in $y$, in $u$, and in
$\xi\in K$,
\begin{equation}
\Delta_y^{\theta_N}(u) \;=\; \Delta_y^{\theta_c}(u)
\;+\; \veps_N\bigl[ d_\xi(y) - \xi\, b_y(u)
- \bbeta P^{\theta_c}_{y,u}\, d_\xi \bigr] + o(\veps_N),
\label{eq:param-expanded}
\end{equation}
which is \eqref{eq:param-deficiency}: the bracket is continuous in
$u$; at $y\ne x$ the unique optimum collapses the directions,
$u\to u_y^\star$ on the shrinking region, and the directional Bellman
equation at $u_y^\star$ makes the limit offset vanish,
$d_\xi(y) = \xi b_y(u_y^\star) + \bbeta P_{y,u_y^\star}d_\xi$
($\lambda_{y,\cdot}(\xi)\equiv 0$); at $y=x$ it converges to
$\lambda_{x,\pi_x(u)}(\xi)$ of \eqref{eq:branch-offset}.  By the Bellman equation for $d_\xi$ at $x$, $\lambda_{xj}(\xi)\ge 0$
with $\min_j\lambda_{xj}(\xi) = 0$; and differentiating the gap
identity $\Delta_{x2}(\theta)-\Delta_{x1}(\theta) =
Q^{\theta}_{u^\star_{x,1}}(V^{\star,\theta}) -
Q^{\theta}_{u^\star_{x,2}}(V^{\star,\theta})$ at $\theta_c$ (using
the directional expansion of the value map) gives the gap slope
$\kappa_x\xi = \xi[b_x(u^\star_{x,1}) - b_x(u^\star_{x,2})] +
\bbeta(P_{x,1}-P_{x,2})d_\xi = -(\lambda_{x1}(\xi)-\lambda_{x2}(\xi))$
hence the stated $\lambda_{x1}(\xi) = (-\kappa_x\xi)_+$,
$\lambda_{x2}(\xi) = (\kappa_x\xi)_+$.  The global control follows
from the local Lipschitz continuity of $V^{\star,\theta}$ in $\theta$
(a consequence of the comparison principle applied to the
$\bbeta$-contractions $T^\theta$), which gives
$\sup_u|\Delta_y^\theta(u)-\Delta_y^{\theta'}(u)| \le
C|\theta-\theta'|$ and hence $\veps_N^{-1}\Delta_y^{\theta_N}(u)
\le \veps_N^{-1}\Delta_y^{\theta_c}(u) + M_{K,y}$ with $M_{K,y}
= C\sup_{\xi\in K}|\xi|$; the asymptotic Lipschitz condition in
$\xi$ is the same uniform bound.  The branchwise uniformity holds because the quadratic lower bound
$\Delta_y^{\theta_c}(u) \ge c\,\dist(u,\mathsf{A}_y^\star)^2$
(uniform in $y$, $\theta$ near $\theta_c$, from the positive definite
Hessians) makes $\Delta_y^{\theta_c}(u)\le L\veps_N$ imply
$\dist(u,\mathsf{A}_y^\star) = O(\veps_N^{1/2})$: standard
stability of nondegenerate local minima yields locally unique
branches $u_{y,j}(\theta)\to u^\star_{y,j}$; continuity and positive
definiteness of the Hessians then imply $r_{yj}(\theta)\equiv 2$ and
$C_{yj}(\theta)\to C_{yj}$, the exponents
$\alpha_{yj} = m_{yj}/r_{yj}$ constant near $\theta_c$;
\eqref{eq:param-kernel} is
the uniform bound $|\veps_N\xi|\sup|\partial_\theta p|\to 0$.
The nondegenerate Hessians persist near $\theta_c$, and the uniform gap
keeps the $\theta_N$-optimal set of $x$ in the shrinking tube around
$\{u^\star_{x,1},u^\star_{x,2}\}$, so no new branches appear and the
other optimal actions stay unique and locally stable at $\theta_N$.
\end{proof}

\section{Propagation of the Slowest Critical Scale under Heterogeneous Exponents}
\label{sec:hetero}

Theorem~\ref{thm:main} rests on the common-exponent condition
\eqref{eq:common-alpha}.  This section shows the condition is not
merely technical: when the critical exponents differ across states, the
value gap is governed at the slowest global scale by the states with
the \emph{largest} local exponents, propagated through the discounted
reachability of the limit policy; which states dominate depends on
reachability, not on a single global worst-case exponent.  Finer
layers require additional transition-order assumptions.  (The two uses of ``critical'' are opposite:
the \emph{strata} minimize $\alpha_{xj}$ and host the rescaled minimum, the \emph{class} maximizing
$\alpha_x$.)  Beyond the critical scale the
picture is subtler (Remark~\ref{rem:layering}).

\subsection{The Critical Class and the Resolvent}
\label{sec:hetero:critical}

Throughout this section every state is assumed to have a \emph{unique}
optimal action $u_x^\star$ (so that the tangential operator is affine at
each state, cf.\ Corollary~\ref{cor:resolvent}), and write
\begin{equation}
a_x \;=\; \frac{1}{\alpha_x},
\qquad
a_0 \;=\; \min_{x\in\X} a_x
\;=\; \frac{1}{\alpha_{\max}},
\qquad
\calC \;=\; \bigl\{ x\in\X : a_x = a_0 \bigr\},
\label{eq:critical-class}
\end{equation}
where $\alpha_{\max} = \max_x \alpha_x$: the states in $\calC$ have the
\emph{largest} local exponents, their value gaps decaying at the
slowest rate $N^{-a_0}$ and carrying the leading order.  Indeed the
local contribution of state $x$ alone is $c_x N^{-a_x} = c_x
N^{-1/\alpha_x}$, with $c_x = C_x^{-1/\alpha_x}\Gamma(1+1/\alpha_x)$;
for $x\notin\calC$ this is $o(N^{-a_0})$, but the gap at $x$ is not
the local gap alone: the slow states feed $x$ through the dynamics.

\begin{theorem}[Propagation of the critical class under heterogeneous exponents]
\label{thm:hetero}
Assume the hypotheses of Proposition~\ref{prop:contraction},
Assumption~\ref{ass:G}, and that the optimal action $u_x^\star$ is
unique for every state.  Let $P^\star_{xy} = p_{xy}(u_x^\star)$ be the
transition matrix of the limit policy and
$R^\star = (I-\bbeta P^\star)^{-1} = \sum_{t\ge0}\bbeta^t (P^\star)^t$
its resolvent.  Define the \emph{local saturation}
\begin{equation}
b_N(x) \;=\; \E\Bigl[ \min_{1\le i\le N} \Delta_x(U_{x,i}) \Bigr],
\qquad
b_N(x) \;=\; c_x N^{-a_x} + o\bigl(N^{-a_x}\bigr),
\label{eq:local-expansion}
\end{equation}
with the same constant $c_x$ (the volume law of
Proposition~\ref{prop:volume} and the uniform integrability of
Lemma~\ref{lem:ui-min}).  Then:
\begin{enumerate}
\item[(i)] \emph{(Saturation at the critical rate.)}  With
$c^{(0)}_x = c_x \mathbf 1\{a_x = a_0\}$,
\begin{equation}
N^{a_0}\bigl(V^\star(x) - V_N(x)\bigr)
\;\longrightarrow\;
\sum_{y\in\calC} R^\star_{xy}\, c_y
\;=\; \bigl(R^\star c^{(0)}\bigr)(x)
\qquad\text{componentwise.}
\label{eq:hetero-limit}
\end{equation}
The $x$-component of the limit,
$\sum_{t=0}^{\infty}\bbeta^t\sum_{y\in\calC}(P^{\star\,t})_{xy}\,c_y$,
is positive if and only if at least one critical state $y\in\calC$ is
reachable from $x$ under the limit policy with positive probability.
\item[(ii)] \emph{(Resolvent identity.)}  The value gap satisfies,
componentwise,
\begin{equation}
V^\star(x) - V_N(x)
\;=\;
\sum_y R^\star_{xy}\, b_N(y) + o\bigl(N^{-a_0}\bigr):
\label{eq:resolvent-identity}
\end{equation}
the local saturation law \eqref{eq:local-expansion} propagates through
the discounted transition chain of the limit policy.
\end{enumerate}
\end{theorem}

Part~(ii) implies part~(i) (multiply by $N^{a_0}$ and use
\eqref{eq:local-expansion}).

\begin{proof}
\emph{Step 1: the normalized gap equation and the a priori bound.}
Repeating the derivation of Theorem~\ref{thm:main} (linearity of
$u\mapsto Q_V(x,u)$ in $V$), the value gap $w_N = V^\star - V_N$
satisfies exactly $w_N(x) = \E \min_{1\le i\le N}\{\Delta_x(U_{x,i})
+ \bbeta P_{x,U_{x,i}} w_N\}$.
Normalize at the critical scale $N^{-a_0}$: $h_N = N^{a_0} w_N$
satisfies $h_N = \Phix_N^{\mathrm{cr}} h_N$, where
$(\Phix_N^{\mathrm{cr}} h)(x) = \E \min_{1\le i\le N}\{N^{a_0}\Delta_x(U_{x,i})
+ \bbeta P_{x,U_{x,i}} h\}$.  Since $\min_i\{a_i + b_i\} \le \min_i a_i
+ \max_i b_i$, $\norminfty{h_N} \le \bbeta\norminfty{h_N} +
N^{a_0}\norminfty{b_N}$; and $N^{a_0} b_N(x) \to c_x \mathbf 1\{x\in\calC\}$
by \eqref{eq:local-expansion}, so $\sup_N \norminfty{h_N} < \infty$.

\emph{Step 2: the limit operator.}  For bounded $h$, uniformly in
$\norminfty{h}\le R$ for each $R$, $\Phix_N^{\mathrm{cr}} h \to
\PSi^{\mathrm{cr}} h$, where $(\PSi^{\mathrm{cr}} h)(x) = c_x
\mathbf 1\{x\in\calC\} + \bbeta \sum_y P^\star_{xy} h(y)$.  For
$x\in\calC$ the normalization is the common-exponent one ($N^{a_0} =
N^{1/\alpha_x}$).  The tail calculation of Lemma~\ref{lem:uniform}
applies verbatim to the rescaled minima $\min_{1\le i\le N}\{N^{a_0}\Delta_x(U_{x,i})
+ \bbeta P_{x,U_{x,i}}h\}$, whose limit is
$\E\min_{(z,u)\in\Pi_x}\{z + \bbeta P_{x,u}h\}$.  Because the optimal
action is unique, the marks collapse to
$u_x^\star$ (proof of Corollary~\ref{cor:resolvent}), the minimum
factorizes, and
this limit equals $\E[\min_{(z,u)\in\Pi_x} z] + \bbeta P^\star_x h =
c_x + \bbeta P^\star_x h$, with $c_x =
C_x^{-1/\alpha_x}\Gamma(1+1/\alpha_x)$ the Weibull mean
\eqref{eq:orderstat-moment}.  For $x\notin\calC$, $a_x > a_0$ \eqref{eq:critical-class}, and for
every $\varepsilon>0$, $\Prob\{\min_{1\le i\le N} N^{a_0}\Delta_x(U_{x,i})
> \varepsilon\} = (1 - F_x(\varepsilon N^{-a_0}))^N \le \exp\{-N
F_x(\varepsilon N^{-a_0})\} \to 0$, since $N F_x(\varepsilon N^{-a_0})
\sim C_x \varepsilon^{\alpha_x} N^{1-a_0\alpha_x}\to\infty$
($a_0\alpha_x < 1$).  Fix $h$ with $\norminfty{h}\le R$ and let
$M_N = \min_{1\le i\le N}\{N^{a_0}\Delta_x(U_{x,i}) + \bbeta
P_{x,U_{x,i}}h\}$.  By this tail bound, $M_N \le \bbeta R + o_P(1)$,
while any candidate with $N^{a_0}\Delta_x(U_{x,i})\ge L_N$ (slowly
increasing $L_N\to\infty$) has rescaled term at least $L_N - \bbeta
R\to\infty$ and is not the minimizer with probability tending to one.
The winning candidate thus satisfies $\Delta_x \le L_N N^{-a_0}\to 0$,
hence $u\to u_x^\star$ (unique optimum, continuous $\Delta_x$), and
$P_{x,u}h \to P^\star_x h$ uniformly.  Therefore
$M_N \to \bbeta P^\star_x h$ in probability, and convergence of
expectations follows from the two-sided bound
$-\bbeta R\le M_N \le \bbeta R + N^{a_0}\min_i\Delta_x(U_{x,i})$,
whose upper-envelope expectation is
$\bbeta R + N^{a_0}b_N(x) \to \bbeta R$ by \eqref{eq:local-expansion}
($a_x > a_0$).  Both families are equi-$\bbeta$-Lipschitz in $h$,
which upgrades the pointwise convergence to the uniformity in
$\norminfty{h}\le R$.

\emph{Step 3: the fixed point.}  $\Phix_N^{\mathrm{cr}}$ and
$\PSi^{\mathrm{cr}}$ are $\bbeta$-contractions uniformly in $N$ (the
proof of Lemma~\ref{lem:contraction} applies verbatim), and
$\PSi^{\mathrm{cr}}$ is affine: $\PSi^{\mathrm{cr}}h = c^{(0)} +
\bbeta P^\star h$, whose fixed point is $h = R^\star c^{(0)}$.  From $h_N = \Phix_N^{\mathrm{cr}}h_N$
and
$\norminfty{h_N - h}\le (1-\bbeta)^{-1}\norminfty{\Phix_N^{\mathrm{cr}}h - \PSi^{\mathrm{cr}}h}$,
Step~2 gives $\norminfty{h_N - h}\to 0$, which is (i); the positivity
criterion is reachability of $\calC$ under $P^\star$ (an entry
$R^\star_{xy}$ is positive exactly when $y$ is reachable from $x$).

\emph{Step 4: the resolvent identity.}  Combining the exact gap
equation with the local saturation, the residual
\begin{equation}
\begin{split}
\veps_N(x)
&= \E\min_i\Bigl\{\Delta_x(U_i) + \bbeta P_{U_i}w_N\Bigr\}
- b_N(x) - \bbeta P^\star_x w_N\\
&= N^{-a_0}\Bigl[ \Phix_N^{\mathrm{cr}}h_N
- N^{a_0}b_N - \bbeta P^\star h_N \Bigr](x)
\end{split}
\label{eq:hetero-eps}
\end{equation}
satisfies $\veps_N(x) = N^{-a_0} o(1)$: the bracket converges to zero
because $N^{a_0}b_N\to c^{(0)}$ (Step~1),
$\Phix_N^{\mathrm{cr}}h_N - \PSi^{\mathrm{cr}}h_N \to 0$ (Steps~2-3
with $\norminfty{h_N}$ bounded), and
$\PSi^{\mathrm{cr}}h_N - c^{(0)} - \bbeta P^\star h_N = 0$ by
definition.  Applying $R^\star = (I-\bbeta P^\star)^{-1}$ to
$(I-\bbeta P^\star)w_N = b_N + \veps_N$ gives (ii).
\end{proof}

\begin{remark}[Near-optimal leaks and layered rates]
\label{rem:layering}
Theorems (i)-(ii) give the leading order at the critical scale
$N^{-a_0}$ only; a \emph{layered} expansion beyond it needs
hypotheses on the transitions.  Whether a state outside the critical
class saturates at its own local rate depends on the vanishing order
of the transitions of its near-optimal actions: order $\kappa_{xyj}$
near $u_{xj}^\star$ leaks candidates of state $x$ into $y$ with
probability of order $\rho^{\kappa_{xyj}}$, shifting the exponent of
$x$ from $a_x$ to $a_y + \kappa_{xyj}/(m_{xj}+\eta_{xj})$
(the two-state example of the Supplement evaluates this shift,
$2\to 3/4$; a heuristic, not a theorem).  Action-independent
transitions ($\kappa = \infty$) give exact layering.
\end{remark}

\begin{remark}[The resolvent identity]
The resolvent identity \eqref{eq:resolvent-identity} is the
heterogeneous analogue of the affine tangential identity
\eqref{eq:Psi-affine}: at the critical scale $N^{-a_0}$ the
near-optimal layer of a noncritical state shrinks toward $u_x^\star$
(its saturation $c_xN^{-a_x}$ invisible at that scale), so the
winning candidates of \emph{every} state act through $P^\star$, each
state contributing only its critical mass $c_x\mathbf 1\{x\in\calC\}$.
With multiple optimal actions, the corresponding heterogeneous
nonlinear tangential system requires additional analysis and is left
open; Section~6 describes its common-scale analogue.
\end{remark}

\section{Numerical Experiments}
\label{sec:numerics}

This section validates the main results experimentally.  All
Monte Carlo estimates use $K\ge 10^4$ independent replications, pools
drawn i.i.d.\ from the stated measures, and fixed points computed by
value iteration to machine precision.  Empirical scaling exponents
(Example~1) come from a log-log regression, with
$t_{0.975,5}$-scaled delta-method intervals from the fitted slopes;
tables, figures, and derivations are in the Supplement.

\subsection{Example 1: Static Geometric Selection}
\label{sec:numerics:ex1}

A single state with action space $\Ux =
[-\tfrac14,\tfrac14]\cup[\tfrac34,\tfrac54]$ sampled uniformly with
unit density (both optima are interior to the support, as the
full-support hypothesis of Proposition~\ref{prop:contraction}
requires).  This static example (no discounting) illustrates
Proposition~\ref{prop:orderstats}.  The
$Q$-function has two separated optimal actions $u_1 = 0$ and $u_2 =
1$, with deficiency $\Delta(u) = \min\{c_1 |u-u_1|^{r_1},\; c_2
|u-u_2|^{r_2}\}$; for a scalar stratum the volume constant is $C_j =
2c_j^{-1/r_j}$ and the critical exponent $\alpha_j = 1/r_j$.
\emph{Equal exponents} $(r_1,r_2)=(2,2)$, $c_1 = 4c_2$: $\alpha =
1/2$, $C_1 = 1$, $C_2 = 2$, and the branch-2 selection frequency
converges to $C_2/(C_1+C_2) = 2/3$ (Proposition~\ref{prop:orderstats}).
\emph{Different exponents} $(r_1,r_2)=(2,4)$, $c_1 = c_2 = 1$:
$\alpha_1 = 1/2 > \alpha_2 = 1/4$, and branch 2 (flatter deficiency,
larger near-optimal volume) is selected with probability one, the gap
scaling with $\alpha = 1/4$.  Empirical frequencies converge to these
limits, the gap exponents confirming $\alpha = 1/r$.

\subsection{Example 2: Same Geometry, Different Transitions, Dynamic
Reversal}
\label{sec:numerics:ex2}

A two-state MDP with $\bbeta = 0.95$.  The model data are
$r_x(u) = -\Delta_x(u)$ and kernels $p_{xy}(u)$, so $V^\star \equiv
0$; the value deficiency of $u$ is then the $\Delta_x(u)$ below.  State 1 has two optimal actions
$u_{1,1}^\star, u_{1,2}^\star$ with \emph{identical} local geometry
(deficiency $|u-u_{1,j}^\star|$ for both, so $C_1 = C_2 = 2$ and
$\alpha_1 = \alpha_2 = 1$) that differ only in their transitions: the
kernel is branchwise constant on the two supports, $p_{12}(u) = 1$,
$p_{11}(u) = 0$ on $[-\tfrac14,\tfrac14]$ ($u_{1,1}^\star = 0$ moves
to state 2 w.p.~1) and $p_{12}(u) = 0$, $p_{11}(u) = 1$
on $[\tfrac34,\tfrac54]$ ($u_{1,2}^\star = 1$), hence
continuous in a tubular neighborhood of each optimum.  State 2 has a
unique optimum with the very flat deficiency $0.01\,|u - u_2^\star|$,
so $C_2' = 200$ and the fixed-policy gap is $h^\star_2 = 0.1$ ($c_2 =
\E Z_2 = 1/C_2'$, the Weibull mean); its kernel self-loops,
$p_{22}(u) = 1$.  Sampling is uniform of unit density on
$[-\tfrac14,\tfrac14]\cup[\tfrac34,\tfrac54]$ at state~1
($\Delta_1(u) = \min(|u|,|u-1|)$) and on $[-\tfrac12,\tfrac12]$ at
state~2 ($\Delta_2(u) = 0.01\,|u|$).

The state-1 geometry is perfectly symmetric, so the volume ratios
predict the $1/2$-$1/2$ selection of
Proposition~\ref{prop:orderstats}; the dynamic theory of
Theorem~\ref{thm:tiebreak} predicts otherwise, the fixed point
$h^\star$ of the nonlinear tangential operator satisfying
\begin{equation}
h^\star_1 = \E \min\{ Z_1 + \bbeta h^\star_2,\; Z_2 + \bbeta h^\star_1 \},
\label{eq:ex2-fixed}
\end{equation}
where $Z_1, Z_2$ are independent Weibull($\alpha=1$) variables with
tail $e^{-2z}$; solving \eqref{eq:ex2-fixed} with $h^\star_2 = 0.1$
gives $h^\star = (0.4716, 0.100)$, and the selection probability of
$u_{1,1}^\star$ is, by \eqref{eq:selection-prob}, $p_{1,1} =
\Prob\{Z_1 - Z_2 < \bbeta(h^\star_1 - h^\star_2)\} = 1 -
\tfrac12 e^{-2\cdot 0.95\cdot 0.3716} \approx 0.753$.
The pool thus \emph{reverses} the static prediction, the finite-$N$
frequencies converging to the dynamic limit $0.753$ rather than $1/2$.
The finite-$N$ gap is strictly below \emph{both} resolvent
predictions \eqref{eq:resolvent}: the pool beats either fixed policy
by switching (nonlinearity at the leading order).

\begin{figure}[ht]
\centering
\includegraphics[width=366.9pt]{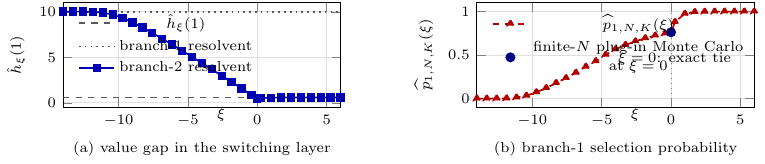}
\caption{Example~3: the critical switching layer of
Theorem~\ref{thm:switching}.  (a)~The rescaled state-1 gap
$\hat h_\xi(1)$ (blue, solid) interpolates between the branch
resolvents $10.0$ and $0.595$ (gray, dashed), dipping near the
prediction $h^\star_1 = 0.4716$ at $\xi = 0$.  (b)~The empirical selection probability
$\widehat p_{1,N,K}(\xi)$ (red, dashed) rises from $0$
to $1$, with $\widehat p_{1,N,K}(0) = 0.763$ against the prediction
$p_1(0) = 0.753$ ($K = 10^4$ pools of $N = 10^4$ candidates).}
\label{fig:ex3}
\end{figure}

\subsection{Example 3: The Critical Switching Layer}
\label{sec:numerics:ex3}
\label{sec:ex3}

A final experiment tests the switching-layer predictions of
Theorem~\ref{thm:switching} on the two-state model of Example~2.
The branches have identical geometry ($C_1 = C_2 = 2$,
$\alpha = 1$) but different transitions: branch~1 exits to state~2
with the constant gap $h^\star_2 = 0.1$, while branch~2
self-loops at state~1; the exact-tie predictions are $h^\star_1 =
0.4716$ and $p_{1,1} = 0.753$ (Theorem~\ref{thm:switching}(e)).

The perturbed family \eqref{eq:switching-perturb} is solved with
deficiency shifts $\veps_N\lambda_j(\xi)$
($\lambda_1 = (-\xi)_+$, $\lambda_2 = (\xi)_+$, exact tie at $\xi=0$).
With $\veps_N = N^{-1}$, $K = 10^4$ pools of $N = 10^4$
candidates are drawn, the rescaled branch minima $\hat Z_{x,j} =
\veps_N^{-1}\min_i\Delta_j(U_{x,i})$ extracted, and the
switching fixed point \eqref{eq:Psi-xi} iterated to machine precision
at $29$ values of $\xi\in[-14,6]$ (plus the exact tie).

Figure~\ref{fig:ex3} shows the result.  Panel (a) displays the rescaled
state-1 gap $\hat h_\xi(1)$ across the layer, interpolating between the
resolvents and dipping below \emph{both} at the exact tie: $\hat h_0(1) =
0.470 \pm 0.007$ vs.\ the prediction $h^\star_1 = 0.4716$.  The dynamic
selection at the tie strictly dominates each fixed branch policy.
Panel (b) confirms the selection transition: $\widehat p_{1,N,K}(0) =
0.763$ against the prediction $p_1(0) = 0.753$ ($2.3$ standard errors,
combining finite-$N$ and Monte Carlo effects).  The layer is asymmetric, sharp on the branch-1 side
($\widehat p_{1,N,K} \ge 0.998$ by $\xi = 2.4$) and broad on the
branch-2 side; part~(c) guarantees continuity and the
endpoints $p_1\to 0,1$; monotonicity is observed but proved
only under a shared transition operator at $x$.  A
finite-$N$ check (Supplement) gives
$\max_\xi|g_N^\xi(1)-h_\xi(1)| = 7.3, 1.5, 2.2 \times 10^{-3}$ at
$N=10^2,10^3,10^4$.

\section{Conclusion}
\label{sec:conclusion}

This paper develops an exact first-order asymptotic theory for the
value and selection rule of a discounted MDP drawing, at each state
visit, a random pool of $N$ candidate actions.  Under a
geometric regularity condition the analysis reduces the finite-pool
problem to a marked Poisson limit: the near-optimal candidates
converge to a point process carrying the transition kernels, and the
leading order of the value gap is the fixed point of a nonlinear
tangential Bellman operator, a stochastic generalization of the
resolvent.  The same
limit governs the asymptotic selection rule: for isolated competing
optimal actions the selector is genuinely dynamic at exact ties and
interpolates between branch regimes, admitting the Mecke
representation \eqref{eq:mecke} on general compact optimal sets; in the
unique-optimum heterogeneous case the critical class is propagated at
the slowest global scale through the discounted reachability of the
limit policy, finer layers requiring additional transition-order
assumptions.  Numerical
experiments confirm the predicted rates,
constants, and selection probabilities; future directions include a
fixed-point theory for the heterogeneous tangential system without the
common-exponent condition, boundary and corner cases excluded by the
regularity assumption, a lifetime-fixed pool (a random Bellman
operator), a layered theory of the vanishing orders of near-optimal
transitions, and importance-sampling designs.

\bibliographystyle{siamplain}
\bibliography{refs}

\end{document}


\title{Supplementary Materials for ``Poisson Tangent Limits and Critical
Policy Switching for Sampled Bellman Operators''
\thanks{Supplementary material to the paper of the same name.
Unprefixed numbers (Theorem~6.5, equation~(4.4), Assumption~3.1) refer
to the main paper; statements and equations introduced here carry the
prefix S: Lemma~S.1, Lemma~S.2, and equations~(S.1)-(S.15).}}
\author{Ming-Zhe Dai\thanks{School of Automation, Central South University,
Changsha 410083, China (\email{mingzhe\_dai@csu.edu.cn}).}
\and
Chengxi Zhang\thanks{School of IoT Engineering, Jiangnan University, Wuxi,
Jiangsu 214000, China (\email{dongfangxy@163.com}). Corresponding author.}}
\maketitle

\begin{abstract}
This document collects the material supporting the main paper: the
complete numerical output of the three experiments of Section~8 (all
tables and the figures of Examples~1 and 2), the implementation details
for reproducibility, and the routine derivations of the closed-form
predictions used in the experiments and in Remark~7.2.  It also
contains the proofs of the supporting results relied on by the main
text: the topological details of the marked Poisson convergence of
Proposition~4.1, the uniform-integrability estimates used in
Proposition~4.2 and in the proof of Lemma~5.2 (Lemmas~S.1 and S.2), and
the tubular-coordinate identities behind the local volume law (3.7)
and its marked refinement (4.1).
\end{abstract}

\section{Extended Numerical Details}
\label{app:numerics}

This section reports the complete numerical output of the experiments
of Section~8 of the main paper.  All values are Monte Carlo estimates
with $K$ independent replications as stated; candidate pools are drawn
i.i.d.\ from the stated sampling measures, by three independent
implementations (one per example) of the generative models described
in Section~8 of the main paper, with implementation details in
Section~\ref{app:impl} below.

\subsection{Example 1: Static Geometric Selection}
\label{app:ex1}

The model is that of Section~8.1 of the main paper: a single state,
$\bbeta = 0$, $\Ux = [-\tfrac14,\tfrac14]\cup[\tfrac34,\tfrac54]$
(the sampling support), deficiency
$\Delta(u) = \min\{c_1|u-u_1|^{r_1}, c_2|u-u_2|^{r_2}\}$ at the two
optima $u_1 = 0$, $u_2 = 1$, sampled from the density-one uniform
measure on $[-1/4,1/4]\cup[3/4,5/4]$.  For a single scalar stratum
the local volume constant is the elementary integral
$C_j = \int \one\{ c_j|w|^{r_j} \le 1\}\,dw = 2 c_j^{-1/r_j}$, the
critical exponent is $\alpha_j = 1/r_j$, and the value gap constant is
$C^{-1/\alpha}\Gamma(1+1/\alpha)$ with $C = C_1 + C_2$ by
Proposition~4.2 of the main paper.

\emph{Equal exponents} $(r_1,r_2) = (2,2)$, $c_1 = 4$, $c_2 = 1$:
$C_1 = 2\cdot 4^{-1/2} = 1$, $C_2 = 2$, $C = 3$, $\alpha = 1/2$,
branch-2 frequency $C_2/C = 2/3$, and the rescaled gap satisfies
$N^2\,\E\min_i\Delta(U_i) \to C^{-1/\alpha}\Gamma(1+1/\alpha)
= 3^{-2}\Gamma(3) = 2/9$.

\emph{Different exponents} $(r_1,r_2) = (2,4)$, $c_1 = c_2 = 1$:
$\alpha_1 = 1/2$, $\alpha_2 = 1/4$ (critical), branch-2 frequency
tends to $1$, and $N^4\E\min_i\Delta(U_i) \to
C_2^{-4}\Gamma(5) = 24/16 = 3/2$.

Each of the seven pool sizes uses $K = 2\times10^4$ independent
replications (the Weibull QQ plot of Figure~\ref{app:fig:ex1-weibull}
uses $K = 10^5$).  Across the grid
$N\in\{10,30,10^2,3\cdot10^2,10^3,3\cdot10^3,10^4\}$:
for equal exponents the empirical branch-2 frequencies remain close to
$2/3$; for different exponents they approach $1$ monotonically from
below, from $0.923$ at $N = 10$ to $0.9999$ at $N = 10^4$.  The
rescaled gaps track their limits
closely: $N^2\E\min_i\Delta(U_i)$ rises from $0.171$ at $N = 10$
to $0.2179$ at $N = 10^4$ against the prediction $2/9 \approx
0.2222$, and $N^4\E\min_i\Delta(U_i)$ oscillates around the
prediction $3/2$ within $6\%$ for every $N \ge 3\cdot10^2$.

Figure~\ref{app:fig:ex1-branch} reports the empirical branch-2
selection frequencies against the two theoretical limits, and
Figure~\ref{app:fig:ex1-weibull} gives the Weibull QQ plot of the
rescaled minimal deficiency at $N = 10^4$ ($K = 10^5$ replications)
against the limit tail $e^{-Cz^\alpha} = e^{-3z^{1/2}}$ of
Proposition~4.2.  Table~\ref{app:tab:ex1} gives the values at $N = 10^4$: the
frequencies match $2/3$ and $1$ to three decimals, and the fitted
exponents $\widehat\alpha = 0.507 \pm 0.008$ and $\widehat\alpha =
0.261 \pm 0.010$ ($95\%$ $t_5$-intervals of the slope fits, with
$\mathrm{se} = 0.0031$ and $0.0039$) lie $2.3$ and $2.8$ standard
errors from the predictions $1/2$ and $1/4$: the first contains the
prediction within its interval, the second sits $1.1$ interval
half-widths from $1/4$ and outside it.  Refitting on the four largest
pool sizes $N \ge 300$ gives $\widehat\alpha = 0.4996$ and $0.2490$
with standard errors $0.0007$ and $0.0013$, within one standard error
of the predictions; the excess of the second full-grid fit is
therefore the finite-$N$ bias of the fitted slope at the small pool
sizes (Section~\ref{app:impl}).

\begin{figure}[ht]
\centering
\includegraphics[width=366.9pt]{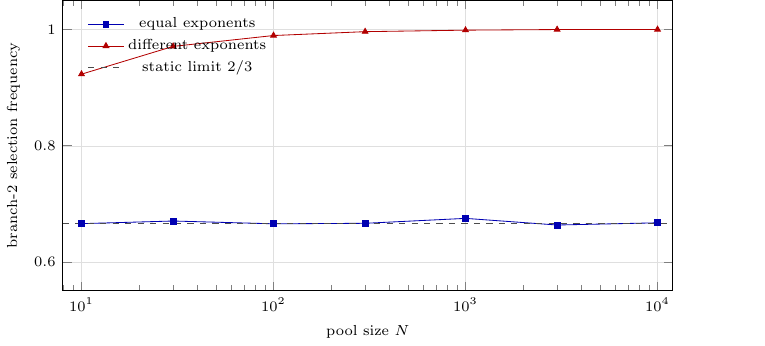}
\caption{Example~1: branch-2 selection frequencies versus the pool size
$N$ (logarithmic abscissa).  Squares: equal exponents
$(r_1,r_2)=(2,2)$, $c_1=4c_2$, converging to the static volume ratio
$C_2/(C_1+C_2) = 2/3$ (dashed); triangles: different exponents
$(r_1,r_2)=(2,4)$, converging to $1$.}
\label{app:fig:ex1-branch}
\end{figure}

\begin{figure}[ht]
\centering
\includegraphics[width=366.9pt]{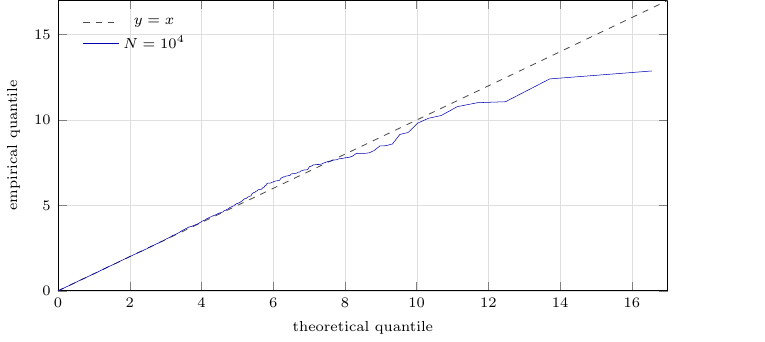}
\caption{Example~1, equal exponents: QQ plot of the rescaled minimal
deficiency $N^2\min_i\Delta(U_i)$ at $N = 10^4$ ($K = 10^5$
replications) against the Weibull law with tail
$e^{-Cz^\alpha}$, $C = 3$, $\alpha = 1/2$
(Proposition~4.2 of the main paper).}
\label{app:fig:ex1-weibull}
\end{figure}

\begin{table}[ht]
\centering
\caption{Example~1: empirical branch-2 selection frequencies at
$N = 10^4$ and value-gap scaling exponents fitted over the seven pool sizes.}
\input{tables_py/tab_ex1.tex}
\label{app:tab:ex1}
\end{table}

\subsection{Example 2: Same Geometry, Different Transitions, Dynamic
Reversal}
\label{app:ex2}

The model is that of Section~8.2: a two-state MDP with $\bbeta =
0.95$; state~1 has two optimal actions with identical geometry
($C_1 = C_2 = 2$, $\alpha = 1$), action $u_{1,1}^\star$ transferring
to state~2 with probability one and $u_{1,2}^\star$ self-looping; state
2 has the unique optimum with deficiency $0.01|u-u_2^\star|$, so
$C_2' = 200$ and $h^\star_2 = (C_2')^{-1}\Gamma(2)/(1-\bbeta) = 0.1$.
Sampling is uniform of unit density: on
$[-\tfrac14,\tfrac14]\cup[\tfrac34,\tfrac54]$ at state~1, with
$u_{1,1}^\star = 0$ and $u_{1,2}^\star = 1$ interior points of the
support and deficiency $\Delta_1(u) = \min(|u|,|u-1|)$, and on
$[-\tfrac12,\tfrac12]$ at state~2 with $u_2^\star = 0$ and
$\Delta_2(u) = 0.01|u|$.

The tangential fixed point at state 1 solves the one-dimensional
equation
\begin{equation}
h^\star_1 \;=\; \E\min\{Z_1 + \bbeta h^\star_2,\; Z_2 + \bbeta h^\star_1\},
\label{app:eq:ex2-fixed}
\end{equation}
with $Z_1, Z_2$ independent with tail $e^{-2z}$ (scale $C_j^{-1}
= 1/2$).  Iterating \eqref{app:eq:ex2-fixed} from $0$ converges to
$h^\star = (0.4716, 0.100)$.  The selection probability of
$u_{1,1}^\star$ is, by equation (6.7) of the main paper,
\begin{equation}
p_{1,1} \;=\; \Prob\{Z_1 - Z_2 < \bbeta(h^\star_1 - h^\star_2)\}
\;=\; 1 - \tfrac12 e^{-2\bbeta(h^\star_1 - h^\star_2)}
\;\approx\; 0.753,
\label{app:eq:ex2-p}
\end{equation}
using the identity $\Prob\{Z_1 - Z_2 < c\} = 1 -
\frac12 e^{-2c}$ for $c \ge 0$ and independent rate-2 exponentials
(derived in the Routine Derivations section; the evaluation at
$h^\star_1 - h^\star_2 = 0.3716$ gives $1 - \tfrac12 e^{-0.7060}
\approx 0.753$).

Figure~\ref{app:fig:ex2} reports the finite-$N$ selection frequency of
action $u_{1,1}^\star$: it converges to the dynamic limit $0.753$, far
from the static volume prediction $1/2$.  Each grid point uses
$K = 2\times10^4$ independent replications.  Across the whole grid
$N \in \{10,\dots,10^4\}$ the empirical frequency stays within four
standard errors of the limit (the largest deviation, $3.2$ standard
errors, occurs at the smallest pool size $N = 10$), while the static
prediction $1/2$ is $80$ standard errors away.  The rescaled gaps track the fixed point
\eqref{app:eq:ex2-fixed}: state 1 within $1\%$, state 2 within $2\%$,
for $N \ge 10^2$ ($N\cdot(V^\star(1) - V_N(1))$ lies between $0.4673$
and $0.4707$ against $h^\star_1 = 0.4716$, and
$N\cdot(V^\star(2) - V_N(2))$ stays within $2\%$ of $h^\star_2 = 0.1$,
reaching $0.10013$ at $N = 10^4$).  Table~\ref{app:tab:ex2-grid} gives the
complete grid.  Table~\ref{app:tab:ex2}
compares the leading gap with the linear resolvent predictions of
Corollary~5.5 of the main paper for the two fixed policies, $0.595$
(action $u_{1,1}^\star$, transferring to state 2) and $10.0$ (action
$u_{1,2}^\star$, self-loop): the finite-$N$ gap ($0.4677$ at
$N = 10^4$) is strictly below \emph{both}, i.e.\ the pool beats either
fixed policy by switching.

\begin{figure}[ht]
\centering
\includegraphics[width=366.9pt]{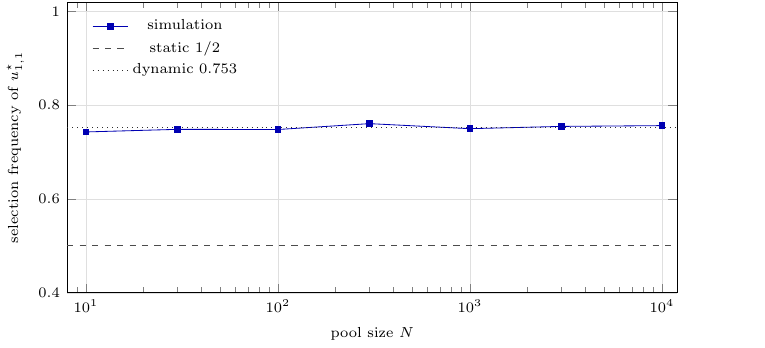}
\caption{Example~2: selection frequency of action $u_{1,1}^\star$
(transferring to the easy state 2) versus $N$ (squares), together with
the static volume prediction $1/2$ (dashed) and the dynamic limit
$p_{1,1} \approx 0.753$ of Theorem~6.5 (dotted).  The finite-$N$
frequencies approach the dynamic limit, demonstrating the dynamic
reversal induced by the transition kernels through the nonlinear fixed
point \eqref{app:eq:ex2-fixed}.}
\label{app:fig:ex2}
\end{figure}

\begin{table}[ht]
\centering
\caption{Example~2: static volume weight, nonlinear tangential fixed
point $h^\star$, finite-$N$ branch frequency at $N = 10^4$, and
resolvent predictions.}
\input{tables_py/tab_ex2.tex}
\label{app:tab:ex2}
\end{table}

\begin{table}[ht]
\centering
\caption{Example~2, complete grid.}
\input{tables_py/tab_ex2_grid.tex}
\label{app:tab:ex2-grid}
\end{table}

\subsection{Example 3: The Critical Switching Layer}
\label{app:ex3}

The experiment of Section~8.3 draws $K = 10^4$ pools of $N = 10^4$
candidates, extracts the two rescaled branch minima
$\hat Z_{x,j} = \veps_N^{-1}\min_i\Delta_j(U_{x,i})$ with
$\veps_N = N^{-1}$, and iterates the switching fixed point
$h_\xi = \PSi_\xi h_\xi$, with $\PSi_\xi$ as in (6.23) of the main
paper, for $29$ values of
$\xi\in[-14,6]$ plus the exact tie $\xi = 0$, all on the \emph{same}
sample of branch minima (so the Monte Carlo error is common across the
layer).  The resulting empirical selection probability
$\widehat p_{1,N,K}(\xi)$ is a continuous
sigmoid with $\widehat p_{1,N,K}(0) = 0.763$ against the prediction
$p_1(0) = 0.753$ of
Theorem~6.5 and $\widehat p_{1,N,K} \ge 0.998$ by $\xi = 2.4$; the
rescaled state-1
gap $\hat h_\xi(1)$ interpolates between the branch resolvents $10.0$
and $0.595$,
dipping below \emph{both} at the exact tie ($\hat h_0(1) = 0.470 \pm
0.007$
against $h^\star_1 = 0.4716$).  The layer is asymmetric: sharp on the
branch-1 side, broad on the branch-2 side where the self-loop branch is
expensive.  The observed
monotonicity of $\widehat p_{1,N,K}$ is a property of this example,
asserted in
general only when the branches share a transition operator at $x$
(Theorem~6.9(c) of the main paper).

With $g_N^\xi(1) = N(V^\star - V_N^\xi)(1)$ the rescaled state-1
gap, the finite-$N$ behavior is reported in
Table~\ref{app:tab:ex3-finiten} at pool sizes
$N = 10^2, 10^3, 10^4$ ($K = 10^5$ pools, seed
\texttt{20260809}, standard errors in parentheses) against the exact
fixed point $h_\xi(1)$: the bias
$\max_\xi |g_N^\xi(1) - h_\xi(1)|$ is $7.3\times10^{-3}$,
$1.5\times10^{-3}$, and $2.2\times10^{-3}$ at $N = 10^2, 10^3,
10^4$, respectively; the latter two are at the Monte Carlo resolution
(about $2\times10^{-3}$, with standard errors about
$1.5\times10^{-3}$).  The data are consistent with polynomial
convergence to the limit but do not resolve the exact rate.

\begin{table}[ht]
\centering
\caption{Example~3: finite-$N$ convergence of the switching layer.}
\footnotesize
\input{tables_py/tab_ex3_finiten.tex}
\label{app:tab:ex3-finiten}
\end{table}

\section{Implementation Details}
\label{app:impl}

All simulations were implemented in Python using NumPy (PCG64,
NumPy $\ge 1.17$).  The model definitions, parameter values, pool
sizes, replication counts, numerical tolerances, and random seed are
specified below so that the reported results can be independently
replicated within Monte Carlo error.  All runs use the PCG64 generator
with seed 20260809.  No external datasets are used; all numerical
results are generated from the fully specified models.

Candidate pools are drawn in chunks (4096 replications at a time in
Example~1) to bound memory use; the replication count is $K \ge 10^4$
for every reported value ($K = 10^5$ for the Weibull QQ plot of
Example~1 and for the finite-$N$ switching check of Example~3,
Table~\ref{app:tab:ex3-finiten}); pool sizes run from $N = 10$ to $10^4$ in Examples~1 and
2; Example~3 uses $N = 10^4$, plus the convergence check
$N \in \{10^2,10^3,10^4\}$ of Table~\ref{app:tab:ex3-finiten}.

Value fixed points are computed by value iteration to machine
precision (tolerance $10^{-12}$ on the sup-norm change); for the
dynamic rules of Examples~2 and 3 the fixed point of the tangential
equation is iterated likewise.  The value-gap scaling
exponents are read off from a log-log regression of $\log_{10}$
gap on $\log_{10} N$ over the seven pool sizes $N \in \{10, 30,
10^2, 3\cdot10^2, 10^3, 3\cdot10^3, 10^4\}$: the exponent is
$\widehat\alpha = -1/\widehat s$ for the fitted slope $\widehat s$,
and the reported spreads are the $t_{0.975,5}$-scaled delta-method
diagnostic intervals obtained from the fitted slopes,
$\widehat\alpha \pm t_{0.975,5}\,\mathrm{se}(\widehat\alpha)$ with
$t_{0.975,5}\approx 2.571$.

\section{Routine Derivations of Closed-Form Predictions}
\label{app:derivations}

\subsection{The exponential difference identity}
For independent rate-2 exponentials $Z_1, Z_2$ (tail $e^{-2z}$) and
$c \ge 0$, the convolution
$\Prob\{Z_1 - Z_2 < c\} = \int_0^\infty 2e^{-2z_2}(1 - e^{-2(z_2+c)})\,dz_2
= 1 - \frac12 e^{-2c}$ is the evaluation used in
\eqref{app:eq:ex2-p}.

\subsection{The near-optimal leak of Remark~7.2}

Remark~7.2 of the main paper states the mechanism; the full example
follows.  Two states with a scalar action space,
$\mu_x = \mathrm{Unif}[-1,1]$, $\bbeta = 0.9$, $V^\star = 0$, and
\begin{equation}
\Delta_1(u) = |u|^{1/2},\quad P_{11}(u) = 1,\quad
\Delta_2(u) = u^2,\quad P_{21}(u) = |u|^{1/4},\quad
P_{22}(u) = 1 - |u|^{1/4},
\end{equation}
with rewards $r_x(u) = -\Delta_x(u)$, satisfying the hypotheses of
Theorem~7.1 (unique interior optima, homogeneous deficiencies); the
local exponents are $a_1 = 1/2$, $a_2 = 2$, the limit policy has
$P^\star = I$, so state~2 cannot reach the critical state~1, and yet
does \emph{not} saturate at order $N^{-2}$.

Writing $D_N = \min_i|U_i|$ and $w_N = V^\star - V_N$, the gap
equation at state 2 is solved, to leading order, by the balance
\begin{equation}
w_2 \;=\;
\frac{\E D_N^2 + \bbeta \E D_N^{1/4}\, w_1}
{1 - \bbeta + \bbeta \E D_N^{1/4}}
\;\sim\;
\frac{\bbeta\,\Gamma(5/4)\,\Gamma(3/2)}
{(1-\bbeta)^2}\; N^{-3/4},
\label{app:eq:leak-rate}
\end{equation}
because $\E D_N^s \sim \Gamma(s+1)N^{-s}$: the term
$u^2 + \bbeta|u|^{1/4}(w_1 - w_2)$ is strictly increasing in $|u|$
since $w_1 > w_2$, so the winner is $D_N$ at distance scale $N^{-1}$,
where the local saturation $\E D_N^2 \sim 2N^{-2}$ is dwarfed by the
leak $\bbeta\E D_N^{1/4}w_1$ (with $w_1 \sim
\Gamma(3/2)(1-\bbeta)^{-1}N^{-1/2}$ from state 1's gap equation), the
denominator contributing a further $(1-\bbeta)^{-1}$.
The transition vanishing order $\kappa_{21} = 1/4$ with
$m_{21}+\eta_{21} = 1$ shifts the state-$2$ exponent from its local
value $2$ to $e_1 + 1/4 = 3/4$.

The exact evaluation at $N = 10^3, 10^4, 10^5$
($\E D_N^s = \Gamma(s+1)\Gamma(N+1)/\Gamma(N+s+1)$, Beta moments)
gives $N^{3/4}w_2 = 29.5, 39.8, 49.6$, approaching
$\bbeta\Gamma(5/4)\Gamma(3/2)(1-\bbeta)^{-2} \approx 72.3$.

\section{Topological Details of the Marked Poisson Convergence}

This section makes precise the topological statements of
Proposition~4.1 and the truncation step of Lemma~5.2.

\subsection{The Space of Locally Finite Counting Measures}
\label{app:topology:space}

Let $S = [0,\infty)\times\mathsf{A}_x^\star$ with the product of the
Euclidean topology on $[0,\infty)$ and the relative topology on the
compact metric space $\mathsf{A}_x^\star$.  A Borel measure $\xi$ on
$S$ is \emph{locally finite} if $\xi(K)<\infty$ for every compact
$K\subseteq S$.  Since every compact subset of $[0,\infty)$ is
contained in $[0,L]$ for some $L<\infty$ and $\mathsf{A}_x^\star$ is
compact, a measure on $S$ is locally finite iff it is finite on every
set $[0,L]\times \mathsf{A}_x^\star$, $L<\infty$; the endpoint $z=0$
must be included here, since a locally finite measure on
$(0,\infty)\times\mathsf{A}_x^\star$ need not be finite on
$(0,L]\times\mathsf{A}_x^\star$.  Let
$\mathbf{M}(S)$ denote the set of
locally finite counting measures on $S$ (integer-valued, with finite
mass on compact sets), and equip it with the \emph{vague topology}: the
coarsest topology under which $\xi\mapsto\int_S f\,d\xi$ is continuous
for every continuous $f$ with compact support.  The space
$\mathbf{M}(S)$ is Polish, and convergence in the vague topology is
equivalent to convergence of $\xi_N(K)\to\xi(K)$ for every relatively
compact Borel set $K$ whose boundary satisfies
$\xi(\partial K)=0$ (portmanteau).  The intensity measures
$\Lambda_x$ of (4.4) are locally finite: indeed
$\Lambda_x([0,L]\times B) = L^{\alpha_x}\sigma_x(B)<\infty$, and
$\Lambda_x(\{0\}\times\mathsf{A}_x^\star) = 0$, so the limit process
has no atom at $z=0$.

\subsection{Kallenberg's Criterion}
\label{app:topology:kallenberg}

A sequence of point processes $\Pi_N$ on $S$ converges in distribution
to a Poisson process $\Pi$ with intensity $\Lambda$ if and only if, for
every continuous nonnegative $f$ with compact support,
\begin{equation}
\E\exp\Bigl\{-\langle f,\Pi_N\rangle\Bigr\}
\;\longrightarrow\;
\exp\Bigl\{-\int_S\bigl(1-e^{-f(s)}\bigr)\Lambda(ds)\Bigr\}.
\label{eq:kallenberg}
\end{equation}
Since $\Lambda$ is diffuse in the $z$-coordinate, the limit process
is $\Prob$-almost surely simple with distinct $z$-coordinates; this is
used in the continuous mapping argument of Lemma~5.2 of the main
paper.

\subsection{The Laplace Functional Calculation}
\label{app:topology:laplace}

The calculation (4.6)-(4.7) requires
two justifications beyond the pointwise convergence
\eqref{eq:local-volume-A} of this Supplement.  First, the support of $f$ is contained in
$[0,L]\times\mathsf{A}_x^\star$, and the change of variables
$z = N^{1/\alpha_x}t$ maps the region $\{z\le L\}$ to
$t\le LN^{-1/\alpha_x}$, on which the marked volume law (4.1),
equivalently its localized form (S.14), applies
uniformly (the contribution of $t>LN^{-1/\alpha_x}$ is zero because
$f = 0$ there).  Second, the Laplace functional of a finite Poisson
process on $[0,L]\times\mathsf{A}_x^\star$ is
$\exp\{-\int(1-e^{-f})d\Lambda\}$ by independence of the points, and
the passage from the finite truncated process to the full process
$\Pi_x$ uses the fact that $\Lambda_x$ is finite on the truncation and
$\sigma$-finite on $S$, the standard extension theorem for Poisson
processes on $\sigma$-finite spaces giving uniqueness of the limiting
process.

\subsection{The Nearest-Point Projection and Marked Continuity}
\label{app:topology:projection}

For each stratum $\calM_{xj}$ and its tubular neighborhood, the
nearest-point projection $\pi_{xj}(u)$ is the unique point of
$\calM_{xj}$ with $\dist(u,\calM_{xj}) = \norm{u - \pi_{xj}(u)}$, well
defined and continuous on the tube (the $C^2$ compact submanifold has a
uniform reach, so the projection is single-valued inside the tube).
The global projection $\pi_x: \Ux_x\to\mathsf{A}_x^\star$ is defined by
concatenating the $\pi_{xj}$ and is measurable.  The sample minimum
$D_{x,1:N} = \min_{1\le i\le N}\Delta_x(U_{x,i})$ converges to $0$ in
probability: for every small fixed $t>0$ the volume law (4.1) gives
$F_x(t) := \Prob(\Delta_x(U)\le t) > 0$, so
$\Prob(D_{x,1:N}>t) = (1-F_x(t))^N \to 0$ as $N\to\infty$.  Write
$U_N^\star$ for the action chosen by the sampled model at $x$: with
$w_N := V^\star - V_N$, it minimizes the completed cost
$u\mapsto\Delta_x(u) + \bbeta P_{x,u}(V^\star - V_N) = \Delta_x(u) +
\bbeta P_{x,u}w_N$ over the sampled pool.  Every candidate $U_i$ in
the pool satisfies
\begin{equation*}
\Delta_x(U_N^\star) \;\le\; \Delta_x(U_i) + \bbeta P_{x,U_i}w_N
- \bbeta P_{x,U_N^\star}w_N \;\le\; D_{x,1:N} + 2\bbeta\norminfty{w_N},
\end{equation*}
the first inequality from the minimality of the completed cost at
$U_N^\star$ and the second from the bounds on $P_{x,\cdot}w_N$; take
$U_i$ attaining $D_{x,1:N}$.  Since $\norminfty{w_N}\to 0$ by
Proposition~2.1 of the main paper, the chosen candidate lies in
$\{\Delta_x\le t'\}$ with probability tending to one; taking $t'$
small enough that $\{\Delta_x\le t'\}$ is covered by the tubes (the
tubular-coordinate section of this Supplement), candidates outside
the tubes have probability tending to zero, and the choice of
$\pi_x$ outside the tubes is therefore immaterial for every limit in
the paper.

The marked process carries the \emph{transition operator} $P_{x,U}$
as its mark, a point of the compact metric space of operators on
$\R^n$.  By uniform continuity of $u\mapsto p_{xy}(u)$ on the compact
space $\Ux_x$, the mark of a candidate converges to the mark of its
projection with error $o(1)$ uniformly on the critical region
$\{\Delta_x\le t\}$; this is the quantitative content of the passage
from $\Pi_{x,N}$ to the kernel-marked process in the proof of
Lemma~5.2 of the main paper.

\section{Uniform Integrability and Truncation}

This section collects the uniform integrability argument used in
Proposition~4.2 (moment convergence (4.10)) and in
Lemma~5.2 (convergence of expectations (5.14)).

\begin{lemmasupp}[Uniform integrability of rescaled order statistics]
\label{lem:ui}
Assume Assumption~3.1.  For each $x\in\X$ and each fixed
$j\ge 1$, the family of random variables
\begin{equation}
\bigl\{ N^{1/\alpha_x} D_{x,j:N} \bigr\}_{N\ge j}
\label{eq:ui-family}
\end{equation}
is uniformly integrable.
\end{lemmasupp}

\begin{proof}
By Proposition~3.2 and the two-sided bound
(3.13), there exist $0<c<C<\infty$ and $t_0>0$ such that
$c\,t^{\alpha_x}\le F_x(t) := \Prob\{\Delta_x(U)\le t\}\le
C\,t^{\alpha_x}$ for $0<t\le t_0$.  Fix $x_0\in(0,t_0)$ and set
$z_N := N^{1/\alpha_x}x_0$.  For
$z\le z_N$ (so that $t = zN^{-1/\alpha_x}\le x_0\le t_0$),
\begin{equation}
\Prob\bigl(N^{1/\alpha_x}D_{x,1:N} > z\bigr)
\;=\;
\left[1 - \Prob(\Delta_x(U)\le zN^{-1/\alpha_x})\right]^N
\;\le\;
e^{-c z^{\alpha_x}}.
\label{eq:ui-tail}
\end{equation}
For $z>z_N$, use $D_{x,1:N}>zN^{-1/\alpha_x} \ge x_0 \Rightarrow$ all
candidates have $\Delta_x(U)>x_0$:
\begin{equation}
\Prob\bigl(N^{1/\alpha_x}D_{x,1:N} > z\bigr)
\;\le\;
\Prob(D_{x,1:N} > x_0)
\;=\;
\left[1 - \Prob(\Delta_x(U)\le x_0)\right]^N
\;\le\;
e^{-c' N},
\label{eq:ui-tail2}
\end{equation}
where $c' = \Prob(\Delta_x(U)\le x_0) > 0$ (the volume law is
positive at $x_0$).  Moreover $\Delta_x$ is bounded on the compact
space $\Ux_x$ (it is continuous), say $\Delta_x\le \Delta_{\max}$, so
$D_{x,1:N}\le\Delta_{\max}$ and the tail in
\eqref{eq:ui-tail2} vanishes for
$z > \bar z_N := N^{1/\alpha_x}\Delta_{\max}$.  Hence, for every
$a>0$, writing $Y_N := N^{1/\alpha_x}D_{x,1:N}$ and using the
layer-cake formula
$\E[Y_N\one\{Y_N>a\}] = a\Prob(Y_N>a) + \int_a^\infty
\Prob(Y_N>z)\,dz$,
\begin{align}
& \E\Bigl[ N^{1/\alpha_x}D_{x,1:N}
\, \one\bigl\{N^{1/\alpha_x}D_{x,1:N} > a\bigr\} \Bigr] \notag\\
&\le\;
a e^{-c a^{\alpha_x}}
+ \int_a^\infty e^{-c z^{\alpha_x}}\, dz
+ \int_{\max(a,z_N)}^{\bar z_N} e^{-c'' N}\, dz \notag\\
&\le\;
a e^{-c a^{\alpha_x}}
+ \int_a^\infty e^{-c z^{\alpha_x}}\, dz
+ \int_{\max(a,z_N)}^{\bar z_N}
e^{-c'' (z/\Delta_{\max})^{\alpha_x}}\, dz,
\label{eq:ui-integral}
\end{align}
where $c'' = c x_0^{\alpha_x}>0$ (the sandwich at $t=x_0$), the
endpoint term $a\Prob(Y_N>a)$ is bounded by $a e^{-c a^{\alpha_x}}$ for
$a\le z_N$ by \eqref{eq:ui-tail} and by $a e^{-c''(a/\Delta_{\max})^{\alpha_x}}$
for $z_N < a \le \bar z_N$ (then $\Prob(Y_N>a)\le
\Prob(Y_N>z_N)\le e^{-c''N}$), vanishing for $a>\bar z_N$, and the last
inequality in \eqref{eq:ui-integral} uses $z\le\bar z_N$.  All three terms on the right of
\eqref{eq:ui-integral} tend to $0$ as $a\to\infty$ uniformly in $N$
(the middle one by integrability of $e^{-c z^{\alpha_x}}$ on
$[a,\infty)$, the third by the same exponential integrability after the
change of variables $z\mapsto z/\Delta_{\max}$); this is the definition
of uniform integrability.

For each fixed $j\ge 2$ the minimum tail \eqref{eq:ui-tail} does not
control the $j$-th order statistic (order statistics move only
upward), so the intermediate region $z\le z_N$ needs the exact tail
\begin{equation}
\Prob(D_{x,j:N} > t)
\;=\;
\sum_{k=0}^{j-1} \binom{N}{k} F(t)^k \bigl(1-F(t)\bigr)^{N-k},
\qquad F(t) = \Prob(\Delta_x(U)\le t),
\label{eq:orderstat-tail}
\end{equation}
$\{D_{x,j:N}>t\}$ being exactly the event that fewer than $j$
candidates satisfy $\Delta_x(U)\le t$.  For $z\in[1,z_N]$, so that
$t = zN^{-1/\alpha_x}\le t_0$, the two-sided sandwich gives
$c z^{\alpha_x}/N \le F(t) \le C z^{\alpha_x}/N$; each term of
\eqref{eq:orderstat-tail} with $k\ge 1$ then satisfies, for
$N\ge 2(j-1)$ and a constant $K_j$ depending only on $j$,
\begin{equation}
\binom{N}{k} F(t)^k \bigl(1-F(t)\bigr)^{N-k}
\;\le\;
\frac{(C z^{\alpha_x})^k}{k!}\, e^{-c z^{\alpha_x}/2}
\;\le\;
K_j\, z^{\alpha_x(j-1)} e^{-c z^{\alpha_x}/2},
\label{eq:orderstat-middle}
\end{equation}
because $\binom{N}{k}\le N^k/k!$, $(1-F)^{N-k}\le e^{-F(N-k)}$ and
$F(t)\ge c z^{\alpha_x}/N$ ($N-k\ge N/2$); the $k=0$ term is at most
$e^{-c z^{\alpha_x}}$, and summing over $k$ yields, for
$1\le z\le z_N$ and all large $N$,
\begin{equation}
\Prob\bigl(N^{1/\alpha_x}D_{x,j:N} > z\bigr)
\;\le\;
K_j\, \bigl(1 + z^{\alpha_x(j-1)}\bigr)\, e^{-c z^{\alpha_x}/2}.
\label{eq:orderstat-intermediate}
\end{equation}
For $z>z_N$ the event $D_{x,j:N} > zN^{-1/\alpha_x}\ge x_0$ forces
fewer than $j$ candidates below $x_0$: a Binomial($N,p_0$) tail with
$p_0 = \Prob(\Delta_x(U)\le x_0)>0$, at most
$e^{-c' N}\,\operatorname{poly}(N)$, uniformly in $z\le\bar z_N$, and
$D_{x,j:N}\le\Delta_{\max}$ as for the minimum.  The layer-cake
identity applied to $Y_N := N^{1/\alpha_x}D_{x,j:N}$ now splits the
integral exactly as above, with \eqref{eq:orderstat-intermediate}
replacing \eqref{eq:ui-tail} in the intermediate stretch and the
endpoint term bounded by the same integrals: all pieces tend to $0$ as
$a\to\infty$, uniformly in $N$ (the intermediate one by
$\int_1^\infty (1+z^{\alpha_x(j-1)}) e^{-c z^{\alpha_x}/2}\,dz<\infty$,
the tail one by
$N^{1/\alpha_x}\,\operatorname{poly}(N)\,e^{-c' N}\to 0$).
\end{proof}

\begin{lemmasupp}[Uniform integrability with bounded perturbations]
\label{lem:ui2}
Let $(\zeta_N)$ be uniformly integrable and let $(\xi_N)$ be a family
with $\sup_N|\xi_N - \zeta_N|\le R<\infty$ almost surely.  Then
$(\xi_N)$ is uniformly integrable.
\end{lemmasupp}

\begin{proof}
For $a>R$, on the event $\{|\xi_N|>a\}$,
$|\zeta_N|\ge |\xi_N| - R > a - R$, so
\begin{align}
\E\bigl[|\xi_N|\,\one\{|\xi_N|>a\}\bigr]
&\le\;
\E\bigl[|\zeta_N|\,\one\{|\zeta_N|>a-R\}\bigr]
+ R\,\Prob(|\xi_N|>a) \notag\\
&\le\;
\E\bigl[|\zeta_N|\,\one\{|\zeta_N|>a-R\}\bigr]
+ R\,\Prob(|\zeta_N|>a-R),
\label{eq:ui2}
\end{align}
and both terms tend to zero as $a\to\infty$ uniformly in $N$ by the
uniform integrability of $(\zeta_N)$ and Markov's inequality.
\end{proof}

In the proof of Lemma~5.2 of the main paper the bounded perturbation
is $\xi_N = m_N(h)$, $\zeta_N = m_N$, with $|m_N(h) - m_N|\le\bbeta R$;
Lemma~\ref{lem:ui2} supplies the uniform integrability to pass from
distributional convergence to convergence of expectations in (5.14).

\section{Tubular Coordinates and the Local Volume Law}
\label{app:tubular}

This section records the tubular-coordinate calculation underlying
Proposition~3.2 of the main paper and the marked volume identity
(4.1) of the main paper.

\subsection{Tubular Neighborhoods of Compact Submanifolds}
\label{app:tubular:setup}

Let $\calM\subseteq\R^d$ be a compact $C^2$ submanifold of dimension
$k$ without boundary.  By the tubular neighborhood theorem there exist
$\rho_0>0$ and a $C^1$ diffeomorphism
$\Theta:\{(z,v): z\in\calM,\ v\in N_z\calM,\ \norm{v}<\rho_0\}\longrightarrow U_{\rho_0}(\calM)$
onto the tubular neighborhood $U_{\rho_0}(\calM)=\{u:\dist(u,\calM)<\rho_0\}$, where $N_z\calM$ is the normal space.  The map is
$u = \Theta(z,v) = z + v + o(\norm{v})$; the Jacobian satisfies
$\det D\Theta(z,v) = 1 + O(\norm{v})$ uniformly in $(z,v)$, so the
Euclidean volume element is
$du = (1+O(\norm{v}))\, dv\, d\calH^k(z)$, where $\calH^k$ is the
$k$-dimensional Hausdorff measure on $\calM$; the radius $\rho_0$, the
implicit constants, and the bounds below are uniform over $z\in\calM$
by compactness.

\subsection{Change of Variables}
\label{app:tubular:change}

Fix a stratum $\calM_{xj}$ with codimension $m = d - k$.  By the
separation in Assumption~3.1 the tubular neighborhoods
$U_{\rho_0}(\calM_{x\ell})$ are pairwise disjoint, and since
$\Delta_x$ is positive and continuous off the strata, $t_0$ is chosen
small enough that the level set $\{\Delta_x\le t_0\}$ is
contained in their union.  For $t\in(0,t_0)$, restricted to
$U_{\rho_0}(\calM_{xj})$, write
$u = \Theta_{xj}(z,v)$, $v = \rho\theta$ with
$\rho = \norm{v}\in[0,\rho_0)$ and $\theta$ a unit normal vector.  The
deficiency satisfies, by (3.4) of the main paper,
$\Delta_x(\Theta_{xj}(z,\rho\theta)) = \rho^{r_{xj}} q_{xj,z}(\theta) + o(\rho^{r_{xj}})$
uniformly in $(z,\theta)$, and the density satisfies
(3.5) of the main paper.  The rescaled change of variables
$w = t^{-1/r_{xj}} v$, $v = t^{1/r_{xj}} w$,
maps the tube $\{\norm{v}<\rho_0\}$ onto
$\{\norm{w}<t^{-1/r_{xj}}\rho_0\}$, which contains the ball
$\{\norm{w}\le c^{-1/r_{xj}}\}$ of (3.13) of the main paper for all $t$ small.
Under this rescaling, for any Borel $B\subseteq\calM_{xj}$ whose
boundary has zero $\sigma_{xj}$-measure,
\begin{align}
& t^{-\alpha_{xj}}\,
\mu_x\Bigl\{u\in\Theta_{xj}(\text{tube}): \Delta_x(u)\le t,\;
\pi_x(u)\in B\Bigr\} \notag\\
&\qquad \longrightarrow
\int_B \int_{N_z\calM_{xj}}
\one\bigl\{q_{xj,z}(w)\le 1\bigr\}
\, g_{xj}\!\left(z,\tfrac{w}{\norm{w}}\right)
\, \norm{w}^{\eta_{xj}}\, dw\, d\calH^k(z),
\label{eq:local-volume-A}
\end{align}
which is the stratumwise form of the marked identity (4.1); taking
$B=\calM_{xj}$ recovers (3.7).  The proof is identical to that of
Proposition~3.2, the limit
justified by dominated convergence (uniformity of the remainders;
finiteness of the limiting integral, checked there).  For the fixed
finite collection of continuity sets used below, the convergence in
\eqref{eq:local-volume-A} holds simultaneously.

\subsection{Localized Uniformity}
\label{app:tubular:uniformity}

For the pointwise convergence of the Laplace functionals
(4.6) of the main paper, the following form of dominated convergence
is used.
Fix $L<\infty$ and partition $[0,L]$ into intervals $I_k$ and
$\mathsf{A}_x^\star$ into Borel pieces $B_\ell$ whose
$\sigma_x$-boundaries vanish.  Applying \eqref{eq:local-volume-A} on
each cell gives, for every such partition,
\begin{equation}
\sup_{k,\ell}\;
\left| N\mu_x\bigl\{u : N^{1/\alpha_x}\Delta_x(u)\in I_k,\;
\pi_x(u)\in B_\ell\bigr\} - \Lambda_x(I_k\times B_\ell) \right|
\;\longrightarrow\; 0,
\qquad N\to\infty.
\label{eq:uniformity-A}
\end{equation}
For $f$ continuous with compact support in
$[0,L]\times\mathsf{A}_x^\star$, the integrand $1-e^{-f}$ is uniformly
continuous, so on a sufficiently fine partition it is constant on each
cell up to $\varepsilon$; replacing it by its cellwise values and
applying \eqref{eq:uniformity-A} cell by cell, the Riemann sums
$\sum_{k,\ell}(1-e^{-f})\Lambda_x(I_k\times B_\ell)$ converge to
$\int(1-e^{-f})\,d\Lambda_x$ as the mesh tends to zero, the dominated
convergence passage in (4.6) of the main paper.

%% file: tables_py/tab_ex1.tex
\begin{tabular}{lccccc}
\toprule
configuration & $\widehat p_N(10^4)$ & $p_\infty$ & $\widehat\alpha$ & $\alpha = 1/r$ \\
\midrule
equal exponents     & $0.6672 \pm 0.0033$ & $2/3$ & $0.507 \pm 0.008$ & $1/2$ \\
different exponents & $0.9999 \pm 0.0001$ & $1$ & $0.261 \pm 0.010$ & $1/4$ \\
\bottomrule
\end{tabular}

%% file: tables_py/tab_ex2.tex
\begin{tabular}{lcl}
\toprule
quantity & value & \\
\midrule
static volume weight $C_2/(C_1+C_2)$
  & $1/2$ & \\
finite-$N$ frequency $\widehat p_N$ at $10^4$ (se)
  & $0.7566$ & $(0.0030)$ \\
dynamic limit $p_{1,1}$ (Theorem 6.5) & $0.753$ & \\
fixed point $h^\star$ & $(0.4716, 0.100)$ & \\
empirical leading gap $N\,\mathrm{gap}_1$ at $10^4$
  & $0.4677$ & \\
resolvent, ``exit to state 2'' policy (Corollary~5.5)
  & $0.595$ & \\
resolvent, ``self-loop'' policy (Corollary~5.5)
  & $10.0$ & \\
\bottomrule
\end{tabular}

%% file: tables_py/tab_ex2_grid.tex
\begin{tabular}{lcccc}
\toprule
$N$ & frequency $\widehat p_N$ & (se) & $N\,\mathrm{gap}_1$ & $N\,\mathrm{gap}_2$ \\
\midrule
10 & 0.7431 & (0.0031) & 0.4321 & 0.0913 \\
30 & 0.7487 & (0.0031) & 0.4578 & 0.0979 \\
100 & 0.7483 & (0.0031) & 0.4673 & 0.0990 \\
300 & 0.7609 & (0.0030) & 0.4701 & 0.0990 \\
1000 & 0.7501 & (0.0031) & 0.4707 & 0.0994 \\
3000 & 0.7550 & (0.0030) & 0.4697 & 0.0983 \\
10000 & 0.7566 & (0.0030) & 0.4677 & 0.1001 \\
\bottomrule
\end{tabular}

%% file: tables_py/tab_ex3_finiten.tex
\begin{tabular}{lcccc}
\toprule
& $N=10^2$ & $N=10^3$ & $N=10^4$ & $h_\xi(1)$ \\
\midrule
$\xi = -4$ & $4.3686$ $(0.0008)$ & $4.3725$ $(0.0008)$ & $4.3717$ $(0.0008)$ & 4.3730 \\
$\xi = -2$ & $2.4276$ $(0.0008)$ & $2.4312$ $(0.0008)$ & $2.4324$ $(0.0009)$ & 2.4327 \\
$\xi = +0$ & $0.4677$ $(0.0009)$ & $0.4707$ $(0.0009)$ & $0.4701$ $(0.0009)$ & 0.4716 \\
$\xi = +1$ & $0.5765$ $(0.0014)$ & $0.5811$ $(0.0014)$ & $0.5837$ $(0.0014)$ & 0.5814 \\
$\xi = +2$ & $0.5886$ $(0.0015)$ & $0.5933$ $(0.0015)$ & $0.5916$ $(0.0015)$ & 0.5932 \\
$\xi = +4$ & $0.5876$ $(0.0015)$ & $0.5947$ $(0.0016)$ & $0.5931$ $(0.0016)$ & 0.5950 \\
\midrule
$\max_\xi |g_N^\xi(1) - h_\xi(1)|$ & $0.0073$ & $0.0015$ & $0.0022$ & \\
\bottomrule
\end{tabular}